\documentclass[11pt,reqno]{amsart}
\usepackage{amsmath,amssymb,amsfonts,amsthm,enumerate,color,ifthen,hyperref}
\usepackage{xargs}
\usepackage[textwidth=4cm, textsize=footnotesize]{todonotes}
\usepackage[square,numbers]{natbib}

\usepackage[caption = false]{subfig}
\usepackage{graphicx}
\usepackage{floatrow}
\newfloatcommand{capbtabbox}{table}[][\FBwidth]

\usepackage{aliascnt,bbm}
\def\rset{\mathbb R}
\def\zset{\mathbb Z}

\def\eqsp{\;}

\newcommand{\pscal}[2]{\left\langle#1,#2\right\rangle}

\newcommand{\eqdef}{\ensuremath{\stackrel{\mathrm{def}}{=}}}

\def\Xset{\mathcal{X}} 
\def\Yset{\mathcal{Y}} 
\def\F{\mathcal{F}} 
\def\cB{\mathsf{B}} 
\def\e{\mathcal{E}}

\def\M{\mathcal{M}}

\def\D{\mathcal{D}}
\def\A{\mathcal{A}}

\newcommandx\sequence[3][2=t,3=\zset]
{\ifthenelse{\equal{#3}{}}{\ensuremath{\{ #1_{#2}\}}}{\ensuremath{\{ #1_{#2}, \eqsp #2 \in #3 \}}}}

\def\PP{\mathbb{P}} 
\newcommand{\CPP}[3][]
{\ifthenelse{\equal{#1}{}}{{\mathbb P}\left(\left. #2 \, \right| #3 \right)}{{\mathbb P}_{#1}\left(\left. #2 \, \right | #3 \right)}}
\def\PE{\mathbb{E}} 
\newcommand{\CPE}[3][]
{\ifthenelse{\equal{#1}{}}{{\mathbb E}\left[\left. #2 \, \right| #3 \right]}{{\mathbb E}_{#1}\left[\left. #2 \, \right | #3 \right]}}

\def\tv{\mathrm{tv}}

\def\Cset{\mathcal{C}} 

\DeclareMathOperator{\ReLU}{ReLU}

\usepackage{bm}

\tikzset{fc/.style={black,draw=black,rectangle,minimum height=1cm}}
\tikzset{conv/.style={black,draw=black,rectangle,minimum height=1cm}}
\tikzset{pool/.style={black,draw=black,rectangle,minimum height=1cm}}
\theoremstyle{plain}
\newtheorem{theorem}{Theorem}
\newtheorem{assumption}{H\hspace{-3pt}}

\newaliascnt{proposition}{theorem}
\newtheorem{proposition}[proposition]{Proposition}
\aliascntresetthe{proposition}

\newaliascnt{lemma}{theorem}
\newtheorem{lemma}[lemma]{Lemma}
\aliascntresetthe{lemma}
\newaliascnt{corollary}{theorem}

\aliascntresetthe{corollary}

\theoremstyle{definition}
\newaliascnt{definition}{theorem}

\aliascntresetthe{definition}

\newtheorem{algorithm}{Algorithm}

\newaliascnt{remark}{theorem}
\newtheorem{remark}[remark]{Remark}
\aliascntresetthe{remark}

\newaliascnt{example}{theorem}

\aliascntresetthe{example}

\def\rmd{\mathrm{d}}

\def\1{\mathbbm{1}}

\def\argmax{\operatorname{Argmax}}

\begin{document}

\title[Asynchronous MCMC sampling  for sparse Bayesian inference]{A fast asynchronous MCMC sampler for sparse Bayesian inference}
\thanks{This work is partially supported by the NSF grant DMS 1513040}

\author{Yves Atchad\'e}\thanks{ Y. Atchad\'e: Boston University, 111 Cummington Mall, Boston 02215 MA, United States. {\em E-mail address:} yvesa@umich.edu}
\author{Liwei Wang}\thanks{ L. Wang: Boston University, 111 Cummington Mall, Boston 02215 MA, United States. {\em E-mail address:} wlwfoo@bu.edu}
  
\subjclass[2010]{62F15, 62Jxx}

\keywords{Sparse Bayesian inference, Asynchronous MCMC sampling, MCMC mixing, Bayesian deep learning}

\maketitle

\begin{center} (May 2021) \end{center}

\begin{abstract}
We propose a very fast approximate Markov Chain Monte Carlo (MCMC) sampling framework that is applicable to a large class of sparse Bayesian inference problems, where the computational cost per iteration in several models is of order $O(ns)$, where $n$ is the sample size, and $s$ the underlying sparsity of the model. This cost can be further reduced by data sub-sampling when stochastic gradient Langevin dynamics are employed. The  algorithm is an extension of the asynchronous Gibbs sampler of \cite{johnson:etal:13}, but can be viewed from a statistical perspective as a form of Bayesian iterated sure independent screening (\cite{fan:etal:09}).  We show that in high-dimensional linear regression problems, the Markov chain generated by the proposed algorithm admits an invariant distribution that recovers correctly the main signal with high probability under some statistical assumptions.  Furthermore we show that its mixing time is at most linear in the number of regressors. We illustrate the  algorithm with several   models.
\end{abstract}

\section{Introduction}\label{sec:intro}

 There is a rich  and extensive literature on high-dimensional sparse Bayesian inference built mainly around shrinkage priors and spike and slab priors (see e.g. \cite{mitchell:beauchamp88,george:mcculloch97,castillo:etal:12,castillo:etal:14,atchade:15b,carvalho:etal:10,piironen:etal:17,biswas:etal:2021} and the references therein). However the computational cost per iteration for sampling from the resulting posterior distributions grows at least as $O(n^2p)$ in Gaussian linear regression models with $n$ data samples and $p$ regressors ($n\geq p$), and becomes quickly prohibitive, particularly  in  non-Gaussian models.   Indeed, computing high-dimensional integrals remains the main challenge in the practical implementation of Bayesian inference. The problem has grown much worse  over the last decade or so with the rise of deep neural networks and other highly over-parameterized models (\cite{bhadra:etal:20}). 
 
 As a step forward, we propose herein a very fast but approximate MCMC scheme for sparse Bayesian models with spike and slab priors. 
The algorithm builds on a version of the spike and slab prior  developed in \cite{AB:19}, and closely related to the pseudo-prior device of \cite{carlin:chib:95}. The proposed prior introduces a variable $\delta\in\Delta\eqdef\{0,1\}^p$ (called sparsity structure), with prior distribution $\{\pi(\delta),\;\delta\in\Delta\}$ of the form  
\begin{equation}\label{prior:1}
\pi(\delta) \propto p^{-\textsf{u}\|\delta\|_0},\end{equation}
for some user-defined parameter $\textsf{u}>1$. Given $\delta$ the components of $\theta$ are assumed to be conditionally independent mean-zero Gaussian random variables, with variance $\rho_0^{-1}$ (resp. $\rho_1^{-1}$) if the corresponding component of $\delta$ is $0$ (resp. $1$), for user-defined parameters $\rho_0$ and $\rho_1$.  The limiting case $\rho_0^{-1}=0$ corresponds to the well-known  spike and slab prior with a point mass at $0$ (\cite{mitchell:beauchamp88}). This type of spike and slab priors goes back at least to (\cite{george:mcculloch97}). Given $(\delta,\theta)$ the conditional distribution of the data is then assumed to be $\D\sim f_{\theta_\delta}$, where $\theta_\delta\eqdef\theta\cdot\delta$ is the component-wise product of $\theta,\delta$, and $f_\theta(\cdot)$ is a density on the data sample space. Setting $\ell(\theta) \eqdef \log f_\theta(\D)$, the resulting  posterior distribution of $(\delta,\theta)$ given $\D$ has density on $\Delta\times\rset^p$ given by
\begin{equation}\label{post:Pi}
\Pi(\delta,\theta\vert \D) \propto \left(\frac{1}{p^\textsf{u}}\sqrt{\frac{\rho_1}{\rho_0}}\right)^{\|\delta\|_0} \exp\left( -\frac{\rho_0}{2}\|\theta -  \theta_{\delta}\|_2^2   -\frac{\rho_1}{2}\|\theta_\delta\|_2^2 + \ell(\theta_\delta)\right).\end{equation}
Note that, because the two alternative prior densities $\textbf{N}(0,\rho_0^{-1})$ and $\textbf{N}(0,\rho_1^{-1})$ have densities with respect to the Lebesgue measure, (\ref{post:Pi}) does not possess the well-known trans-dimensionality issue that pose problems with spike and slab priors with points mass at the origin. Furhermore, because the components of $\theta$ are independent under the prior and the likelihood function depends on $\theta$ through $\theta_\delta$, the marginal posterior distribution of $(\delta,\theta_\delta)$ under (\ref{post:Pi}) does not depend on $\rho_0$, and in particular is the same as with  the corresponding spike and slab prior with point mass at the origin. Hence (\ref{post:Pi}) incurs no loss of information in the estimation of $(\delta,\theta_\delta)$ compared with spike and slab priors with point mass at the origin (we refer the reader to \cite{AB:19} for more details).


\subsection{Main contributions}
We propose a fast MCMC method to sample approximately from (\ref{post:Pi}) where  the  computational cost per iteration in generalized linear models is of order $O(ns)$, where $n$ is the sample size, and $s$ the underlying sparsity of the model. This cost can be further reduced by sub-sampling when stochastic gradient Langevin dynamics (\cite{sgld}) is employed.  Furthermore, we show that for  linear regression models the mixing time of the algorithm is $O(p)$, provided a large enough sample size is available. The algorithm can be viewed as a form of Bayesian sure independent screening  (\cite{fan:lv:08,fan:etal:09}) in the sense that, as in sure independent screening, the algorithm alternates between a fast component-wise variable screening step where the components of $\delta$ are sampled independently (conditionally on $\theta$), and a sparse model refit step where the parameter $\theta$ is re-estimated.  From the MCMC viewpoint, the proposed algorithm is an extension of the asynchronous Gibbs sampler (\cite{smola:nar:10,johnson:etal:13,desa:16}) where several variables are updated asynchronously and in parallel.

We test the algorithm empirically on linear and logistic regression models, and with a deep neural network model (lenet-5 applied to the MNIST-FASHION (\cite{xiao2017fashion}) dataset). The  application to logistic regression models  show that the algorithm is an order of magnitude faster than the mean-field variational approximation of (\ref{post:Pi}). In deep neural network models the proposed algorithm combined with stochastic gradient Langevin dynamics can be easily implemented by modifying existing stochastic gradient descent implementation. 

\subsection{Related work}
Sparse model estimation has been a major theme in statistics over the last two decades (\cite{buhlGeer11,hastie:etal:15,wainwright:19}), driven in large part by biomedical and engineering applications. In deep learning, despite the double descent paradox and the common practice of fitting highly  overparametrized models, there is also a growing interest in sparse modeling (\cite{wen:etal:16,louizos:etal:18,bellec:etal:18,gale:etal:19,lee:etal:19,dettmers:etal:19,frankle:etal:19}).  

Most existing Bayesian implementation of large scale sparse models  typically relies on variational approximation of the posterior distribution, and tremendous progress has been made on the topic over the last few years (\cite{louizos:etal:17,ghosh:etal:19,wang:etal:19,tran:etal:20,abdellatif:20,yang:etal:20,zhang:gao:20}).  However variational approximation is a general methodology, not a specific algorithm. And the extra step of building a good variational approximation family for a given problem --  an important requirement for the consistency of the method (\cite{zhang:gao:20}) --  is not a one-size-fits-all process. Furthermore, fitting variational approximation families that are not conjugate is generally a costly nonconvex problem. For instance, we observed on logistic regression models that the mean field variational approximation of (\ref{post:Pi}) is computationally more expensive than the algorithm proposed in this work.

\subsection{Outline} 
The paper is organized as follows. We end the introduction with a compendium of our notations. The asynchronous sampler is developed in Section \ref{sec:algo}. Some  theoretical insights are  provided in Section \ref{sec:theory}, but to keep the focus on the main ideas we placed the proofs in the appendix. The numerical illustrations are collected in Section \ref{sec:numerics}. The paper ends with some concluding remarks in Section \ref{sec:conclusion}. \textsf{MATLAB} code for the logistic regression and deep  neural network examples are available at \textsf{https://github.com/odrinaryliwei/S-SGLD}.

\subsection{Notations}\label{sec:notations}
We introduce here some Markov chain notations that are used below, largely following  \cite{meyn:tweedie:08}.  A Markov kernel $P$ on some measurable space $(\mathbb{T},\mathcal{B})$  acts both on bounded measurable
functions $f$ on $\mathbb{T}$ and on $\sigma$-finite measures $\mu$ on
$(\mathbb{T},\mathcal{B})$ via $Pf(\cdot) \eqdef \int
P(\cdot,\rmd y) f(y)$ and $\mu P(\cdot) \eqdef \int \mu(\rmd x) P(x, \cdot)$ respectively. If $W: \mathbb{T}\to [1, +\infty)$ is a function, the $W$-norm of a function
$f: \mathbb{T}\to \rset$ is defined as $|f|_W \eqdef \sup_{\mathbb{T}} |f|
/W$.  When $W=1$, this is the supremum norm.  If $\mu$ is a signed measure on $(\mathbb{T},\mathcal{B})$, the total variation norm $\| \mu
\|_{\tv}$ is defined as $\| \mu \|_{\tv} \eqdef \frac{1}{2}\sup_{\{f, |f|_1 \leq 1 \}} | \mu(f)|$,  and the $W$-norm of $\mu$  is
defined as $\| \mu \|_{W} \eqdef \frac{1}{2}\sup_{\{g, |g|_W \leq 1 \}} |\mu(g)|$, where $\mu(f)$  denotes the integral $\int_{\mathbb{T}}f(x)\mu(\rmd x)$. Given two Markov kernels $P,Q$ on $(\mathbb{T},\mathcal{B})$, their product is the Markov kernel defined as $PQ(x,A) \eqdef \int
P(x,\rmd y ) Q(y,A)$.  In particular we define $P^n$, the $n$-th power of $P$, as $P^0(x,A) \eqdef \delta_x(A)$   and $P^{n+1} = PP^n$, $n \geq 0$, where $\delta_x(\rmd  t)$ stands for the Dirac mass at $x$.  Note that  $\mu(PQ) = (\mu P)Q$, and $(\mu P)(f) = \mu(Pf)$.

We also collect here  our notations on sparse models.  Throughout our parameter space is $\rset^p$ equipped with its Euclidean norm $\|\cdot\|_2$ and inner product $\pscal{\cdot}{\cdot}$.  We also use $\|\cdot\|_0$ which counts the number of non-zero elements, and $\|\cdot\|_\infty$ which returns the largest absolute value. We set $\Delta\eqdef\{0,1\}^p$. Elements of $\Delta$ are called sparsity structures. For $\delta,\delta'\in\Delta$, we write $\delta\subseteq\delta'$ if $\delta_j\leq \delta_j'$ for all $1\leq j\leq p$, and we write $\delta\supseteq\delta'$ if $\delta'\subseteq\delta$. Given $\delta\in\Delta$, and $\theta\in\rset^p$, we write $\theta_\delta$ to denote the  component-wise product of $\theta$ and $\delta$, and $\delta^c\eqdef 1-\delta$. We will also write $[\theta]_\delta=(\theta_j,\;j\in\{1\leq k\leq p:\;\delta_k=1\})$ which collect the components of $\theta$ with corresponding components of $\delta$ equal to $1$. Conversely, assuming $\|\delta\|_0>0$, and for $u\in\rset^{\|\delta\|_0}$, we define $(u,0)_\delta$ as the element of $\rset^p$ such that $[(u,0)_\delta]_\delta = u$.


\section{The asynchronous sampler}\label{sec:algo}
Probability distributions of the form (\ref{post:Pi}) are commonly handled using Metropolis-Hastings within Gibbs (\cite{robertetcasella04,handbook:11}). As a start we  follow the same approach, and derive an asymptotically exact algorithm that alternates between an update of $\delta$ given $\theta$, and an update of $\theta$  given $\delta$. To  update $\theta$ given $\delta$,  we  utilize the fact that given $\delta$ the selected components (denoted $[\theta]_\delta$) and the non-selected components (denoted $[\theta]_{\delta^c}$) of $\theta$ are conditionally independent, and  that the components of $[\theta]_{\delta^c}$ are i.i.d. $\textbf{N}(0,\rho_0^{-1})$. Assuming $\|\delta\|_0>0$, it is clear from (\ref{post:Pi}) that the conditional distribution of  $[\theta]_\delta$ given $\delta$ has density  on $\rset^{\|\delta\|_0}$ proportional to 
\begin{equation}\label{cond:dist:theta}
u\mapsto \exp\left(-\frac{\rho_1}{2}\|u\|_2^2 + \ell((u,0)_\delta)\right).\;\;\;\end{equation}
We then naturally update $[\theta]_\delta$ using a Markov kernel on $\rset^{\|\delta\|}$ with invariant distribution proportional to (\ref{cond:dist:theta}) that we denote $P_\delta$. Any convenient  MCMC algorithm can be used here (Random Walk Metropolis, Metropolis adjusted Langevin, Hamiltonian Monte Carlo, or others), and one can leverage the sparsity of $\delta$ for a  fast computation of $\ell((u,0)_\delta)$.

We use a Gibbs sampler to update $\delta$ given $\theta$. It comes out  from (\ref{post:Pi}) that  the conditional distribution of $\delta_j$ given $\theta,\delta_{-j}$ is the Bernoulli distribution $\textbf{Ber}(q_j(\delta,\theta))$, with probability of success given by
\begin{equation}\label{cond:dist:eq:1}
q_j(\delta,\theta) \eqdef \left(1 +\exp\left (\mathsf{a} + \frac{1}{2}(\rho_1 -\rho_0)\theta_j^2+ \ell(\theta_{\delta^{(j,0)}}) - \ell(\theta_{\delta^{(j,1)}})\right)\right)^{-1},\end{equation}
where $\mathsf{a} \eqdef \textsf{u}\log(p) + \frac{1}{2}\log(\rho_0/\rho_1) $, and where $\delta^{(j,0)}$ (resp $\delta^{(j,1)}$) is the same as $\delta$ except possibly at component $j$ where $\delta^{(j,0)}_j=0$ (resp. $\delta^{(j,1)}_j=1$). Naturally,  $q_j(\delta,\theta)$ does not depend on $\delta_j$. We update $J$ randomly selected components of $\delta$ at each iteration. Put together these two steps form an asymptotically exact MCMC algorithm to sample from (\ref{post:Pi}) that is our ideal sampler.
 \medskip

\begin{algorithm}\label{algo:1}[Asymptotically Exact Sampler]$\hrulefill$\\
Draw $(\delta^{(0)},\theta^{(0)})$ from some initial distribution, and repeat the following steps for $k=0,\ldots$. Given $(\delta^{(k)},\theta^{(k)})= (\delta,\theta)\in\Delta\times\rset^p$:
\begin{description}
\item [(STEP 1)] For all $j$ such that $\delta_j=0$, draw independently $\bar\theta_j\sim  \textbf{N}(0,\rho_0^{-1})$. Provided that $\|\delta\|_0>0$, draw $[\bar\theta]_\delta\sim P_\delta([\theta]_\delta,\cdot)$,  where  $P_\delta$ is a Markov kernel on $\rset^{\|\delta\|_0}$ with invariant density proportional to (\ref{cond:dist:theta}).
\item [(STEP 2)] Set $\bar\delta = \delta$. Randomly select a subset $\mathsf{J}\subset\{1,\ldots,p\}$ of size $J$.
\begin{itemize}
\item For each $j\in\mathsf{J}$: draw $V_j\sim \textbf{Ber}(q_j(\bar\delta,\bar\theta))$, where $q_j$ is as in (\ref{cond:dist:eq:1}), and set $\bar \delta_j=V_j$.
\end{itemize}
\item [(STEP 3)]  Set $(\delta^{(k+1)},\theta^{(k+1)})= (\bar \delta, \bar \theta)$.
\end{description}
\vspace{-0.6cm}
\end{algorithm}
$\hrulefill$
\medskip

 Assuming that the cost of computing the likelihood function scales like $O(n\|\delta\|_0)$, and ignoring the cost of generating univariate Gaussian random variables,  the computation cost of the $k$-th iteration of Algorithm \ref{algo:1} is of order $O(nJ\|\delta^{(k)}\|_0)$. Clearly that cost  increases with $J$, but  the  mixing time of the resulting algorithm decreases with $J$. We are not aware of any sensible way of selecting $J$ that balances these two costs. For easy subsequent comparisons we will follow the guideline that we set below for selecting $J$ in the asynchronous algorithm.

\subsection{Asynchronous approximation}\label{sec:algo:2}
Algorithm \ref{algo:1} becomes slow in problems where $n$ is large and there is no efficient way of computing the log-likelihood differences $\ell(\theta_{\delta^{(j,0)}}) - \ell(\theta_{\delta^{(j,1)}})$ in (\ref{cond:dist:eq:1}). We propose to speed up this step of the algorithm by replacing the log-likelihood difference by an approximation.
To gain some intuition, consider a linear regression problem where $\ell(\theta) = -\frac{1}{2}\|y-X\theta\|_2^2$, $y\in\rset^n$, $X\in\rset^{n\times p}$, with columns normalized to $\|X_j\|_2=\sqrt{n}$. In that case the second order Taylor approximation of $\ell$ is exact, and writes
\[\ell(\theta_{\delta^{(j,0)}}) - \ell(\theta_{\delta^{(j,1)}}) = -\theta_j\pscal{X_j}{y-X\theta_{\delta^{(j,0)}}} + \frac{\theta_j^2n}{2}.\]
This expression shows that the first derivative term and the constant $\mathsf{a}$ in (\ref{cond:dist:eq:1}) are the main determining terms. To see this, suppose that $j$ is a relevant variable with true regression coefficient $\beta_{\star j}$, say. Suppose also that $j$ is currently not selected ($\delta_j=0$). In that case, the corresponding regression parameter $\theta_j$ is of  order $1/\sqrt{\rho_0}$, since it is drawn from $\textbf{N}(0,\rho_0^{-1})$. Therefore, and  assuming that the regression errors are sub-Gaussian, it can be shown that with $\rho_0 = n$ (as we advocate below),
\[-\theta_j\pscal{X_j}{y-X\theta_{\delta^{(j,0)}}} \approx -\theta_j\theta_{\star j}n = -\left(\theta_j\sqrt{\rho_0}\right) \frac{n}{\sqrt{\rho_0}} \theta_{\star j}\approx \pm |\theta_{\star j}|\sqrt{n},\] 
whereas the  second order term is $O(n/\rho_0)=O(1)$. Hence, provided that the current estimate $\theta_j$ has the correct sign (which happens with probability $1/2$), the first derivative dominates and $\ell(\theta_{\delta^{(j,0)}}) - \ell(\theta_{\delta^{(j,1)}})$ is negatively large, and $q_j(\delta,\theta)$ is close to $1$. Note however that if the variable  $j$ is irrelevant, then the first order term is $O(\sqrt{\log(p)})$, whereas the second order term remains $O(1)$. In that case  the term $\mathsf{a} \eqdef \textsf{u}\log(p) + \frac{1}{2}\log(\rho_0/\rho_1) $ from the prior dominates and  $q_j(\delta,\theta)$ is close to $0$. 

The discussion so far is essentially the idea of residual correlation well-known in variable selection and sure screening: we fit a model without a variable $X_j$, say, and we consider adding $X_j$ to the model if its correlation with the residual is high. 
We extend this idea to the general model as follow by using the  approximation:
\[\ell(\theta_{\delta^{(j,0)}}) - \ell(\theta_{\delta^{(j,1)}}) \approx -\theta_j\frac{\partial \ell}{\partial\theta_j} (\theta_{\delta^{(j,0)}}) -\frac{\theta_j^2}{2}\left(\frac{\partial \ell}{\partial\theta_j} (\theta_{\delta^{(j,0)}})\right)^2,\]
where we revert the sign of the quadratic term for increased sensitivity.  This leads to the following approximation of $q_j(\delta,\theta)$ in (\ref{cond:dist:eq:1}) :
\begin{multline}\label{cond:dist:eq:2}
\tilde q_j(\delta,\theta) \eqdef \left(1 +\exp\left (\mathsf{a} +\frac{1}{2}(\rho_1 -\rho_0)\theta_j^2 - \theta_j G_j(\theta_{\delta})  -\frac{ \theta_j^2}{2} G_j(\theta_{\delta})^2 \right)\right)^{-1},\\
\;\;\;\mbox{ where }\;\; G_j(\cdot) \eqdef \frac{\partial\ell}{\partial \theta_j}(\cdot).\end{multline}

Suppose now that we randomly select a subset $\mathsf{J}$ of $\{1,\ldots,p\}$ as in (STEP 2) of Algorithm \ref{algo:1}, and we need to approximate the $J$ terms $q_j(\delta,\bar\theta)$. In keeping with the idea of residual correlation explained above, we approximate $q_j(\delta,\bar\theta)$  by $\tilde q_j(\vartheta,\bar\theta)$, where $\vartheta = \vartheta(\mathsf{J},\delta) \in\{0,1\}^p$ with $k$-th component defined as
\begin{equation}\label{delta:J}
\vartheta_k = \left\{\begin{array}{ll}\delta_k & \mbox{ if }\;\; k\notin\mathsf{J}\\ 0 & \mbox{ otherwise }. \end{array}\right. 
\end{equation}

For additional flexibility we allow  to use an approximation $\tilde P_\delta$ of the kernel $P_\delta$ in (STEP 2) of Algorithm \ref{algo:1}, and we do not require $\tilde P_\delta$ to have (\ref{cond:dist:theta}) as invariant distribution.  The full algorithm is summarized in Algorithm \ref{algo:2}. 

Note that, because $\vartheta$ does not depend on  $(\delta_j)_{j\in\mathsf{J}}$, the variables $(V_j)_{j\in\mathsf{J}}$ in (STEP 2) of Algorithm \ref{algo:2} are now conditional independent Bernoulli random variables, and can be sampled  in parallel (instead of sequentially as in Algorithm \ref{algo:1}).  Furthermore, the computation of the $J$ probabilities $\tilde q_j(\vartheta,\bar\theta)$ requires the calculation of only one gradient $G(\theta_\vartheta)$.

 \medskip

\begin{algorithm}\label{algo:2}[Asynchronous sampler]$\hrulefill$\\
Draw $(\delta^{(0)},\theta^{(0)})$ from some initial distribution, and repeat the following steps for $k=0,\ldots$.  Given $(\delta^{(k)},\theta^{(k)})= (\delta,\theta)\in\Delta\times\rset^p$:
\begin{description}
\item [(STEP 1)] For all $j$ such that $\delta_j=0$, draw independently $\bar\theta_j\sim  \textbf{N}(0,\rho_0^{-1})$. Provided that $\|\delta\|_0>0$, draw $[\bar\theta]_\delta\sim \tilde P_\delta([\theta]_\delta,\cdot)$,  where  $\tilde P_\delta$ is a Markov kernel on $\rset^{\|\delta\|}$.
\item [(STEP 2)] Set $\bar\delta = \delta$. Randomly select a subset $\mathsf{J}\subset\{1,\ldots,p\}$ of size $J$. Form the vector $\vartheta = \vartheta(\mathsf{J},\delta)$ as in (\ref{delta:J}).
\begin{itemize}
\item For each $j\in\mathsf{J}$, draw independently $V_j\sim \textbf{Ber}(\tilde q_j(\vartheta,\bar\theta))$, where $\tilde q_j(\vartheta,\bar\theta)$ is as in (\ref{cond:dist:eq:2}), and set $\bar \delta_j = V_j$.
\end{itemize}
\item [(STEP 3)]  Set $(\delta^{(k+1)},\theta^{(k+1)})= (\bar \delta, \bar \theta)$.
\end{description}
\vspace{-0.6cm}
\end{algorithm}
$\hrulefill$
\medskip

Assuming again a likelihood function cost of $O(n\|\delta\|_0)$, and ignoring the cost of generating univariate Gaussian random variables,  the computation cost of the $k$-th iteration of Algorithm \ref{algo:2} is now of order $O(n\|\delta^{(k)}\|_0 + J)$, which can be substantially better than $O(nJ\|\delta^{(k)}\|_0)$  achieved by Algorithm \ref{algo:1}, depending on  $J$. Here it is clearly advantageous to take $J$ large, as close to $n\|\delta^{(k)}\|_0$  as possible. However there is a third new dimension to consider here: as $J$ increases, the limiting distribution of Algorithm \ref{algo:2} (assuming it exists) diverges further away from $\Pi(\cdot\vert\D)$, due to the accumulation of errors in the asynchronous sampling. We argue below that a sensible choice is setting $J=100$, or $J = \alpha n$, for some small fraction $\alpha$ ($\alpha\in [0.01,0.1]$).

Algorithm \ref{algo:2} has several interesting connections. From a statistical perspective the algorithm can be viewed as a Bayesian analog of the  iterated sure independent screening (ISIS) of (\cite{fan:lv:08,fan:etal:09}). Sure independent screening is a statistical inference algorithm that alternates between a fast component-wise variable screening step (based on marginal correlation thresholding, or marginal maximum likelihood estimate thresholding), and a model refit step.  Algorithm \ref{algo:2} has the same structure: (STEP 2) corresponds to the variable screening step -- which boils down to residual correlation in the linear regression case -- followed by a refit step based on MCMC draws. Unlike SIS that relies on hard-thresholding, the variable screening step of Algorithm \ref{algo:2} uses the prior distribution to  control sparsity. 

Viewed through the lense of MCMC methods, Algorithm \ref{algo:2} appears as an approximate version of Algorithm \ref{algo:1} where the update of $\delta$  in (STEP 2) is replaced by (an inexact form of) the asynchronous Gibbs sampler aka Hogwild! (\cite{smola:nar:10,johnson:etal:13}), recently analyzed by \cite{desa:16,daskalakis:18}. Indeed, we note that in (STEP 2) of algorithm \ref{algo:2} the newly simulated variable $(\bar\delta_{j_i})_i$ are conditionally independent and the update of $\bar\delta_{j_i}$ is based on $\vartheta$, a variable that remains unchanged through the sweep.

\subsection{Further extension using stochastic gradient Langevin dynamics}

Most statistical problems require a full pass through the dataset to evaluate the likelihood function and its derivatives. Therefore in big data problems the cost of computing the likelihood and its gradient in Algorithm \ref{algo:2} may become a computational bottleneck. Stochastic gradient Langevin dynamics (SGLD) algorithms have recently emerged as very useful algorithms for big data problems where only a subset of the data is used to approximate the likelihood at each iteration (\cite{sgld,ahn:etal:12,ma:fox:15,li:etal:16,dubey:etal:16,bardenet:doucet:17}). To be more specific, we suppose here that the log-likelihood function has the form
\[\ell(\theta) \eqdef \sum_{i=1}^n \ell_i(\theta),\;\;\;\;\;\;\theta\in\rset^p.\]
 In that case, provided that the log-likelihood functions has a Lipschitz gradients, one can naturally take  $\tilde P_\delta$ in (STEP 1) of Algorithm \ref{algo:2} as a SGLD kernel. The resulting algorithm is presented in Algorithm \ref{algo:4}.  Note that one can also approximate the Bernoulli probability $\tilde q_j$ in (\ref{cond:dist:eq:2}) using the selected mini-batch. For that purpose, and given a mini-batch $\mathsf{I}\subset\{1,\ldots,n\}$ of size $B$, $1\leq j\leq p$, and $(\delta,\theta)\in\Delta\times \rset^p$, we set
\begin{multline}\label{cond:dist:eq:3}
\hat q_j(\delta,\theta) \eqdef \left(1 +\exp\left (\mathsf{a} +\frac{1}{2}(\rho_1 -\rho_0)\theta_j^2 - \theta_j \hat G_j(\theta_{\delta})  -\frac{ \theta_j^2}{2} \hat G_j(\theta_{\delta})^2 \right)\right)^{-1},\\
\;\;\;\mbox{ where }\;\; \hat G_j(\cdot) = \frac{n}{B}\frac{\partial}{\partial \theta_j}\left[\sum_{i\in\mathsf{I}} \ell_i\right](\cdot).\end{multline}

For numerical implementation it is important to notice that in (STEP 2) all the probability $\hat q_j(\vartheta,\bar\theta)$  are  computed based on the same gradient estimate $\hat G(\bar\theta_{\vartheta})$, that is computed only once.

\medskip

Using the same cost computing assumption as above,  and ignoring the cost of generating univariate Gaussian random variables,  we see that the computation cost of the $k$-th iteration of Algorithm \ref{algo:4} is  now of order $O(B\|\delta^{(k)}\|_0 + J)$. However, Algorithm \ref{algo:4} now has a higher approximation error. Our numerical experiments suggest  that this higher approximation error does not impact mixing, but rather the quality of the limiting distribution.

\medskip

\begin{algorithm}\label{algo:4}[Asynchronous sparse SGLD]$\hrulefill$\\
Draw $(\delta^{(0)},\theta^{(0)})$ from some initial distribution, and repeat the following steps for $k=0,\ldots$.  Given $(\delta^{(k)},\theta^{(k)})= (\delta,\theta)\in\Delta\times\rset^p$:
\begin{description}
\item [(STEP 1)] For all $j$ such that $\delta_j=0$, draw independently $\bar\theta_j\sim  \textbf{N}(0,\rho_0^{-1})$. Provided that $\|\delta\|_0>0$, randomly select a data mini-batch $\mathsf{I}\subset\{1,\ldots,n\}$ of size $B$, draw $Z\sim\textbf{N}(0, I_{\|\delta\|_0})$, and set 
\begin{equation}\label{sgld}
[\bar\theta]_{\delta} = [\theta]_{\delta} +\gamma\left(-\rho_1[\theta]_{\delta} + \frac{n}{B}\sum_{i\in\mathsf{I}}  \nabla \ell_i(\theta_{\delta})\right) +\sqrt{2\gamma}Z,\end{equation}
where $\gamma>0$ is the step-size.
\item [(STEP 2)] Set $\bar\delta = \delta$. Randomly select a subset $\mathsf{J}\subset\{1,\ldots,p\}$ of size $J$. Form the vector $\vartheta = \vartheta(\mathsf{J},\delta)$ as in (\ref{delta:J}). 
\begin{itemize}
\item For each $j\in\mathsf{J}$, draw independently $\bar \delta_j\sim \textbf{Ber}(\hat q_j(\vartheta,\bar\theta))$, where $\hat q_j$ is as in (\ref{cond:dist:eq:3}).
\end{itemize}
\item [(STEP 3)]  Set $(\delta^{(k+1)},\theta^{(k+1)})= (\bar \delta, \bar \theta)$.
\end{description}
\vspace{-0.6cm}
\end{algorithm}
$\hrulefill$
\medskip


\section{Approximate correctness}\label{sec:theory}
In this section the dataset $\D$ is assumed fixed, and we shall omit the dependence of the Markov kernels on $\D$. Let $K$ (resp. $\tilde K$) be the transition kernel of the Markov chain generated by Algorithm \ref{algo:1}  (resp. Algorithm \ref{algo:2}).  We first write the expression of $K$ and $\tilde K$ and introduce some useful notations in the process. Given $\theta\in\rset^p$, and $j\in\{1,\ldots,p\}$, let $Q_{\theta,j}$ be the Markov kernel on $\Delta$ which, given $\delta\in\Delta$, leaves $\delta_i$ unchanged for all $i\neq j$, and update  $\delta_j$ using a draw from $\textbf{Ber}(q_j(\delta,\theta))$:
\[Q_{\theta,j}(\delta,\delta') =q_j(\delta,\theta)^{\delta_j'} (1- q_j(\delta,\theta))^{1-\delta_j'}  \prod_{i\neq j}\textbf{1}_{\{\delta_i = \delta_i'\}},\;\;\delta,\delta'\in\Delta.\]
 Given $\mathsf{J}=\{j_1,\ldots,j_J\}\subseteq \{1,\ldots,p\}$, we multiply the Markov kernels $Q_{\theta,j_i}$ together to form $Q_{\theta,\mathsf{J}}$:
\[Q_{\theta,\mathsf{J}} \eqdef  Q_{\theta,j_1}\times \cdots\times Q_{\theta,j_J},\]
where the Markov kernel multiplication is as defined in Section \ref{sec:notations}. Let $\tilde Q_{\theta,\mathsf{J}}$ be the Markov kernel on $\Delta$ given by
\[\tilde Q_{\theta,\mathsf{J}}(\delta,\delta') \eqdef \prod_{j\notin\mathsf{J}} \textbf{1}_{\{\delta_j'=\delta_j\}} \prod_{i=1}^J \tilde q_{j_i}(\vartheta,\theta)^{\delta_{j_i}'}(1 - \tilde q_{j_i}(\vartheta,\theta))^{1-\delta_{j_i}'},\;\;\delta,\delta'\in\Delta,\]
where $\vartheta=\vartheta(\mathsf{J},\delta)$ is as defined in (\ref{delta:J}). The Markov kernel $K$ of algorithm \ref{algo:1} can then be written as
\[K((\delta,\theta);(\rmd\delta',\rmd\theta')) = K_\delta(\theta,\rmd\theta') \sum_{\mathsf{J}:\;|\mathsf{J}|=J}{p\choose  J}^{-1} Q_{\theta',\mathsf{J}}(\delta,\rmd\delta'),\]
where $K_\delta$ denotes the transition kernel of (STEP 1), which can be written as
\[K_\delta(\theta,\rmd \theta') \eqdef P_\delta([\theta]_\delta,\rmd [\theta']_\delta) \prod_{j:\;\delta_j=0} \textbf{N}(0,\rho_0^{-1})(\rmd \theta_j'),\]
where $\textbf{N}(\mu,\sigma^2)(\rmd x)$ denotes the probability measure of the Gaussian distribution $\textbf{N}(\mu,\sigma^2)$ on $\rset$.  The Markov kernel $\tilde K$ has the same structure, but with $P_\delta$ replaced by $\tilde P_\delta$ and $Q_{\theta,\mathsf{J}}$ replaced by $\tilde Q_{\theta,\mathsf{J}}$. Hence the difference between Algorithm \ref{algo:1} and Algorithm \ref{algo:2} is driven by the difference between $P_\delta$ and $\tilde P_\delta$ (STEP 1), and the difference between $Q_{\theta,\mathsf{J}}$ and $\tilde Q_{\theta,\mathsf{J}}$ (STEP 2). By construction, the invariant distribution of $K$ is $\Pi(\cdot\vert \D)$. We  show next that under some additional assumptions $\tilde K$  possesses an invariant distribution that we denote $\tilde\Pi(\cdot\vert \D)$. We make the assumption that  for each fixed model $\delta$, the kernel $\tilde P_\delta$ used in (STEP 1) satisfies  a geometric drift condition.  More specifically, we make the following assumption.

\begin{assumption}\label{H1}
For each $1\leq j\leq p$ there exists $V_j:\rset\to [1,\infty)$, such that $\max_j \int V_j(x)e^{-\rho_0^{-1} x^2/2}\rmd x<\infty$, and the following holds. For each $\delta\in\Delta$, with $\|\delta\|_0>0$, there exist $\lambda_\delta\in (0,1)$, $b_\delta<\infty$ such that
\[\tilde P_\delta V_\delta(u) \leq \lambda_\delta V_\delta(u) + b_\delta,\;\;u\in\rset^{\|\delta\|_0},\]
where $V_\delta(u) \eqdef V(\delta,(u,0)_\delta)$, and $V(\delta,\theta)\eqdef \sum_{j=1}^p \delta_jV_j(\theta_j)$.
\end{assumption}

The next result follows easily from H\ref{H1}.

\begin{proposition}\label{prop:1}
Assume H\ref{H1}, and suppose that $\tilde K$ is phi-irreducible and aperiodic, and for all $b<\infty$, the set $\{(\delta,\theta)\in\Delta\times \rset^p:\; V(\delta,\theta) \leq b\}$ is a petite set for $\tilde K$. Then $\tilde K$ possesses a unique invariant distribution $\tilde \Pi$, and there exist $\tilde\lambda\in(0,1)$, a constant $C_0$ such that for all $(\delta,\theta)\in\Delta\times\rset^p$, and all $k\geq 0$,
\begin{equation}\label{geo:ergo:tilde}
\|\tilde K^k((\delta,\theta),\cdot)  - \tilde\Pi\|_\tv \leq  C_0\tilde \lambda^k V^{1/2}(\delta,\theta).
\end{equation}
\end{proposition}
\begin{proof}
See Section \ref{sec:proof:prop:1}.
\end{proof}

We compare next the stationary distribution $\tilde\Pi(\cdot \vert \D)$ of Algorithm \ref{algo:2}  to the posterior distribution $\Pi(\cdot\vert \D)$ in (\ref{post:Pi}).  Since the invariant distributions $\Pi(\cdot\vert \D)$ and $\tilde \Pi(\cdot\vert \D)$ are eigen-measures of their corresponding Markov operators, comparing $\tilde\Pi(\cdot \vert \D)$ and $\Pi(\cdot \vert \D)$ is a form of Davis-Kahan theorem (\cite{davis:kahan:70}). As such we expect an upper bound on $\tilde\Pi(\cdot \vert \D) - \Pi(\cdot \vert \D)$  to involve the inverse of the spectral gap of $K$ or $\tilde K$, and a comparison of the kernels $K$ and $\tilde K$.  Several such results  have been derived recently in the literature (\cite{pillai:smith:15,rudolf:etal:18,johndrow:mattingly:18}), and the references therein. To the exception of (\cite{pillai:smith:15}) most of these results uses a strategy that compares $K((\delta,\theta),\cdot)$ and $\tilde K((\delta,\theta),\cdot)$ for all $(\delta,\theta)$, yielding bounds that cannot leverage Bayesian posterior contraction. We develop a more suitable bound  that involves  comparing $K((\delta,\theta),\cdot)$ and $\tilde K((\delta,\theta),\cdot)$ only for $(\delta,\theta)\in\cB$, for some set $\cB$ such that $\Pi(\cB\vert \D)\approx 1$. Our approach is similar to, but differs in the details from (\cite{pillai:smith:15}). We shall make the assumption that 

\begin{assumption}\label{H2} There exist $\Delta_0\subseteq\Delta$, measurable sets  $\cB_\delta\subseteq\rset^p$ for each $\delta\in\Delta_0$, and
\[\cB\subseteq \bigcup_{\delta\in\Delta_0} \{\delta\}\times \cB_\delta,\]
such that $\Pi(\cB^c\vert\D)$ is small, where $\cB^c$ denotes the complement of $\cB$. Furthermore,
\begin{equation}\label{control:KV}
\sup_{(\delta,\theta)\in\cB}\;\left[V(\delta,\theta) + K V(\delta,\theta) + \tilde K V(\delta,\theta)\right]<\infty.
\end{equation}
\end{assumption}
\medskip

Given the sets in H\ref{H2}, we define 
\begin{equation}\label{def:eta:1}
\eta_1 \eqdef \sup_{(\delta,\theta)\in\cB}\;\|P_\delta([\theta]_\delta,\cdot) - \tilde P_\delta([\theta]_\delta,\cdot)\|_\tv,
\end{equation}
\begin{equation}\label{def:eta:2}
\eta_2 \eqdef \sup_{(\delta,\theta)\in\cB} \left\| \int\tilde K_\delta(\theta,\rmd\theta') \sum_{\mathsf{J}:\;|\mathsf{J}|=J}{p\choose  J}^{-1} \left(Q_{\theta',\mathsf{J}}(\delta,\cdot) - \tilde Q_{\theta',\mathsf{J}}(\delta,\cdot)\right)\right\|_\tv.
\end{equation}

\begin{theorem}\label{thm:2}Suppose that H\ref{H1} and H\ref{H2} hold, and $\Pi(V\vert \D)<\infty$. Then $\tilde K$ possesses a unique invariant distribution $\tilde \Pi$, and  there exists a finite constant $C_0$ such that
\begin{equation}\label{bound:thm2}
\|\tilde \Pi(\cdot\vert \mathcal{D}) - \Pi(\cdot\vert \mathcal{D})\|_\tv\leq \frac{C_0}{1-\tilde\lambda}\left(\Pi(\cB^c\vert\D)^{1/2} + \eta_1^{1/2} + \eta_2^{1/2} \right),
\end{equation}
where $\tilde\lambda$ is as in Proposition \ref{prop:1}.
\end{theorem}
\begin{proof}
See Section \ref{sec:proof:thm:2}. 
\end{proof}

\begin{remark}\label{rem:thm:2}
Theorem \ref{thm:2} quantifies the observation that Algorithm \ref{algo:2} behaves like Algorithm \ref{algo:1} if the Markov kernels $\tilde P_\delta$ and $\tilde Q_{\theta,\mathsf{J}}$ are close to $P_\delta$ and $Q_{\theta,\mathsf{J}}$ respectively. The result can  leverage posterior contraction properties of $\Pi$ for a more refined comparison. However the dependence of the bound on the spectral gap $1-\tilde\lambda$ is a major roadblock for applying Theorem \ref{thm:2}. Indeed, even in the relatively simple setting of linear regression  models, the ways in which the spectral gap $1- \tilde\lambda$ of $\tilde K$  depends on $n$ and $p$  is difficult to establish. This prevents up from fully characterizing the proximity between $\tilde \Pi(\cdot\vert \mathcal{D})$  and $\Pi(\cdot\vert \mathcal{D})$ in terms of $n,p$, and other primitives of the problem. 
\end{remark}

\subsection{Approximate correctness for linear regression models}
In this section we take a closer look at Algorithm \ref{algo:2} in the case of linear regression models. For reasons explained in Remark \ref{rem:thm:2} we will not rely on Theorem \ref{thm:2}. Instead we analyze directly  the marginal chain $\{\delta^{(k)},\;k\geq 0\}$ produced by Algorithm \ref{algo:2}, building on a coupling argument originally developed by \cite{desa:16}.

Given some random response $Y\in\rset^n$, and a nonrandom design matrix $X\in\rset^{n\times p}$, we consider in this section a log-likelihood function given by
\begin{equation}\label{lin:ll}
\ell(\theta) = -\frac{1}{2\sigma^2}\|Y-X\theta\|_2^2,\;\;\theta\in\rset^p,\end{equation}
for a known constant $\sigma^2$. We write $X_j$ to denote the $j$-th column of $X$, and $X_\delta$ to denote the sub-matrix of $X$ comprised of the columns of $X$ for which  $\delta_j=1$. Without any loss of generality we assume throughout that 
\begin{equation}\label{norm:Xj}
\|X_j\|_2 = \sqrt{n}.
\end{equation}
 In the linear regression considered here,  the conditional distribution of $\theta$ given $\delta$ has a closed-form Gaussian distribution. We can thus assume that (STEP 1) of Algorithm \ref{algo:2} is performed by taking a draw directly from  the conditional distribution of $\theta$ given $\delta$. In this case H\ref{H1} and the assumption of Proposition hold with $V_j\equiv 1$, $\lambda_\delta = 0$, $b_\delta=1$. Hence without any additional assumption we can apply Proposition \ref{prop:1} and conclude that Algorithm \ref{algo:2} admits an invariant distribution $\tilde\Pi(\cdot\vert\D)$.  
To compare $\Pi(\cdot\vert\D)$ and $\tilde\Pi(\cdot\vert\D)$ we make the following assumption.

\begin{assumption}\label{H:post:contr}
\begin{enumerate}
\item There exists an absolute constant $c_0<\infty$ such that
\begin{equation}
\max_{j\neq k}\left|\pscal{X_j}{X_k}\right| \; \leq \; \sqrt{c_0n\log(p)}.
\end{equation}
\item   There exist a parameter value $\theta_\star\in\rset^p$ (with sparsity support denoted $\delta_\star$, and $\ell^0$ norm $s_\star\eqdef\|\delta_\star\|_0$),  such that 
\begin{equation}
\PE_\star\left[Y - X\theta_\star \right]=0,\;\;\;\;\;\PP_\star\left[\left|\pscal{u}{Y-X\theta_\star}\right|>\sigma t\right]\leq c_1e^{-\frac{t^2}{2\|u\|_2^2}},\end{equation}
for all $u\in\rset^n$, $t\in\rset$, and some absolute constant $c_1$.
\item As $n,p\to\infty$, the ratios $\|\theta_\star\|_\infty/\log(p)$ and  $n/p$ remain bounded from above by some absolute constant $c_2$.
\end{enumerate}
\end{assumption}

\medskip

In what follows we write $\PP$ and $\PE$ to denote the probability measure and expectation operator of the Markov chains defined by the algorithms, and we write $\PP_\star$ and $\PE_\star$ for the probability measure and expectation operator related to the data generating distribution as assumed in H\ref{H:post:contr}.

\begin{remark}
Assumption H\ref{H:post:contr}-(2) assumes that the regression errors are sub-Gaussian. (2)-(3) are mild assumptions.  Assumption H\ref{H:post:contr}-(1) is also a standard assumption is sparse signal recovery and assumes that the correlation between any two distinct columns of $X$ is of order $\sqrt{\log(p)/n}$. 
\end{remark}

\medskip

 We set
\[\underline{\theta}_\star \eqdef \min_{j:\;\delta_{\star j}=1} |\theta_{\star j}|.\]

\begin{theorem}\label{thm:3}
Consider the linear regression model presented above and assume  H\ref{H:post:contr}. Suppose that 
\[\rho_1 = 1,\;\;\mbox{ and }\;\;\rho_0 = \frac{n}{\sigma^2}.\]
Let $\mathbb{P}$ denote the distribution of the Markov chain $\{\delta^{(k)},\;k\geq 0\}$ generated by Algorithm \ref{algo:2}, and started from the null model ($\|\delta^{(0)}\|_0=0$).  There exists some constant $C_1,C_2,C_3$ that depends only on $\sigma^2,\|\theta_\star\|_\infty$, $c_0,c_1$, and $c_2$, such that for for all $n,p\geq 2$, if
\begin{equation}\label{ss:cond}
n \geq C_1\max\left(\underline{\theta}_\star^{-2}(1 + s_\star^3)\log(p),\; \underline{\theta}_\star^{-2}J^2\log(p),\; (\log(p))^3\right),\;\mbox{ and }\;\; \mathsf{u} \geq C_2(1+s_\star)^2,
\end{equation}
it holds for all $k\geq 1$,
\begin{multline}\label{eq:approx_mt}
\PE_\star\left[\max_{j:\;\delta_{\star j}=1}\;\;\left|\mathbb{P}(\delta_j^{(k)}=1) -\Pi(\delta_j=1\vert \D)\right|\right] \; \\
\leq  \left(1- \frac{3}{10}\frac{J}{p}\right)^k + 10\left(e^{-C_3 \underline{\theta}_\star \sqrt{n}} + \frac{1}{p}\right).
\end{multline}
\end{theorem}
\begin{proof}
See Section \ref{sec:proof:thm:3}.
\end{proof}

The theorem shows that in linear regression models the limiting distribution of Algorithm \ref{algo:2} recovers correctly  the relevant components of the signal. The bound in (\ref{eq:approx_mt}) can be interpreted as a mixing time bound for the Markov chain $\{\delta^{(k)},\;k\geq 0\}$, where  convergence to stationarity is measured using total variation distance on the relevant one-dimensional marginal distributions.  Importantly, the result shows that the convergence rate is at most linear in  $p/J$. However, the first part of (\ref{ss:cond}) shows that $J$ cannot be taken too large.  The theorem also shows that  the correct scaling for the prior parameter $\rho_0$ in order to achieve a good mixing is $\rho_0\propto n$. We recall that the statistical performance of the posterior does not depend on $\rho_0$. 

The first part of (\ref{ss:cond}) imposes some minimum sample size requirement. We noted empirically that the mixing time of Algorithm \ref{algo:2} degrades when $n$ is too small compared to $p$, particularly in logistic regression models. This suggests that (\ref{ss:cond}) represents some  genuine information limit of the problem.  In limited data settings where (\ref{ss:cond}) may not hold, we recommend combining Algorithm \ref{algo:2} with simulated tempering or related methods for improved mixing. However in the interest of space, we do not pursue these tempering ideas here.

The dependence of (\ref{ss:cond}) on $\underline{\theta}_\star$ is the so-called $\beta$-min condition that is commonly needed for correct model selection (see \cite{meinshausenetal09} for discussion). This condition has also appeared elsewhere in the analysis of high-dimensional MCMC samplers (\cite{yang:etal:15,atchade:asg}). The condition on $\mathsf{u}$ in (\ref{ss:cond}) is admittedly very hard to check since $s_\star$ is not known. We found in practice that for linear and logistic regression models the algorithm performs well when $\mathsf{u}$ is simply taken in the range $\mathsf{u} \in (1,2]$. 

With the same proof strategy, we believe that Theorem \ref{thm:3} can be extended to statistical models that possess the restricted strong concavity property (\cite{negahbanetal10}), under the additional assumption  that one can sample exactly from the conditional distribution of $\theta$ given $\delta$. We did not pursue this because of the limited applicability: the exact sampling assumption is highly unrealistic for non-Gaussian models.  Extending Theorem \ref{thm:3} to cases where a Markov kernel is used in (STEP 1) seems more challenging, but is likely to still hold if the Markov kernel has a strong drift toward the level sets of the target distribution. We leave this for potential future research.

\section{Numerical illustration}\label{sec:numerics}
\subsection{Linear regression}\label{sec:lm}
To illustrate Theorem \ref{thm:3} we estimate empirically the mixing time of Algorithm \ref{algo:1} and Algorithm \ref{algo:2} using the coupling methodology of  \cite{biswas2019estimating}, for increasing values of $p$. We refer the reader to Appendix \ref{append:coupled:chains} for a brief description of the estimation method and the coupled chain used. Here is the simulation set up. We generate $X\in\rset^{n\times p}$ with independent rows drawn from $\textbf{N}_p(0,\Sigma)$, where $\Sigma_{ij} = \varrho^{|j-i|}$, where $\varrho\in\{0,0.9\}$. Then we draw $Y\sim \textbf{N}_n(X\theta_\star,\sigma I_n)$, with $\sigma =1$, and a sparse $\theta_\star$ with $10$ non-zero components uniformly drawn from $(-7,-6)\cup (6,7)$. We scale the sample size as $n=p/2$. For all the results we set
\[\rho_0 = n,\;\;\;\;\rho_1 = 1, \mbox{ and  }\;\; \textsf{u}=1.5.\]
We set $J= 100$ for both MCMC samplers.  To estimate the mixing times we replicated the coupled chains $50$ times.  The estimated mixing times are given on Figure \ref{fig:mixing_lm}, and indeed shows a  linear trend. The results also show that the asynchronous sampler mixes slightly faster than Algorithm \ref{algo:1}, as we expected, due to the quadratic boost in the approximation.  We also look at the sample path of the penalized log-likelihood 
\[\bar\ell(\theta,\delta\vert Y',X') \eqdef -\frac{1}{2\sigma^2}\|Y'-X'\theta_\delta\|_2^2   -  \frac{\rho_1}{2}\|\theta_\delta\|_2^2,\]
evaluated on a test dataset $(Y',X')$ (generated independently from the traning set $(Y,X)$ but from the same model) along the MCMC iterations. By posterior contraction, we expect $\bar\ell(\theta^{(k)},\delta^{(k)}\vert Y',X')$ to concentrate around $\bar\ell(\theta_\star,\delta_\star\vert Y',X')$ as $k\to\infty$. The speed with which $\bar\ell(\theta^{(k)},\delta^{(k)}\vert Y',X')$ approaches $\bar\ell(\theta_\star,\delta_\star\vert Y',X')$ during the MCMC sampling is another empirical indication of mixing. For this comparison we run the MCMC samplers for $\textsf{Niter} = \max(2000,p-2000)$ number of iterations. Figures \ref{fig:pll_lm} and \ref{fig:pll_lm_rho09} show the averages of 50 penalized log-likelihood  sample paths (for each MCMC  sample we generate a new training and test datasets with the same $\theta_\star$). These averaged sample paths offer another look into the mixing of the samplers that is consistent with the empirical mixing times estimates. 

We also compare the parameter estimates. On a given MCMC run we evaluate the accuracy  of the parameter estimation by
\begin{equation}\label{eq:ek}
\e\eqdef \frac{1}{\textsf{Niter}-N_0}\sum_{k=N_0+1}^{\textsf{Niter}} \frac{\|\theta^{(k)}\cdot\delta^{(k)} - \theta_\star\|_2}{\|\theta_\star\|_2},\end{equation}
for a burn-in $N_0$ that we set at $N_0 = \textsf{Niter}-1000$, where $\textsf{Niter}$ is the number of MCMC iterations. Figure \ref{fig:re_lm} and \ref{fig:re_lm_rho09} show the distributions of the relative errors $\e$ produced by Algorithm \ref{algo:1} and \ref{algo:2} under various settings. Again, the difference between Algorithm \ref{algo:1} and Algorithm \ref{algo:2} remains small, even in the case $\varrho=0.9$.

\medskip
\begin{figure}
	\centering
	\includegraphics[width=.49\textwidth]{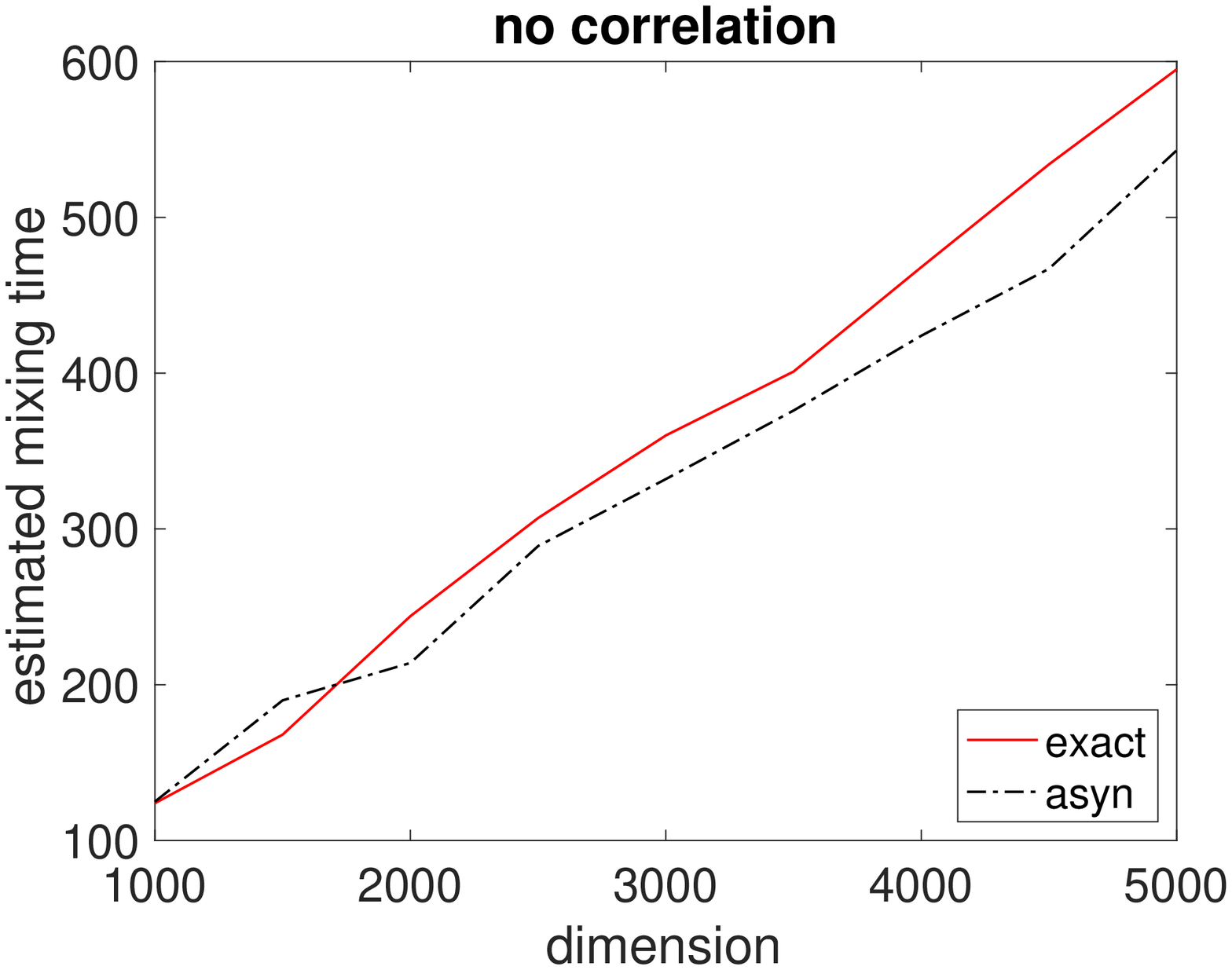}
	\includegraphics[width=.49\textwidth]{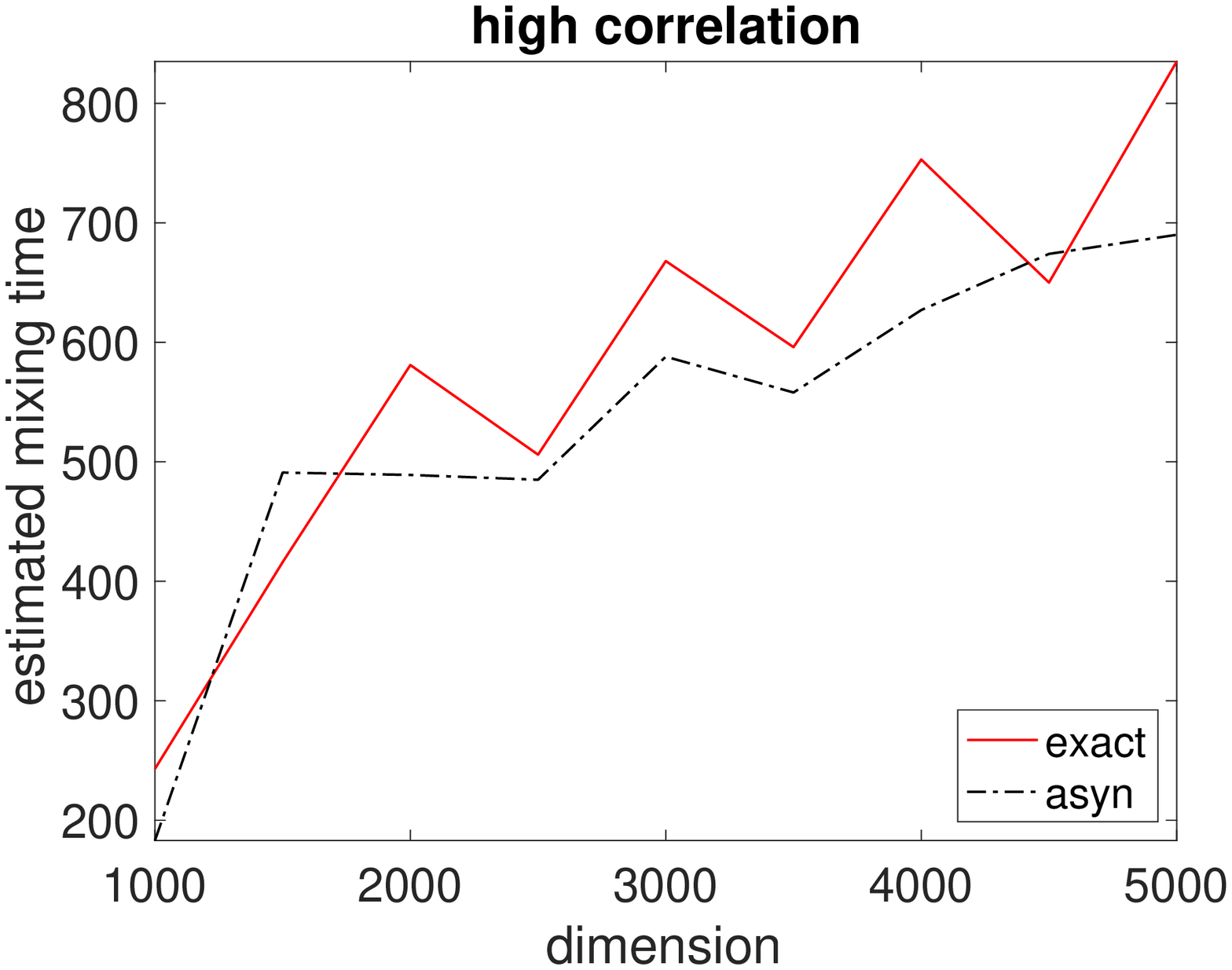}
	\caption{Estimated mixing time for a linear regression example. Figure on the left (resp. right) is $\varrho=0$ (resp. $\varrho=0.9$).}\label{fig:mixing_lm}
\end{figure}

\medskip

\begin{figure}
	\centering
	\includegraphics[width=.49\textwidth]{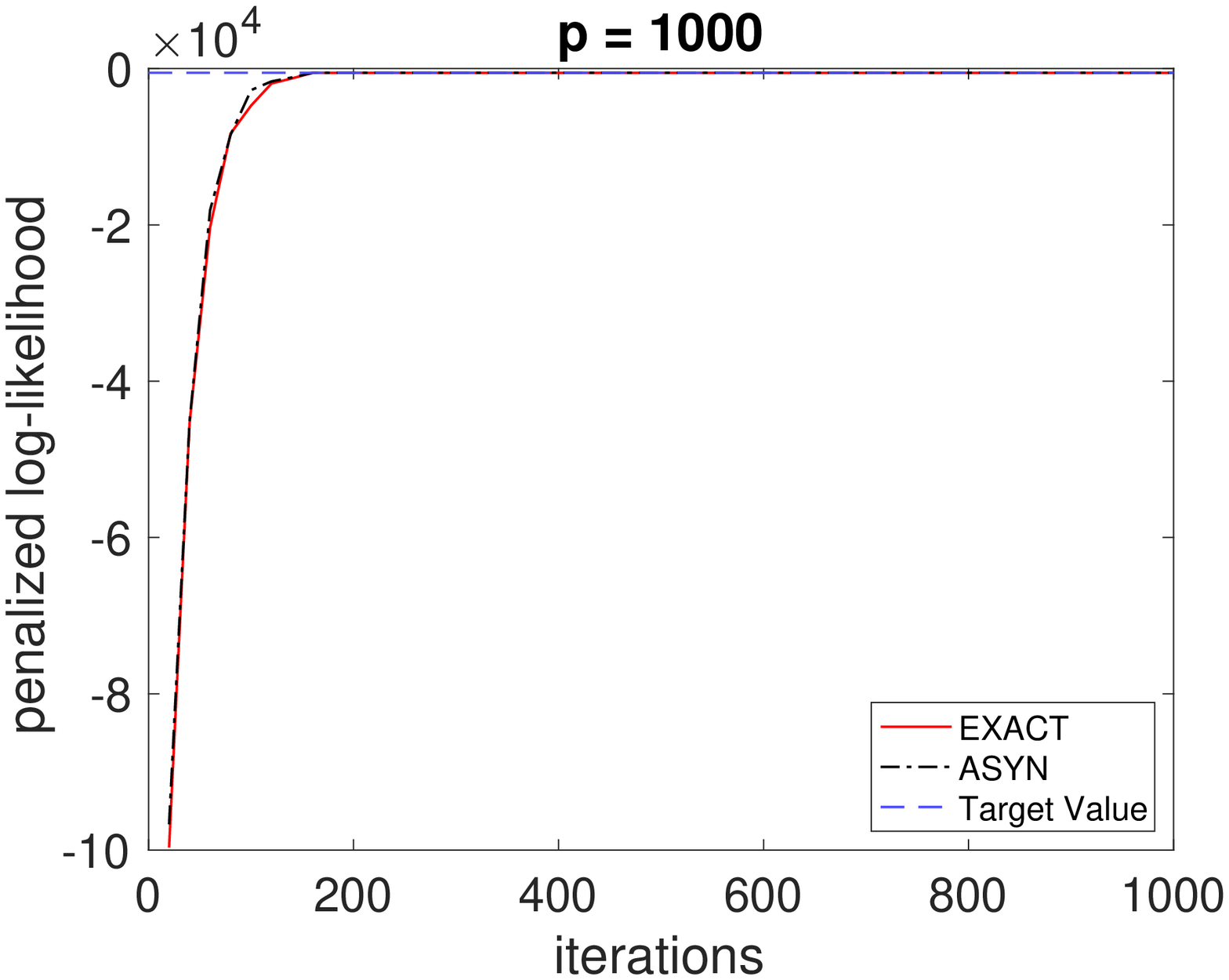}
	\includegraphics[width=.49\textwidth]{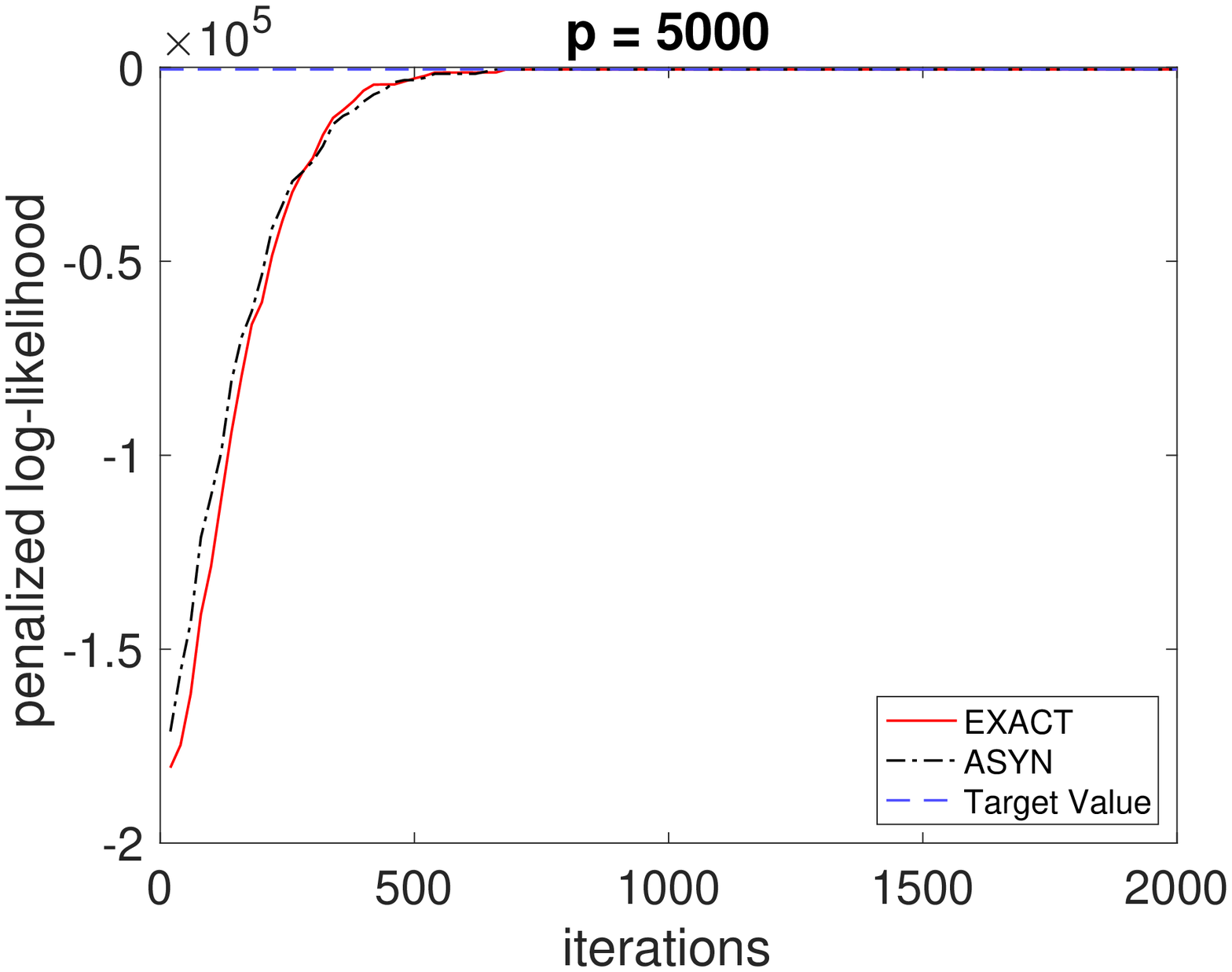}
	\caption{Averaged sample paths of penalized log-likelihood values in linear regression with $\varrho=0$}\label{fig:pll_lm}
\end{figure}
\medskip

\begin{figure}
	\centering
	\includegraphics[width=.49\textwidth]{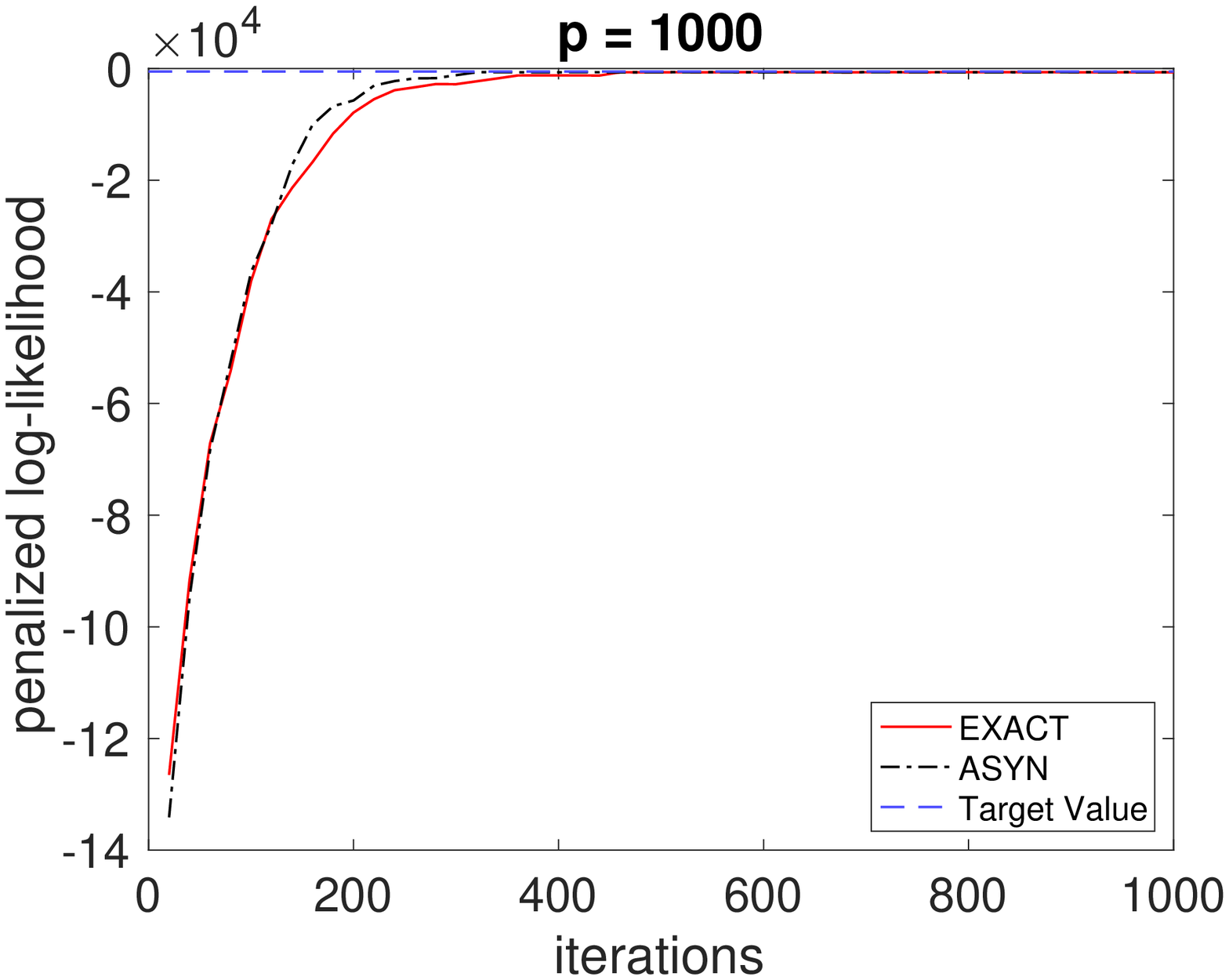}
	\includegraphics[width=.49\textwidth]{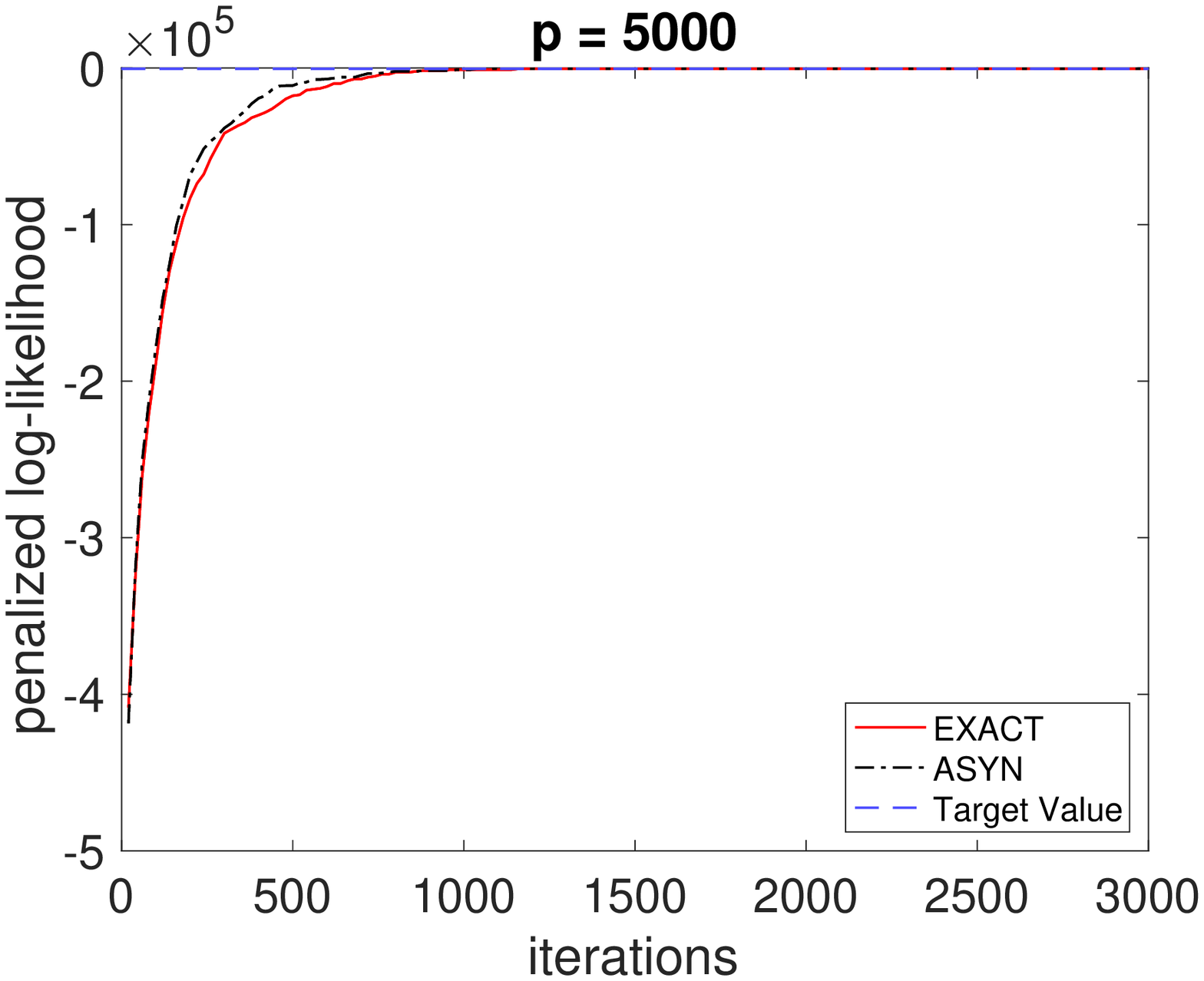}
	\caption{Averaged sample paths of penalized log-likelihood values in linear regression with $\varrho=0.9$}\label{fig:pll_lm_rho09}
\end{figure}
\medskip

\begin{figure}
	\centering
	\includegraphics[width=.49\textwidth]{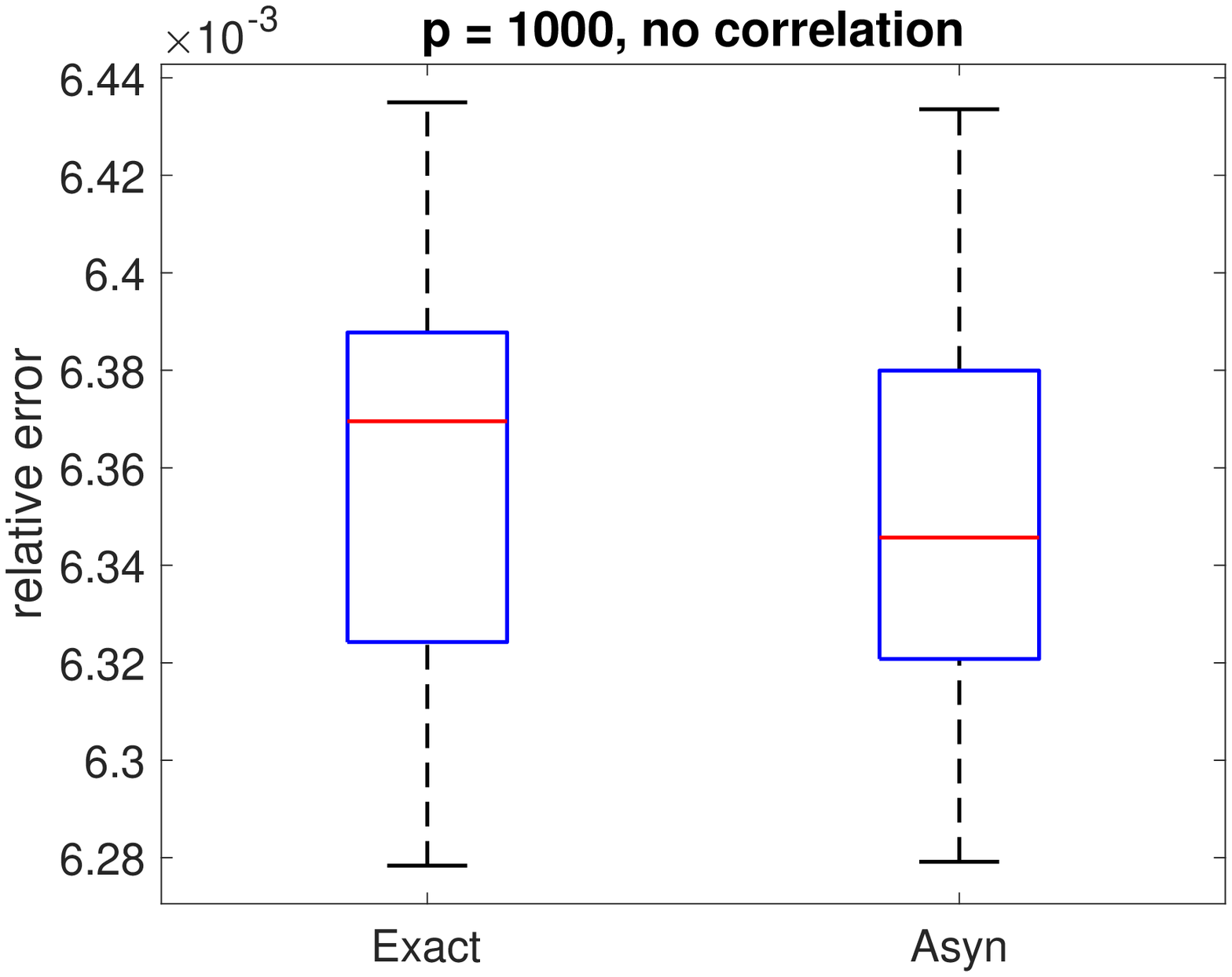}
	\includegraphics[width=.49\textwidth]{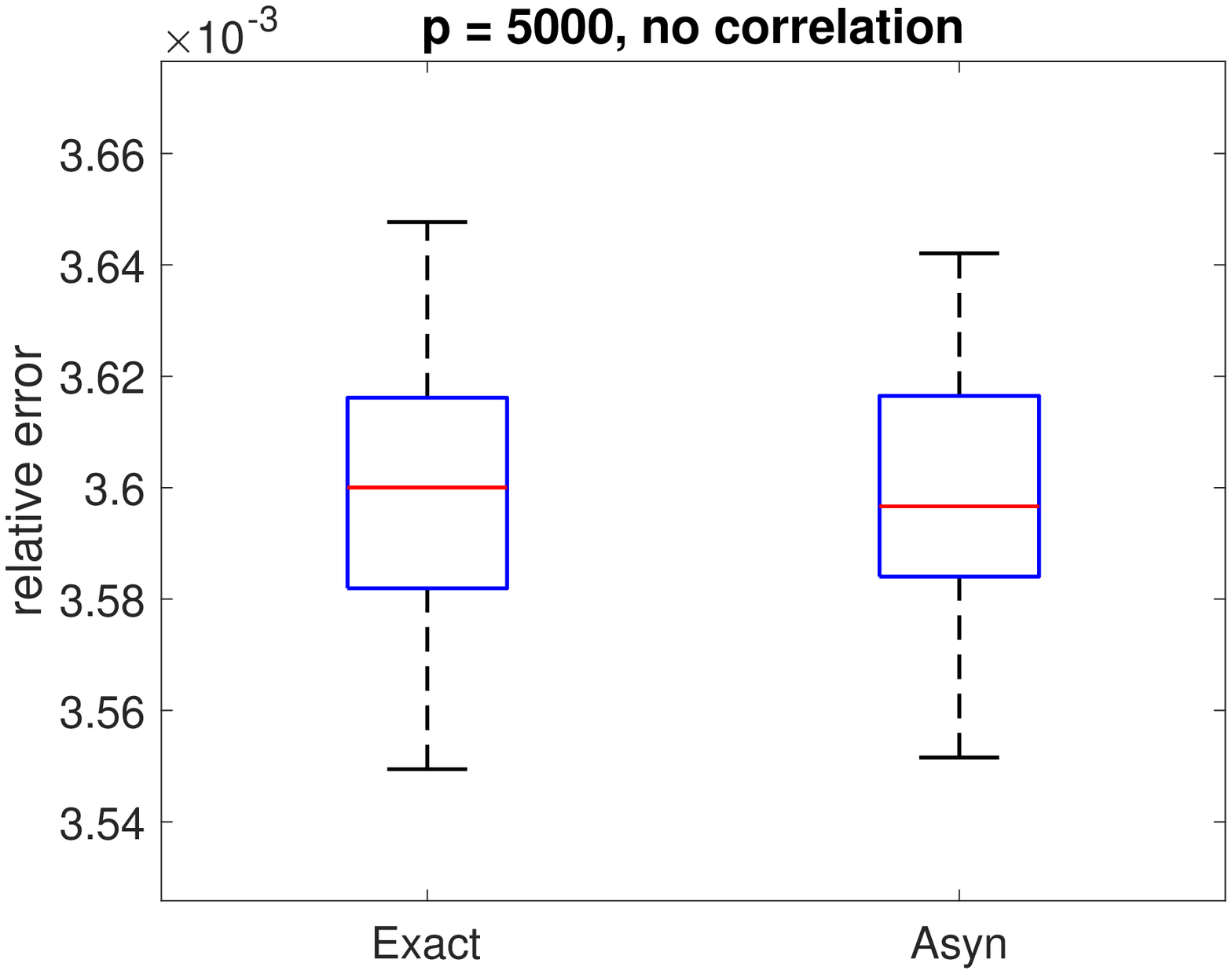}
	\caption{Averaged relative error in linear regression with $\varrho=0$}\label{fig:re_lm}
\end{figure}

\begin{figure}
	\centering
	\includegraphics[width=.49\textwidth]{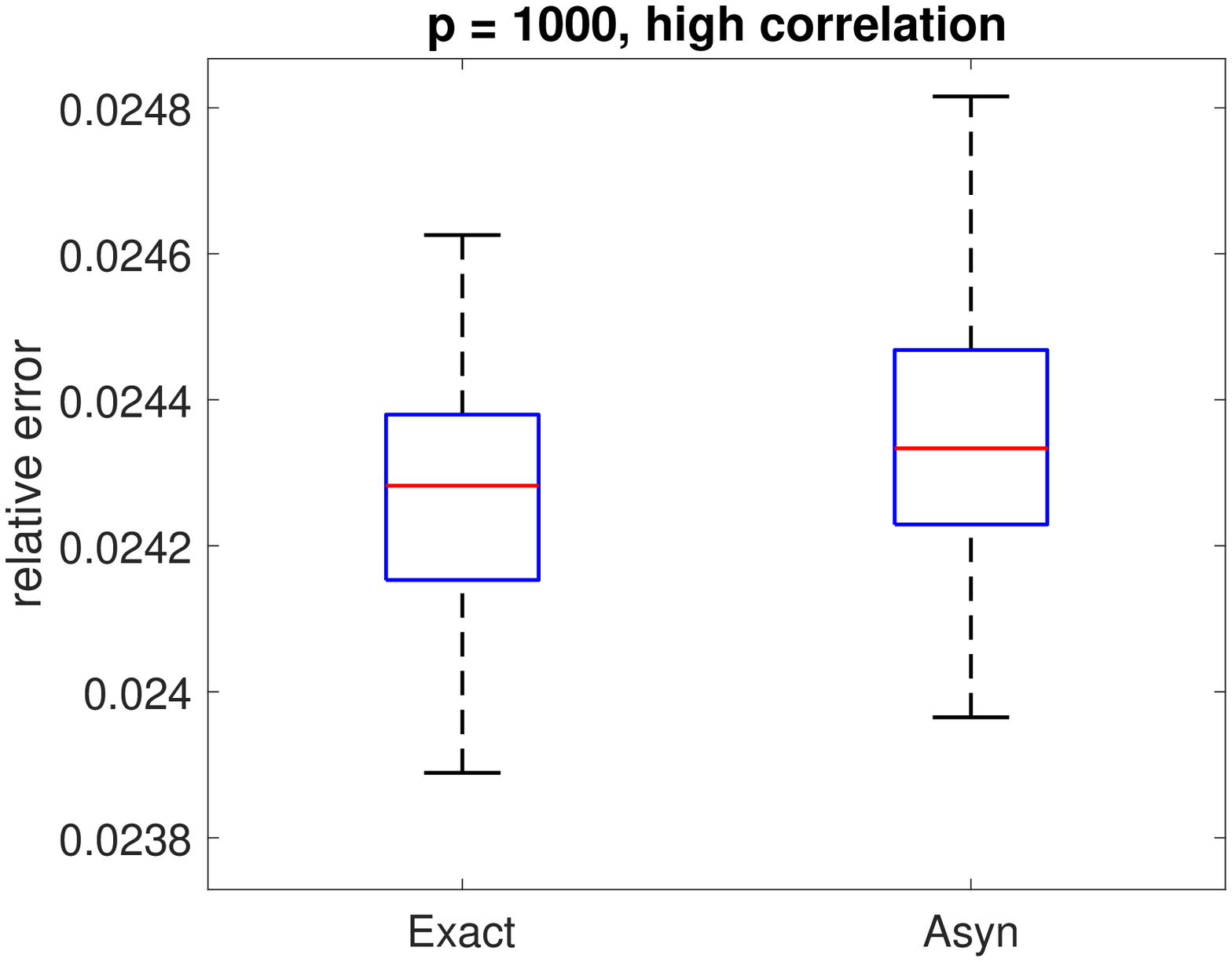}
	\includegraphics[width=.49\textwidth]{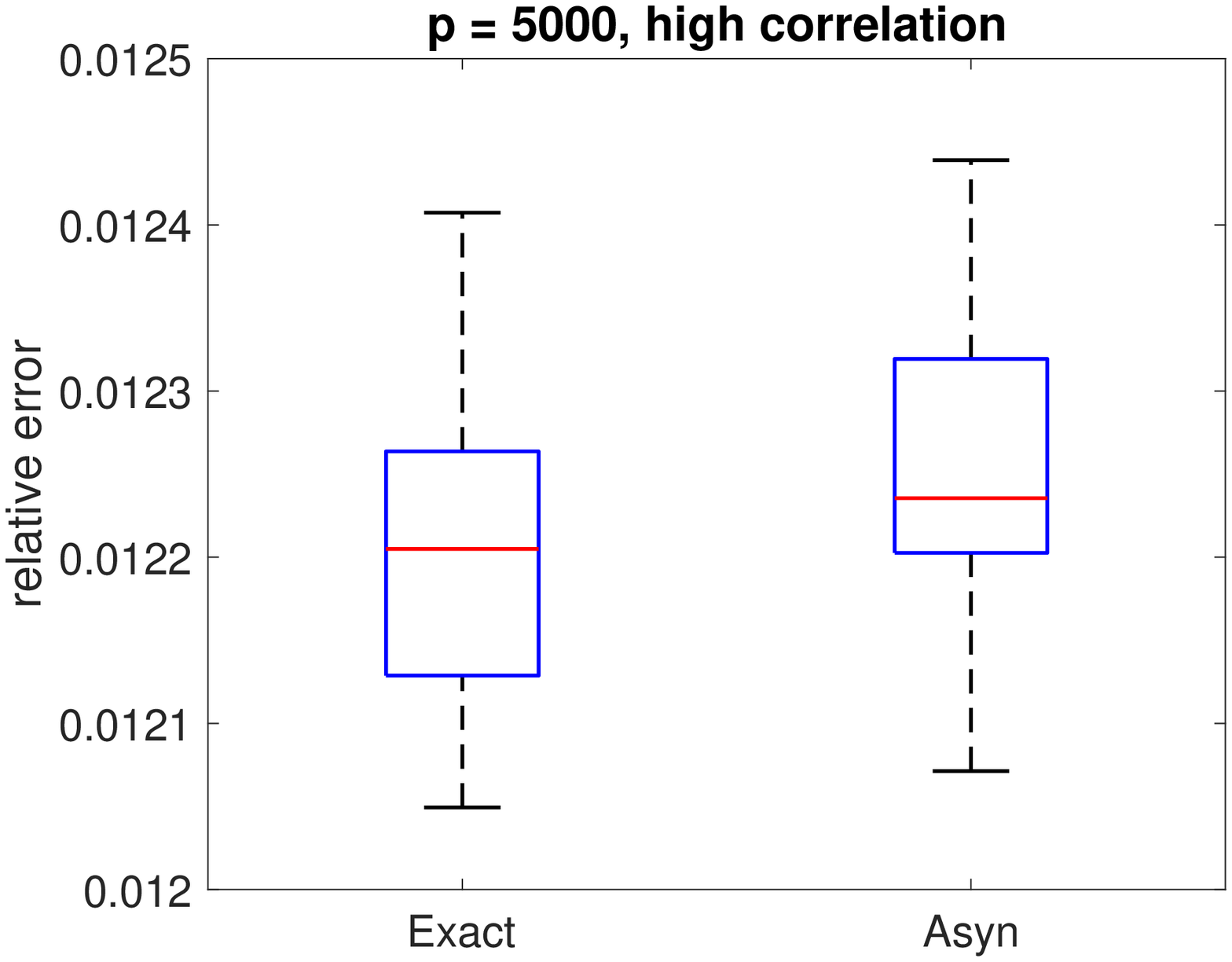}
	\caption{Averaged relative error in linear regression with $\varrho=0.9$}\label{fig:re_lm_rho09}
\end{figure}

\subsection{Logistic regression}\label{sec:logreg}
We also illustrate the behavior of the method on  logistic regression models. We use the same data generating set up for the regressors $X\in\rset^{n\times p}$ and true signal $\theta_\star$ as above. And we  draw the response as $Y_i \sim \textbf{Ber}(p_i)$, with $p_i = \left(1+ \exp^{-\pscal{x_i}{\theta_\star}}\right)^{-1}$, where $x_i$ denotes the $i$-th row of $X$. In this model we cannot draw exactly from the posterior conditional distribution of $\theta$ given $\delta$. Hence we implemented Algorithm \ref{algo:1} with $P_\delta$ taken as (one iteration of) the Metropolis Adjusted Langevin (MaLa) algorithm (\cite{roberts:tweedie96b}), with a step-size fixed to $0.01$. We consider two different implementation of Algorithm \ref{algo:2}. In the first implementation we choose $\tilde P_\delta$ to be the same MaLa 
as in Algorithm \ref{algo:1}. Whereas in the second implementation we choose $\tilde P_\delta$ to be the stochastic gradient Langevin dynamics (SGLD) kernel, with a mini-batch of size $B = 100$, and a step-size fixed to $0.005$.  We call the first implementation the asynchronous sampler (\texttt{ASYN}), and we call the second implementation the sparse asynchronous SGLD sampler (\texttt{SA-SGLD}). To improve mixing, particularly when $\varrho=0.9$, we initialize both algorithms from the \textsf{lasso} estimate of $\theta$.

To evaluate the mixing of the algorithms we look at the sample path of the penalized log-likelihood 
\[\pscal{Y'}{X'\theta_\delta} - \left\|\log\left(1 + e^{X'\theta_\delta}\right)\right\|_1   -  \frac{\rho_1}{2}\|\theta_\delta\|_2^2,\]
evaluated on a an independent test sample $(Y',X')$ (where the $\log$ and $e$ functions are evaluated componentwise), along the MCMC iterations, averaged over $50$ data and MCMC sampling replications. The results are reported on Figures \ref{pll_log}-\ref{pll_log_rho09}. We see again that both approximate samplers behaved very well.

As in the linear regression example we also  compare parameter estimates. On a given MCMC run we evaluate the accuracy  of the parameter estimation using the relative error $\e$ given in (\ref{eq:ek}).  Figures \ref{rr_log}-\ref{rr_logrho09} show the distributions of $\e$ based on 50 MCMC replications, for $p\in\{1000, 5000\}$, and $\varrho\in\{0.0, 0.9\}$. We also use this example to compare the proposed algorithms with a mean field variational approximation (VA) of (\ref{post:Pi}) using the VA family
\[\prod_{j=1}^p \textbf{Ber}(\alpha_j)(\rmd \delta_j) \textbf{N}(\mu_j,v_j^2)(\rmd\theta_j),\]
with parameter $(\alpha_j,\mu_j,v_j^2)_{1\leq j\leq p}$ that we estimate by minimizing the ELBO objective function using stochastic gradient descent. We use the re-parametrization trick of (\cite{kingma:welling:14}). In the stochastic gradient descent we estimate the gradient by drawing small sample of size $100$ from the VA family and small mini-batch of size $100$ from the dataset. For a fair comparison we also initialize $\mu^{(0)}$ from the same \textsf{lasso} estimate. We stop the stochastic gradient descent when the maximum relative change 
\[\max\left(\frac{\|\alpha^{(k)}-\alpha^{(k-1)}\|_2}{\|\alpha^{(k)}\|_2},\frac{\|\mu^{(k)}-\mu^{(k-1)}\|_2}{\|\mu^{(k)}\|_2},\frac{\|v^{(k)}-v^{(k-1)}\|_2}{\|v^{(k)}\|_2}\right) \leq 0.0025.\]
And we evaluate the accuracy  of the produced solution by computing on the last iteration the relative error
\[ \frac{\|\mu^{(k)}\cdot\alpha^{(k)} - \theta_\star\|_2}{\|\theta_\star\|_2}.\]
We observe from Figures \ref{rr_log}-\ref{rr_logrho09} and Table \ref{rtime_log} that both asynchronous MCMC are  more accurate, but also faster than the mean field VA approximation.

\medskip

\begin{figure}
	\includegraphics[width = .49\textwidth]{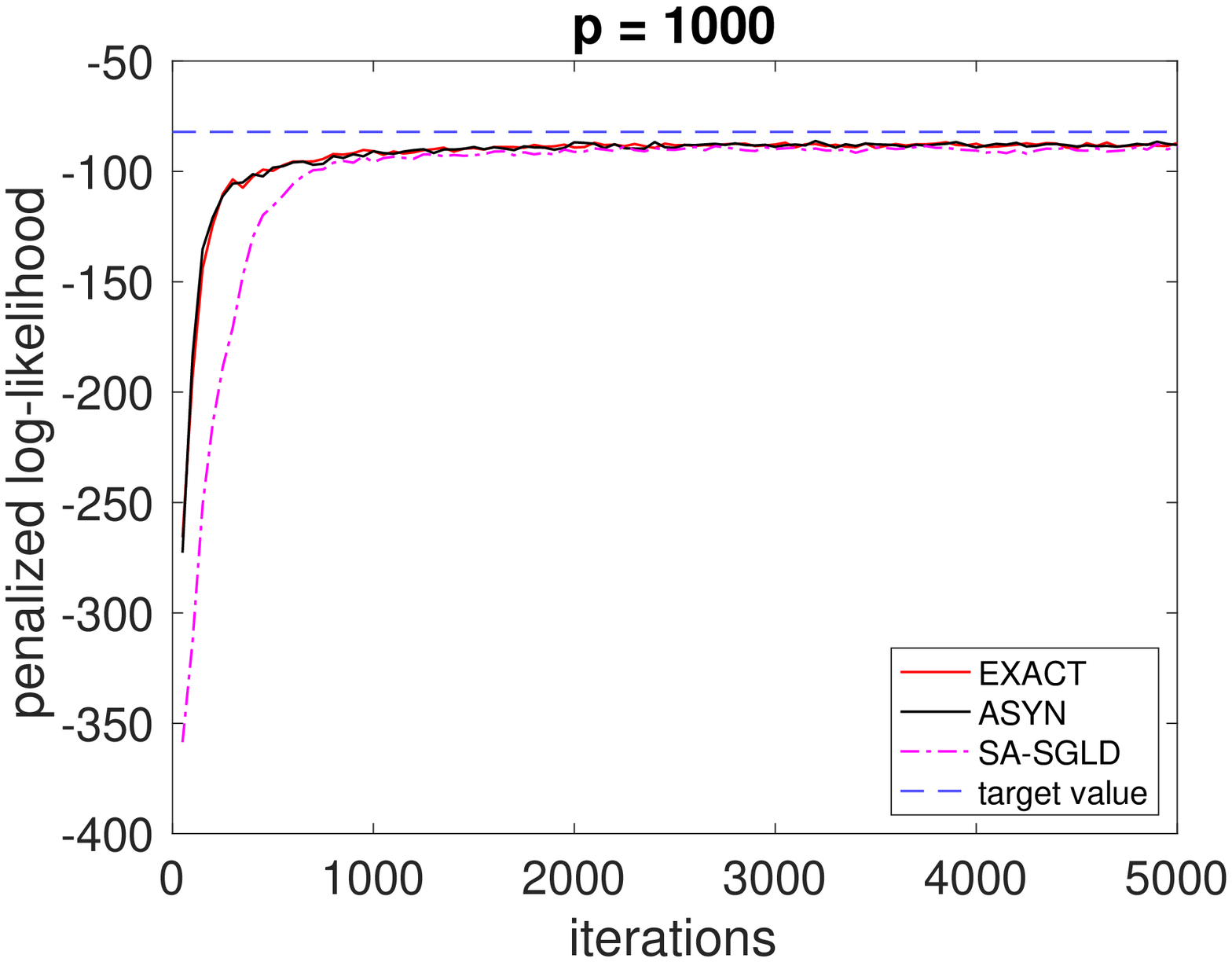}
	\includegraphics[width = .49\textwidth]{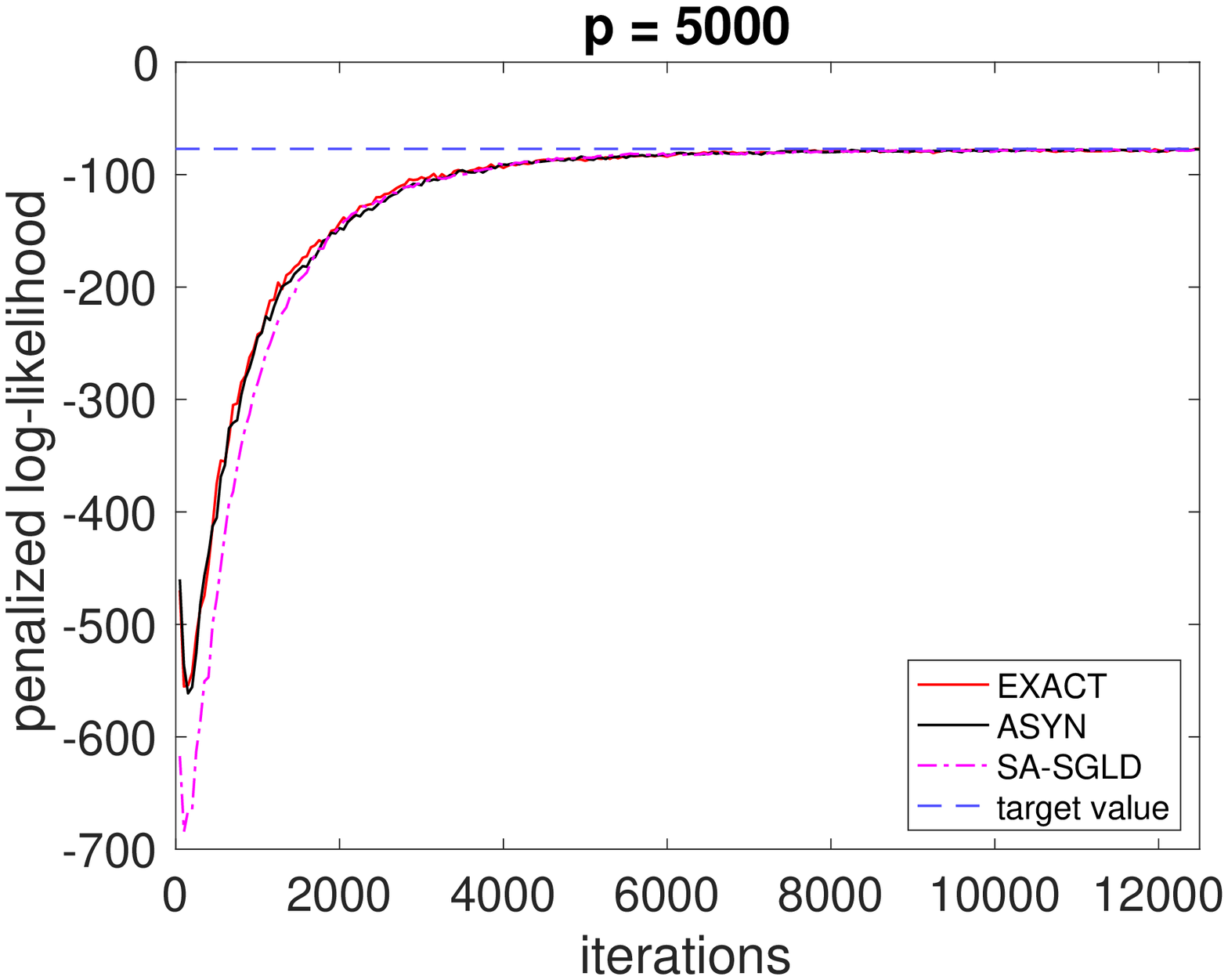}
	\caption{Averaged sample paths of penalized log-likelihood for logistic regression, with $\varrho = 0$}\label{pll_log}
\end{figure}

\medskip

\begin{figure}
	\includegraphics[width = .49\textwidth]{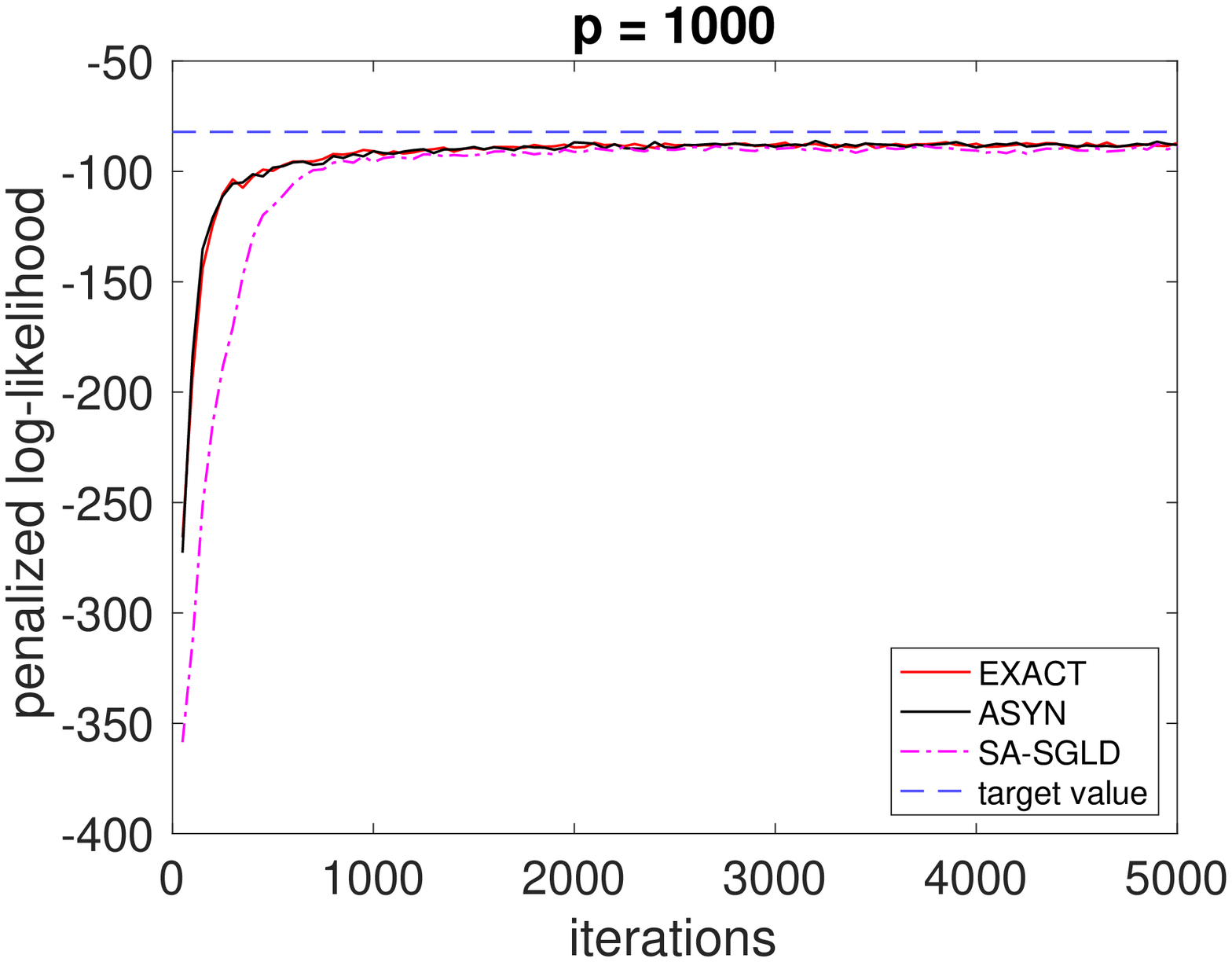}
	\includegraphics[width = .49\textwidth]{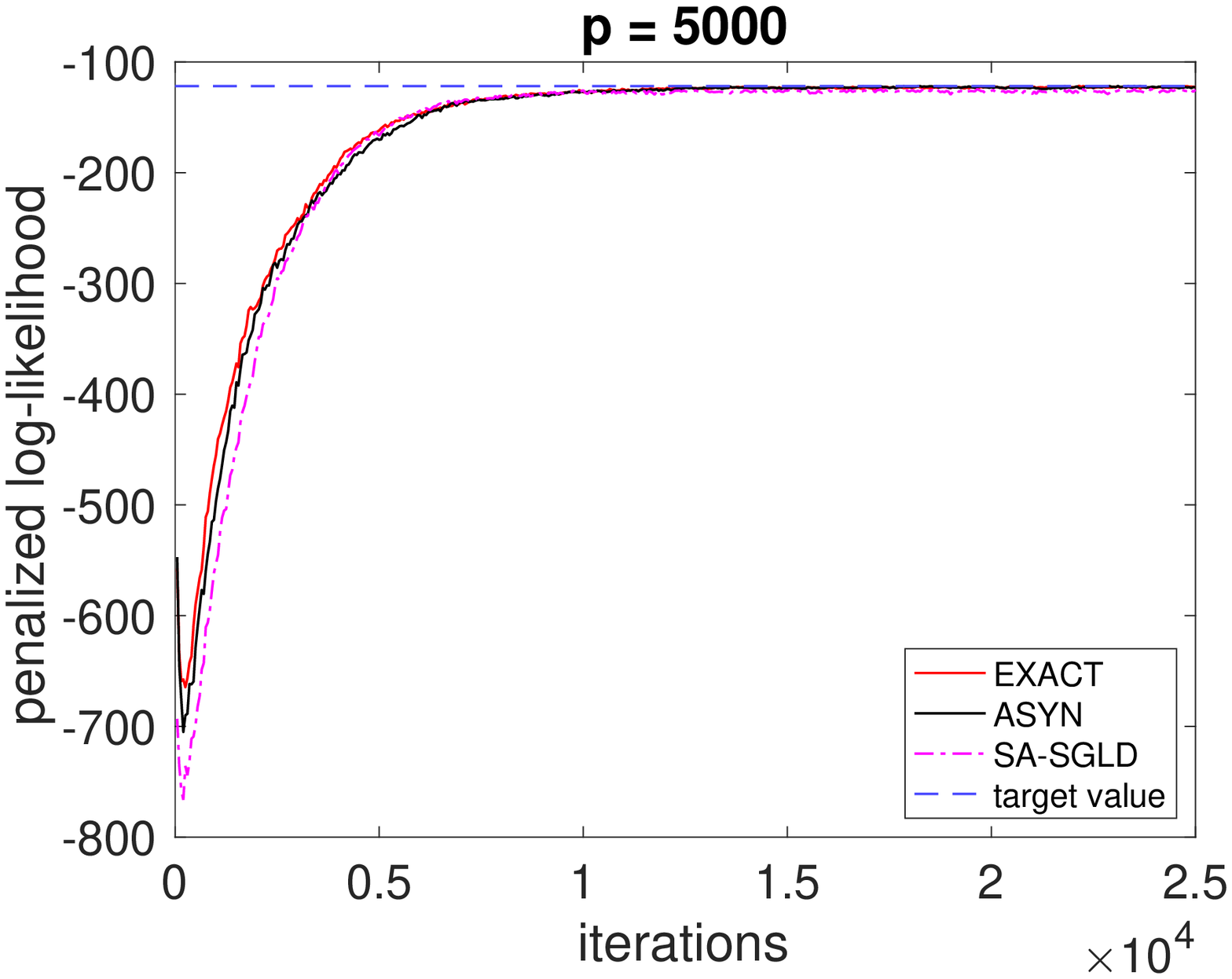}
	\caption{Averaged sample paths of penalized log-likelihood for logistic regression with $\varrho = 0.9$}\label{pll_log_rho09}
\end{figure}

\medskip

\begin{figure}
	\includegraphics[width = .49\textwidth]{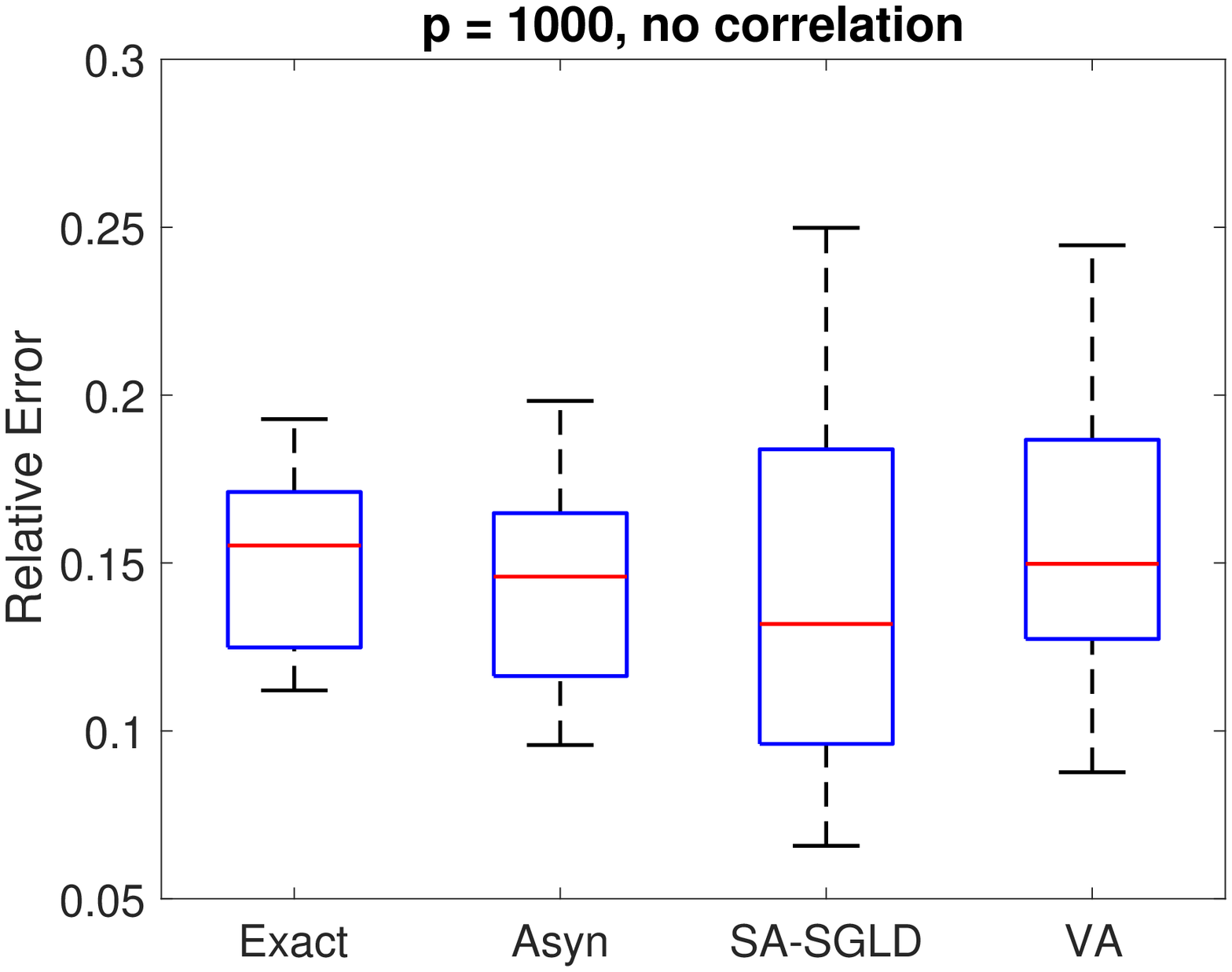}
	\includegraphics[width = .49\textwidth]{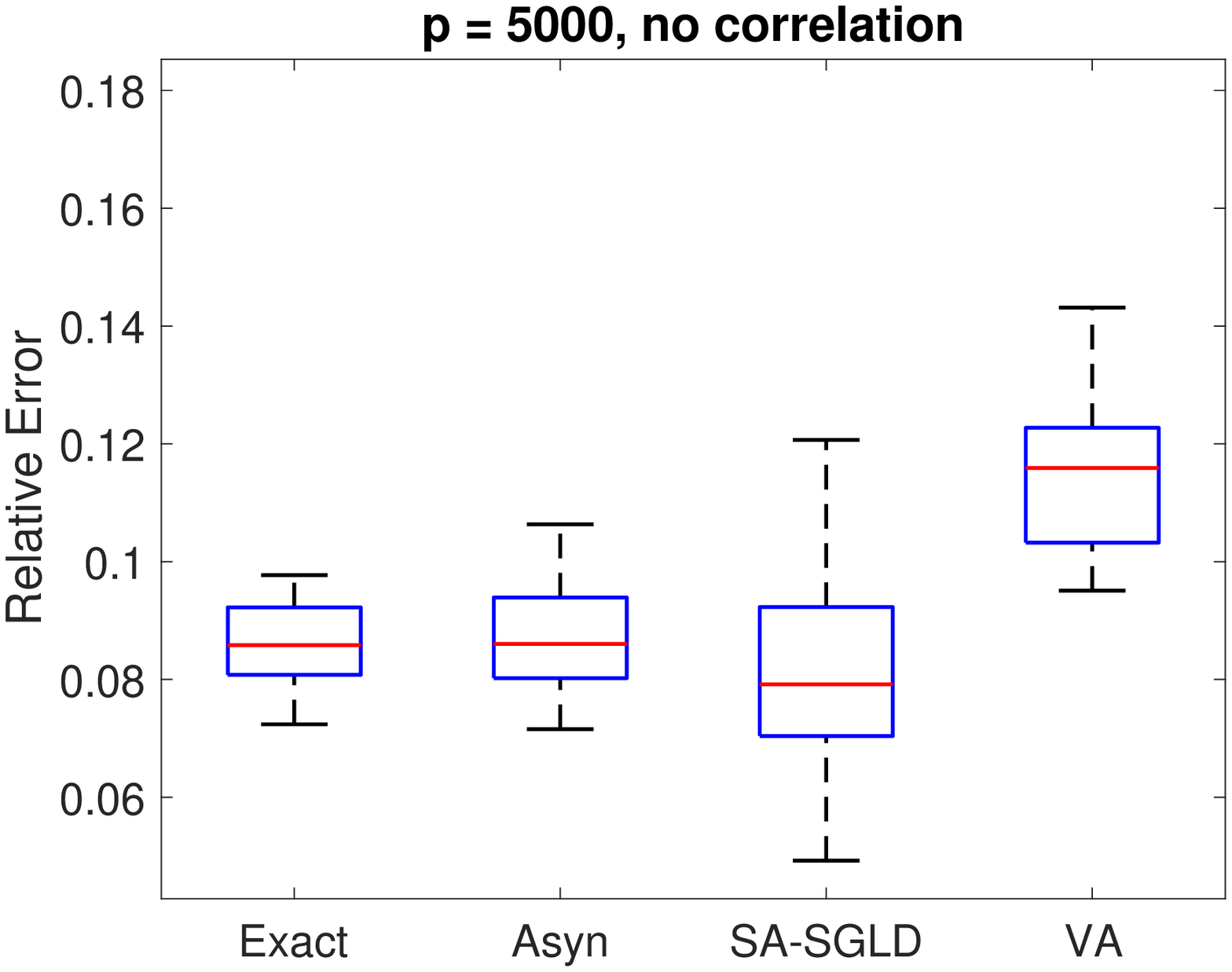}
	\caption{Relative error for logistic regression,  with $\varrho=0$}\label{rr_log}
\end{figure}

\medskip

\begin{figure}
	\includegraphics[width = .49\textwidth]{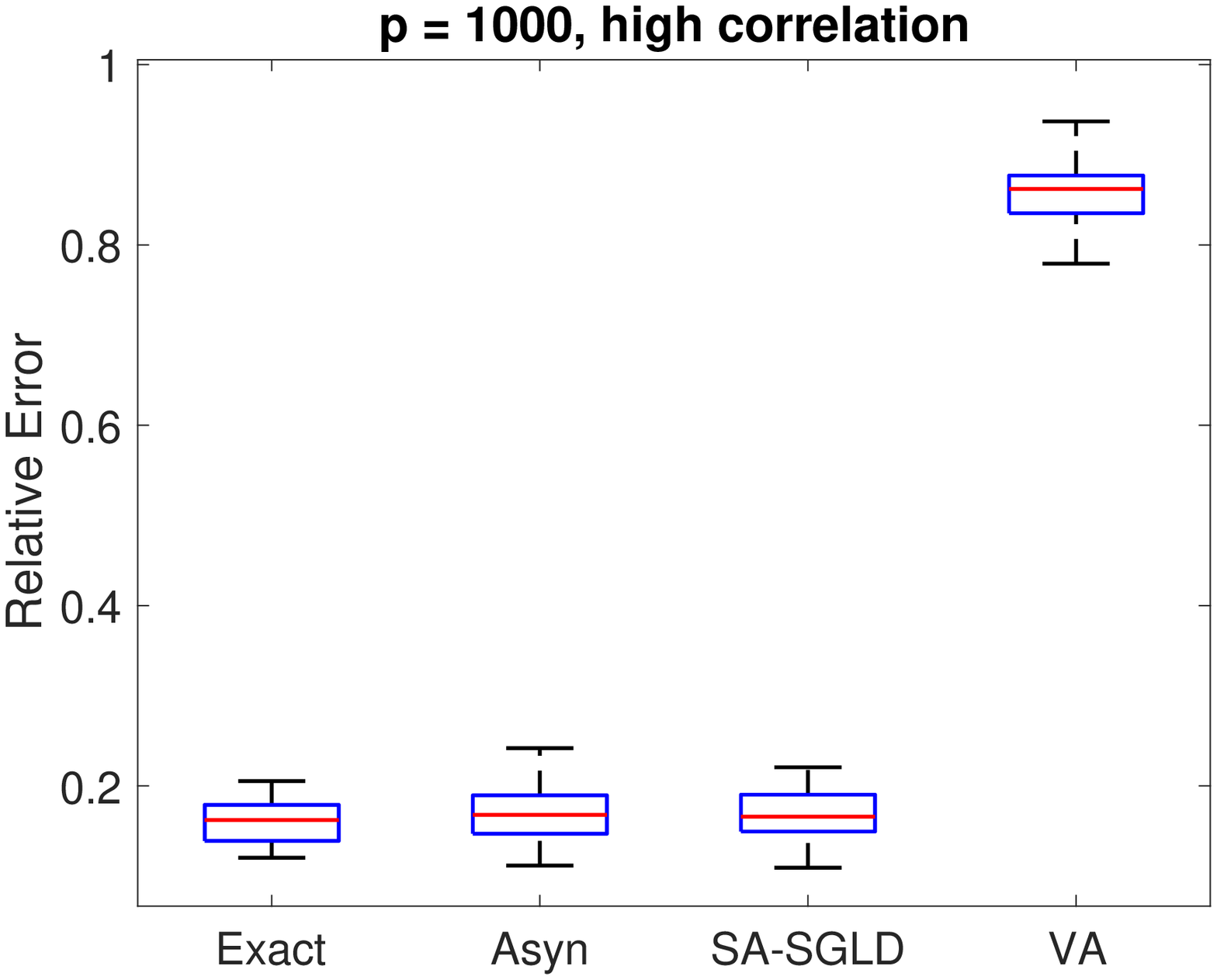}
	\includegraphics[width = .49\textwidth]{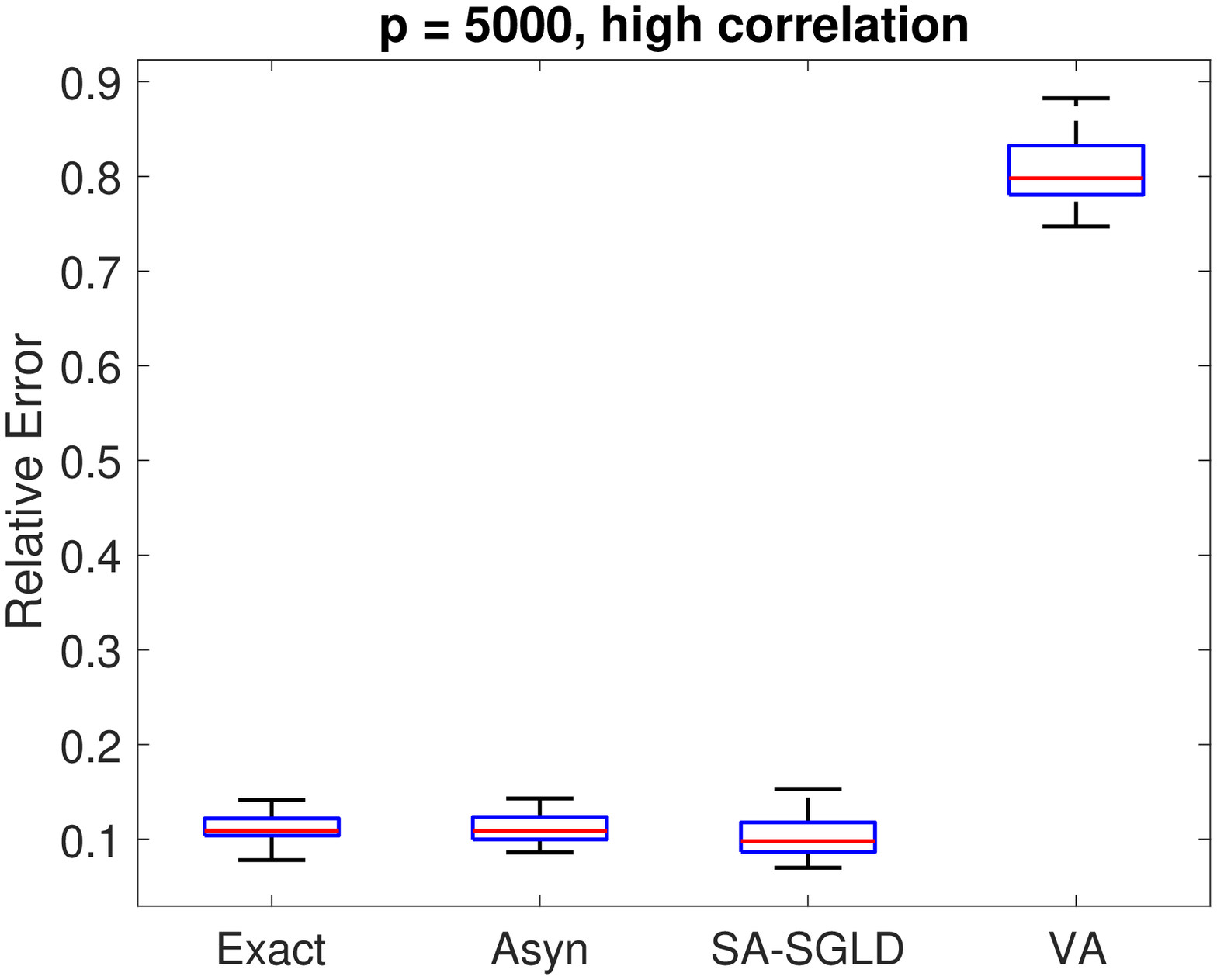}
		\caption{Relative error for logistic Regression,  with $\varrho=0.9$}\label{rr_logrho09}
\end{figure}

\medskip

\begin{table}[htbp]
	\centering 
	\begin{tabular}{cccc}
		\hline
		$p$ & 1000  & 2000 & 5000  \\             
		Exact  &  $71.61\pm 2.80s$  &$141.12 \pm 5.12s$  &$1232.71 \pm 132.90s$\\
		Asynchronous  & $8.00 \pm 0.74s$ &$16.07 \pm 1.64s$ & $142.55 \pm 24.05s$\\ 
		SA-SGLD & $2.14\pm 0.24s$ & $4.72 \pm 0.39s$ & $19.22 \pm 1.27s$\\
		VA & $59.92 \pm 11.48s$ & $116.05 \pm 7.18s $  & $1054.14 \pm 189.68s$\\
		\hline
	\end{tabular}
	\caption{Running times  for  Logistic Regression}\label{rtime_log}
\end{table}

\subsection{Illustration with a deep neural network model}
As mentioned in the introduction there is a growing interest in sparse deep learning. Both in theory (as a way to reconcile deep learning with classical statistical theory), and in the applications (for instance in the growing area of tiny machine learning for mobile AI). Most existing approach for estimating sparse deep learning models are frequentist. Bayesian deep learning can greatly facilitate uncertainty quantification in model predictions. As a proof of concept  we apply Algorithm \ref{algo:4} for a Bayesian classification of MNIST-FASHION (\cite{xiao2017fashion}) image data using  the deep neural network Lenet-5 (\cite{lecun2010mnist}), one of the smallest deep neural network models. The MNIST-FASHION dataset consists of $60,000$ data points $(y_i,{\bf x}_i)$ (plus another $10,000$ test sample), where $y_i\in\{1,\ldots,10\}$ encodes the class of a fashion item (T-shirt, trouser, etc), and ${\bf x}_i$ is a $28\times 28$ image of the item. The dataset is known to be more challenging than the more widely-known handwritten digit MNIST dataset. We model the class outcome $y_i$ as independent random variables draws from a multinomial distribution:
\[y_i\sim \M(\F_\theta({\bf x}_i)),\;\;\;i=1,\ldots,n,\]
with class probabilities proportional to $\exp(\F_\theta({\bf x}_i))$, where $\F_\theta:\;\rset^{28\times 28}\to \rset^{10}$ is a lenet-5 neural network. We actually use a slightly modified lenet-5 architecture obtained by replacing the $\mathsf{tanh}$ activation function by the $\mathsf{ReLU}$ function, and by enlarging the fully-connected layers.  We refer the reader to Figure \ref{fig:lenet} for the architecture of the network, and to \cite{d2l} for an introduction to neural network modeling. The total number of parameter is $p=298,650$. For stability in the learned structures we did not sparsify the convolutional layers (specifically, we keep their corresponding $\delta_j$ set to $1$).

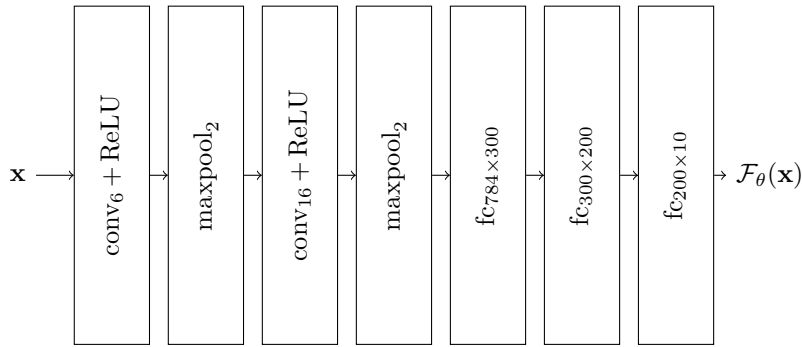
\begin{figure}
  \centering
  \hspace*{-0.75cm}
  \begin{tikzpicture}
    \node (y) at (0,0) {\small${\bf x}$}; 
    \node[conv,rotate=90,minimum width=4.5cm] (conv1) at (1.25,0) {\small$\text{conv}_{6}$\,+\,$\ReLU$};
    \node[pool,rotate=90,minimum width=4.5cm] (pool1) at (2.5,0) {\small$\text{maxpool}_{2}$};
    \node[conv,rotate=90,minimum width=4.5cm] (conv2) at (3.75,0) {\small$\text{conv}_{16}$\,+\,$\ReLU$};
    \node[pool,rotate=90,minimum width=4.5cm] (pool2) at (5,0) {\small$\text{maxpool}_{2}$};
    \node[fc,rotate=90,minimum width=4.5cm] (conv3) at (6.25,0) {\small$\text{fc}_{784\times 300}$};
    \node[fc,rotate=90,minimum width=4.5cm] (pool3) at (7.5,0) {\small$\text{fc}_{300\times 200}$};
    \node[fc,rotate=90,minimum width=4.5cm] (conv4) at (8.75,0) {\small$\text{fc}_{200\times 10}$};
      
      \node (x) at (10,0) {\small$\F_{\theta}({\bf x})$};
      
    \draw[->] (y) -- (conv1);
    \draw[->] (conv1) -- (pool1);
    \draw[->] (pool1) -- (conv2);
    
    \draw[->] (conv2) -- (pool2);
    \draw[->] (pool2) -- (conv3);
    
    \draw[->] (conv3) -- (pool3);
    \draw[->] (pool3) -- (conv4);
    
    \draw[->] (conv4) -- (x);
  \end{tikzpicture}
  \vskip 6px
  \caption{Illustration of the lenet-5 architecture.}
  \label{fig:lenet}
\end{figure}

For the Bayesian inference we use the hyper-parameter $\mathsf{u}=50$, $\rho_0 = 10^4$, $\rho_1=1$. We apply Algorithm \ref{algo:4} with SGLD on the selected components of $\theta$ with a fixed step-size $\gamma = 10^{-7}$. We set  $J=550$ (with stratified sampling accross the layer), and a batch size $B=100$.   We initialize the sampler from the full model with all components active, and the parameter $\theta$  initialized using the default  initialization in \textsf{Matlab}. We then run Algorithm \ref{algo:4} for $\textsf{Niter} = 250,000$ iterations and we use the first $150,000$ as burn-in. The running time took about $4.9$ hours on a 8-core computer node with a \textsf{NVIDIA} TESLA V100 GPU system with $384$ GB GPU memory, using \textsf{MATLAB}  2021a.

During the MCMC, at each iteration $k$, and for each $i$ in the test sample we define the prediction accurary as $\textsf{A}_i^{(k)} \eqdef \textbf{1}_{\{\hat y_i^{(k)}=y_i\}}$, where $\hat y_i^{(k)}\sim \M(\F_{\theta^{(k)}\cdot\delta^{(k)}}({\bf x}_i))$. We average these prediction accuracies to get $\bar{\textsf{A}}^{(k)}$. We also average the prediction accuracies within each group of items to get $\bar{\textsf{A}}^{(k)}(g)$, $g=1,\ldots,10$. To save time we actually compute these statistics  only  every $100$ iterations. Figure \ref{fig:lenet_sum} plots $\{\bar{\textsf{A}}^{(k)},\;k\}$  and the model sparsity $\{\|\delta^{(k)}\|_0/p,\;k\}$ along the MCMC iterations, and Figure \ref{fig:pred_error_mnist} shows the boxplots of the  $\{\bar{\textsf{A}}^{(k)}(g),\;k\}$ for each $g$. Table \ref{table:lenet_sum} shows the posterior sparsity and posterior average accuracy, and includes a comparison to Monte Carlo dropout (\cite{gal:etal:16}). The results shows that it  is possible to significantly compress deep learning models with only modest loss of performance. 

The computational cost (per iteration) of the algorithm is roughly twice that of stochastic gradient descent, its frequentist counterpart. Note however that this cost can potentially be further reduced by exploiting sparsity (as we did with linear and logistic regression models). We did not pursue this here because \textsf{MATLAB} 2021a that we used for this project does not support sparse deep learning computation.

 We end with some words of caution. We are presenting this example mainly as an exploratory exercise in the potential of the proposed framework, without much theoretical guarantee. In particular, due to the poor general understanding of deep neural network models, we currently cannot say much about the properties of the limiting distribution of Algorithm \ref{algo:4}. Furthermore, due to the lighly multimodal nature of the likelihood surface of deep neural network models,  we cannot guarantee either that the algorithm has mixed and is correctly sampling from its limiting distribution. More research is needed on these issues.

\begin{figure}
\begin{floatrow}
\ffigbox{%
  \includegraphics[scale = 0.2]{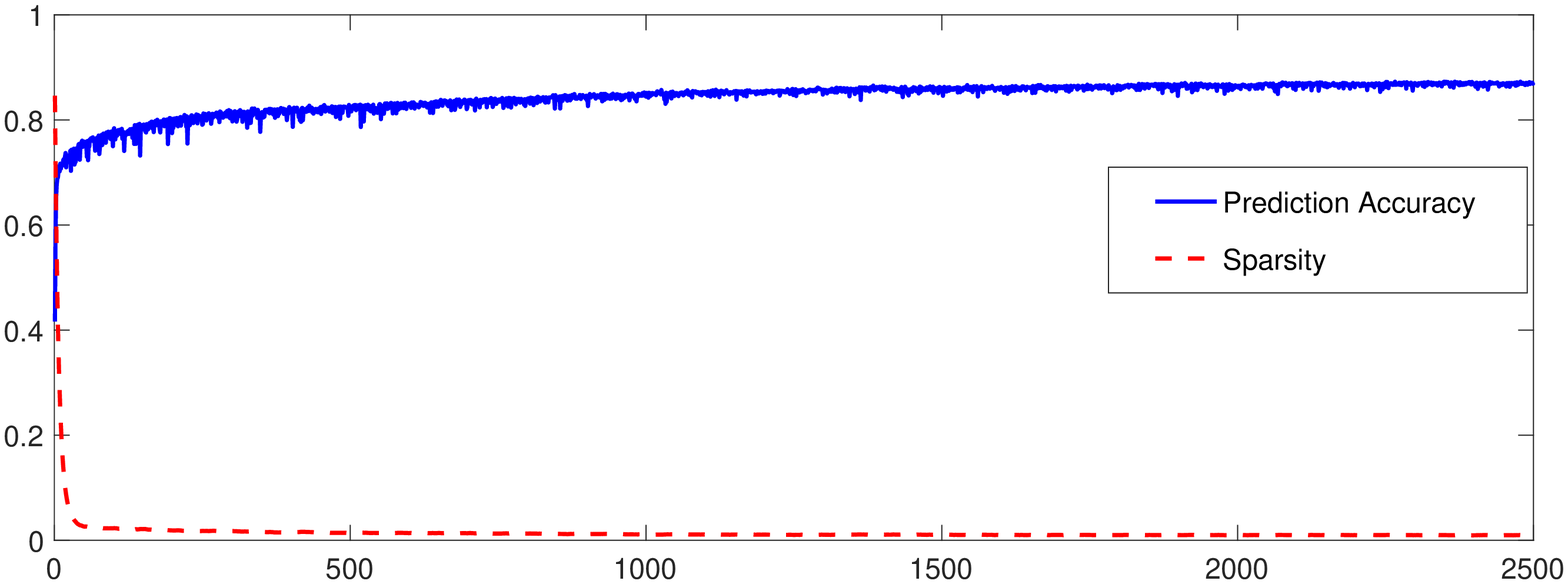}
}{%
  \caption{Prediction accuracy on test sample and sparsity along MCMC run}\label{fig:lenet_sum}%
}
\capbtabbox{%
  \begin{tabular}{lll} \hline
  &Sparsity & Accuracy \\ \hline
  SA-SGLD& 1.00 (0.00) & 86.5 (0.45) \\
  MC Dropout& 100 & 88.04 (0.02) \\ \hline
  \end{tabular}
}{%
  \caption{Estimated posterior sparsity and prediction accuracy on test sample (in percentage)}\label{table:lenet_sum}%
}
\end{floatrow}
\end{figure}

\begin{center}
\begin{figure}
	\includegraphics[scale = 0.3]{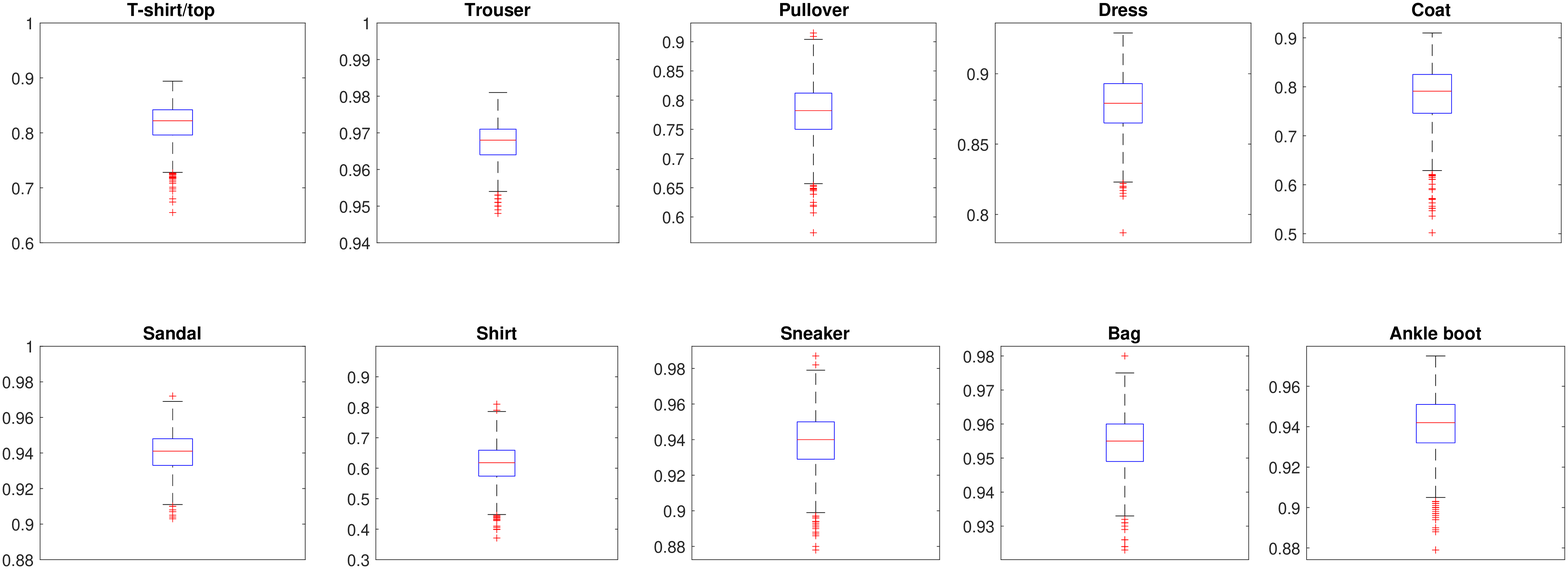}
	\caption{Distribution of posterior predictive accuracy on MNIST-Fashion test samples. Averaged within each class of item.}\label{fig:pred_error_mnist}
\end{figure}
\end{center}

\section{Some concluding remarks}\label{sec:conclusion}
We proposed in this work a fast MCMC algorithm for the Bayesian analysis of sparse high-dimensional models.  The algorithm operates as a form of Bayesian iterated sure independent screening, resulting in tremendous speed. In linear regression models we show that the algorithm mixes quickly to a limiting distribution that recovers correctly the main underlying signal. In limited sample size problems  the algorithm can be advantageously combined with tempering techniques (such as simulated tempering or related ideas) for better mixing properties. Such extensions could also be particularly useful in deep learning where the resulting posterior distributions are known to be highly multimodal.  One interesting aspect of the theoretical analysis done in this work (and  that extends from \cite{desa:16}),  is the use of a metric weaker than the total variation metric and  more directly pertinent for the statistical analysis, to measure MCMC mixing.  Exploring more systematically this idea could be an important theoretical contribution to the literature, particularly in high-dimensional problems.

\appendix
\section{Proof of Proposition \ref{prop:1}}\label{sec:proof:prop:1}
Throughout $C_0$  denotes a generic constant whose value may change from one appearance to the next. The Markov kernel of Algorithm \ref{algo:2} writes
\[\tilde K((\delta,\theta);(\rmd\delta',\rmd\theta')) = \tilde K_\delta(\theta,\rmd\theta') \sum_{\mathsf{J}:\;|\mathsf{J}|=J}{p\choose  J}^{-1} \tilde Q_{\theta',\mathsf{J}}(\delta,\rmd\delta'),\]
where 
\[\tilde K_\delta(\theta,\rmd \theta') \eqdef \tilde P_\delta([\theta]_\delta,\rmd [\theta']_\delta) \prod_{j:\;\delta_j=0} \textbf{N}(0,\rho_0^{-1})(\rmd \theta_j').\]
Recall that $V(\delta,\theta) =  \sum_{j}\delta_j V_j(\theta_j)$. Given a selection $\mathsf{J}=\{j_1,\ldots,j_J\}\subseteq \{1,\ldots,p\}$, and $j_i\in\mathsf{J}$, we  have
\begin{multline*}
\int_\Delta \tilde Q_{\theta,j_i}(\delta,\rmd\delta')V(\delta',\theta)  = V(\delta,\theta)  + \tilde q_{j_i}V_{j_i}(\theta_{j_i}) - \delta_{j_i}V_{j_i}(\theta_{j_i}) \leq V(\delta,\theta) +(1 - \delta_{j_i}) V_{j_i}(\theta_{j_i}),
\end{multline*}
where $\tilde q_j = \tilde q_j(\vartheta,\theta)$. It follows that
\begin{equation}\label{control:QJ}
\int_\Delta \tilde Q_{\theta,\mathsf{J}}(\delta,\rmd\delta') V(\delta',\theta)  \leq 
  V(\delta,\theta) +\sum_{i=1}^J  V_{j_i}(\theta_{j_i})\textbf{1}_{\{\delta_{j_i} = 0\}}.\end{equation}
Note that in deriving (\ref{control:QJ}) we did not use any specific information about the probability $\tilde q_j$. In particular the kernel $Q_{\theta,\mathsf{J}}$ also satisfies (\ref{control:QJ}). Using (\ref{control:QJ}) we have
\begin{multline*}
\int_{\rset^p} \tilde K_\delta(\theta,\rmd\theta')\int_\Delta \tilde Q_{\theta',\mathsf{J}}(\delta,\rmd \delta') V(\delta',\theta')\\
\leq \int_{\rset^{\|\delta\|_0}} \tilde P_\delta([\theta]_\delta,\rmd u)V_\delta(u) + \sum_{i = 1}^J  \int_\rset V_{j_i}(x) \textbf{N}(0,\rho_0^{-1})(\rmd x)\\
\leq \lambda_\delta V(\delta,\theta) + b_\delta + C_0 J,
\end{multline*}
where the first inequality uses the fact that under $\tilde K_\delta$, when $\delta_j=0$ we update $\theta_j$ by drawing from $\textbf{N}(0,\rho_0^{-1})$. With $\lambda = \max_\delta \lambda_\delta$, we conclude that
\begin{equation}\label{eq:drift:tilde:proof}
\int_{\Delta\times\rset^p} \tilde K((\delta,\theta);(\rmd\delta',\rmd\theta'))V(\delta',\theta') \leq \lambda V(\delta,\theta) + C_0.\end{equation}
Furthermore, $\tilde K$ is phi-irreducible and aperiodic by assumption, and the level sets $\{(\delta,\theta):\;  V(\delta,\theta) \leq b\}$ are petite sets for $\tilde K$. Therefore, by Lemma 15.2.8, and  Theorem 15.0.1 of \cite{meyn:tweedie:08} $\tilde K$ admits a unique invariant distribution $\tilde \Pi$, and (\ref{geo:ergo:tilde}) holds. 
\vspace{-0.5cm}
\begin{flushright}
$\square$
\end{flushright}

\section{Proof of Theorem \ref{thm:2}}\label{sec:proof:thm:2}
Throughout $C_0$  denotes a generic constant whose value may change from one appearance to the next.  We shall write $\Pi(\cdot)$ and $\tilde \Pi(\cdot)$ instead of $\Pi(\cdot\vert \D)$ and $\tilde\Pi(\cdot\vert \D)$ respectively. For any two Markov kernels $P_1$ and $P_2$ and for any integer $k\geq  1$, it is easily checked that $P_1^k = P_2^k + \sum_{j=1}^kP_1^{k-j}(P_1-P_2)P_2^{j-1}$. Using this identity, for any  bounded measurable function $f:\;\Delta\times \rset^p\to\rset$, writing $\tilde f  = f-\tilde\Pi(f)$, we have for any $k\geq 0$,
\begin{multline}\label{thm2:basic:eq}
\Pi(f) - \tilde\Pi(f) = \Pi(\tilde f)  =  \Pi (K^k\tilde f)= \Pi(\tilde K^k\tilde  f) +\sum_{j=1}^k \Pi\left((K-\tilde K)\tilde K^{j-1}\tilde  f\right),\\
=\Pi(\tilde K^k\tilde  f) +\Pi\left[(K-\tilde K)\left(\sum_{j=0}^{k-1} \tilde K^{j}\tilde  f\right)\right].\end{multline}
Define
\[g_k\eqdef\sum_{j = 0}^{k-1}\tilde K^j \tilde f.\]
 It follows from (\ref{geo:ergo:tilde}) that for all $(\delta,\theta)\in\Delta\times\rset^p$,
 \begin{equation}\label{eq:bound:gk}
 |g_k(\delta,\theta)| \leq C_0 \|f\|_\infty V^{1/2}(\delta,\theta)\sum_{j=0}^{k-1}\tilde\lambda^j \leq \frac{C_0\|f\|_\infty}{1-\tilde\lambda} V^{1/2}(\delta,\theta).\end{equation}
Without any loss generality we shall assume now that $\|f\|_\infty = 1$. Since $\Pi(V)<\infty$, we get that 
\begin{equation}\label{thm2:eq1}
\left|\Pi(f) - \tilde\Pi(f)\right| \leq C_0\tilde\lambda^k + \int_{\Delta\times\rset^p}\Pi(\rmd \delta,\rmd\theta) \left| Kg_k(\delta,\theta) - \tilde Kg_k(\delta,\theta)\right|.\end{equation}
We split the integral  over $\cB$ and over $\cB^c$.  For the part over $\cB^c$, we use the Cauchy-Schwarz inequality, (\ref{eq:bound:gk}) and (\ref{eq:drift:tilde:proof}) to write
\begin{multline*}
\left|\int_{\cB^c}\Pi(\rmd \delta,\rmd\theta)\tilde Kg_k(\delta,\theta)\right| \leq \Pi(\cB^c)^{1/2}\sqrt{\int\Pi(\rmd \delta,\rmd\theta)|\tilde Kg_k(\delta,\theta)|^2}\leq \Pi(\cB^c)^{1/2}\frac{C_0}{1-\tilde\lambda}\Pi(V\vert \mathcal{D})^{1/2}.
\end{multline*}
Similarly,
\begin{multline*}
\left|\int_{\cB^c}\Pi(\rmd \delta,\rmd\theta)Kg_k(\delta,\theta)\right| \leq \Pi(\cB^c)^{1/2}\sqrt{\int\Pi(\rmd \delta,\rmd\theta)g_k^2(\delta,\theta)}\leq \Pi(\cB^c)^{1/2}\frac{C_0}{1-\tilde\lambda}\Pi(V\vert \mathcal{D})^{1/2}.\end{multline*}
It follows that
\begin{equation}\label{eq:thm1:bound:dif:1}
\left|\Pi(f) - \tilde\Pi(f)\right|  \leq C_0\tilde\lambda^k +  \frac{C_0\sqrt{\Pi(\cB^c)}}{1-\tilde\lambda}
+ \int_{\cB}\Pi(\rmd \delta,\rmd\theta) \left| Kg_k(\delta,\theta) - \tilde Kg_k(\delta,\theta)\right|.
 \end{equation}
  We will use the following simple version of the coupling inequality (\cite{lindvall92}).
 \begin{lemma} \label{lem:coupling}Let $\mu,\nu$, be two probability measures on some Polish space $(\Xset,\cB)$,  and $f:\;\Xset\to\rset$ a measurable function. Then
 \[\left|\mu(f) - \nu(f)\right| \leq \|\mu-\nu\|_\tv^{1/2} \left(\sqrt{\int f^2\rmd\mu} + \sqrt{\int f^2\rmd \nu}\right).\]
 \end{lemma}
 \begin{proof}
 Let $(X,Y)$ be a maximal coupling of $(\mu,\nu)$. Then
 \[\left|\mu(f) - \nu(f)\right| = \left|\PE\left[\textbf{1}_{\{X\neq Y\}}\left(f(X) - f(Y)\right)\right]\right|\leq \|\mu-\nu\|_\tv^{1/2} \PE^{1/2}\left((f(X) - f(Y))^2\right).\]
 The result then follows from the Minkowski inequality.
 \end{proof}
 We will apply this lemma in the following context. Let $\bar\mu(\rmd x,\rmd y) = \mu_1(\rmd x)\mu_2(x,\rmd y)$, $\bar\nu(\rmd x,\rmd y) = \nu_1(\rmd x)\nu_2(x,\rmd y)$ be two probability measures on some Polish space $\Xset\times\Yset$, and let $f:\;\Xset\times\Yset\to\rset$ be a measurable function. Then by writing
 \begin{multline*}
 \bar\mu(f) -\bar\nu(f) = \left(\int_\Xset \mu_1(\rmd x) f_\mu(x) -\int_\Xset\nu_1(\rmd x)f_\mu(x)\right)  \\
 +  \left(\int_\Xset\nu_1(\rmd x)\int_\Yset \mu_2(x,\rmd y) f(x,y) - \int_\Xset\nu_1(\rmd x) \int_\Yset \nu_2(x,\rmd y)f(x,y)\right),\end{multline*}
 where $f_\mu(x) \eqdef \int_\Yset \mu_2(x,\rmd y)f(x,y)$, we deduce from Lemma \ref{lem:coupling} that if $C<\infty$ is such that $\int_\Xset\int_\Yset f^2(x,y)\mu_1(\rmd x)\mu_2(x,\rmd y)\leq C^2$, $\int_\Xset\int_\Yset f^2(x,y)\nu_1(\rmd x)\nu_2(x,\rmd y)\leq C^2$, and $\int_\Xset\int_\Yset f^2(x,y)\nu_1(\rmd x)\mu_2(x,\rmd y)\leq C^2$, then
 \begin{equation}\label{eq:comp}
 |\bar\mu(f) - \bar\nu(f)| \leq C\left(\|\mu_1-\nu_1\|_\tv^{1/2} + \left\|\int\nu_1(\rmd x)\mu_2(x,\cdot) - \int\nu_1(\rmd x)\nu_2(x,\cdot)\right\|_\tv^{1/2}\right).\end{equation}
We apply (\ref{eq:comp}) to $Kg_k(\delta,\theta) - \tilde Kg_k(\delta,\theta)$. Indeed, for $(\delta,\theta)\in\cB$, using (\ref{control:KV}), we have 
\[K g_k^2(\delta,\theta) \leq \frac{C_0^2}{(1-\tilde\lambda)^2}K V(\delta,\theta) \leq \frac{C_0^2}{(1-\tilde\lambda)^2},\;\;\;  \tilde K g_k^2(\delta,\theta) \leq  \frac{C_0^2}{(1-\tilde\lambda)^2},\]
and since $Q_{\theta,\mathsf{J}}$ also satisfies (\ref{control:QJ}), we similarly have
\[ \int\tilde K_\delta(\theta,\rmd\theta') \sum_{\mathsf{J}:\;|\mathsf{J}|=J}{p\choose  J}^{-1} \int_\Delta Q_{\theta',\mathsf{J}}(\delta,\rmd\delta')g_k^2(\delta',\theta') \leq\frac{C_0^2}{(1-\tilde\lambda)^2}. \]
Therefore (\ref{eq:comp}) and (\ref{eq:thm1:bound:dif:1}) yield
\begin{multline}\label{eq:thm1:bound:dif:2}
\left|\Pi(f) - \tilde\Pi(f)\right|  \leq C_0\tilde\lambda^k  +  \frac{C_0\sqrt{\Pi(\cB^c)}}{1-\tilde\lambda} +\frac{C_0}{1-\tilde\lambda} \sup_{(\delta,\theta)\in\cB}\|P_\delta([\theta]_\delta,\cdot) - \tilde P_\delta([\theta]_\delta,\cdot)\|_\tv^{1/2} \\
+  \frac{C_0}{1-\tilde\lambda} \sup_{(\delta,\theta)\in\cB} \left\| \int\tilde K_\delta(\theta,\rmd\theta') \sum_{\mathsf{J}:\;|\mathsf{J}|=J}{p\choose  J}^{-1} \left(Q_{\theta',\mathsf{J}}(\delta,\cdot) - \tilde Q_{\theta',\mathsf{J}}(\delta,\cdot)\right)\right\|_\tv^{1/2}.
 \end{multline}
The result then follows by taking $k\to\infty$.
\vspace{-0.5cm}
\begin{flushright}
$\square$
\end{flushright}

\section{Proof of Theorem \ref{thm:3}}\label{sec:proof:thm:3}
We recall that $\PP$ and $\PE$ denote the probability measure and expectation operator of the Markov chains defined by Algorithms \ref{algo:1} and {\ref{algo:2} (more specifically their coupling distribution as constructed below), and  $\PP_\star$ and $\PE_\star$ denote the probability measure and expectation operator of the data $Y$ as assumed in H\ref{H:post:contr}. 

Throughout we will use $C$ to denote a generic constant that depends only on the constants appearing in H\ref{H:post:contr} ($\sigma^2,\|\theta_\star\|_\infty$, $c_0,c_1$ and $c_2$). The actual value of $C$ may vary from one appearance to the next.

We use a similar argument as  in \cite{desa:16}. Let $\{\delta^{(k)},\;k\geq 0\}$ denote the $\delta$-marginal chain of Algorithm \ref{algo:2}, and let $\{\check\delta^{(k)},\;k\geq 0\}$ be the $\delta$-marginal chain of Algorithm \ref{algo:1}. These processes are also Markov chains because in both cases we have taken $P_\delta=\tilde P_\delta$  to be an exact draw from the posterior conditional distribution of $\theta$ given  $\delta$. We construct a coupling of $\{\delta^{(k)},\;k\geq 0\}$ and the stationary version of $\{\check\delta^{(k)},\;k\geq 0\}$ as follows. First take $\delta^{(0)}$ as the null model, and draw $\check\delta^{(0)}\sim \Pi(\cdot\vert \D)$, the marginal distribution of $\delta$  in (\ref{post:Pi}). For each $k\geq 0$, given $(\delta^{(k)},  \check\delta^{(k)})$, we do the following.
\begin{enumerate}
\item Given $\delta^{(k)},  \check\delta^{(k)}$, we independently  draw $\theta^{(k)}\sim \Pi(\cdot\vert \delta^{(k)},\D)$, $\check\theta^{(k)}\sim \Pi(\cdot\vert \check\delta^{(k)},\D)$, and  we select a random subset $\mathsf{J}^{(k)}= \{\mathsf{J}_1^{(k)},\ldots,\mathsf{J}_J^{(k)}\}$ of size $J$  from $\{1,\ldots,p\}$.
\item  We define $\vartheta\in\Delta$ as $\vartheta_i = 0$ if $i\in\mathsf{J}^{(k)}$, and $\vartheta_i = \delta^{(k)}_i$ otherwise. We also define  $\vartheta^{(0)} = \delta^{(k)}$, and $\check\vartheta^{(0)} = \check\delta^{(k)}$. For  each $r\in\{1,\ldots,J\}$, given  $\mathsf{J}^{(k)}_r = j$, we then do the following.
\begin{enumerate}
\item  We draw $(d_r^{(k)},\check d_r^{(k)})$ from the maximal coupling of $\textsf{Ber}(\tilde q_{j}(\vartheta,\theta^{(k)}))$ and $\textsf{Ber}(q_{j}(\check\vartheta^{(r-1)},\check\theta^{(k)}))$, where $q_j$ and $\tilde q_j$ are given by (\ref{cond:dist:eq:1}) and (\ref{cond:dist:eq:2}) respectively.  
\item We set $\vartheta^{(r)}_{j} = d_r^{(k)}$, $\check\vartheta_{j}^{(r)} =\check d_r^{(k)}$, and $\vartheta_{i}^{(r)} =\vartheta_{i}^{(r-1)}$, $\check\vartheta_{i}^{(r)} =\check\vartheta_{i}^{(r-1)}$, for $i\neq j$.
\end{enumerate}
\item Set $\delta^{(k+1)} = \vartheta^{(J)}$, and 
$\check\delta^{(k+1)} = \check\vartheta^{(J)}$.
\end{enumerate}

By construction, the marginal chain $\{\delta^{(k)}\;k\geq 0\}$ (resp.  $\{\check\delta^{(k)}\;k\geq 0\}$) from the above construction is the asynchronous sampler from Algorithm \ref{algo:2} (resp. a stationary version of Algorithm \ref{algo:1}). By the coupling inequality
\begin{equation}\label{coupling:ineq}
\PE_\star\left[\max_{j:\;\delta_{\star j}=1}\;|\PP(\delta_j^{(k)}=1) -\Pi(\delta_j=1\vert \D)|\right]  \leq \PE_\star\left[\max_{j:\;\delta_{\star j}=1}\;\PP(\delta^{(k)}_j\neq \check\delta^{(k)}_j)\right].\end{equation}
Hence the main part of the proof consists in bounding the right-hand side of the last display.  We do this in paragraph (e). Paragraphs (a)-(d) collect some needed implications of H\ref{H:post:contr}.

\paragraph{\texttt{(a) Restricted eigenvalues}}\;\; Given $s>0$, Let $\delta\in\Delta$ be such that $0<\|\delta\|_0\leq s$, and let $u\in\rset^{\|\delta\|_0}$. Using $|\pscal{X_i}{X_j}|\leq \sqrt{c_0n\log(p)}$ from H\ref{H:post:contr}-(1), and $\|X_j\|_2=\sqrt{n}$, we have
\[u'(X_\delta'X_\delta) u \geq n\|u\|_2^2 -  \sqrt{c_0n\log(p)}\sum_{i\neq j}|u_iu_j|\geq \left(n-s\sqrt{c_0n\log(p)}\right)\|u\|_2^2.\]
We conclude that if the sample size satisfies $n\geq 4c_0s^2\log(p)$, then
\begin{equation}\label{rsc}
\lambda_{\min}\left(X_\delta'X_\delta\right) \geq \frac{n}{2},\;\mbox{ for all } \delta\in\Delta,\;\mbox{ s. t. }\; 0<\|\delta\|_0\leq s,
\end{equation}
where $\lambda_{\min}(A)$ denotes the smallest eigenvalue of $A$.

\paragraph{\texttt{(b) Implications of the sub-Gaussian regression errors}}\;\; For $\delta\in\Delta$, we define
\[L_\delta \eqdef I_n + \frac{1}{\sigma^2}X_\delta X_\delta'.\]
We convene that $L_\delta = I_n$, for $\delta=0$. Clearly, $\|L_\delta^{-1}\|_2\leq1$. Given $s\geq 0$ (here we allow $s$ to be 0), and for some constant $c>0$, we set
\[\e_s\eqdef\left\{y\in\rset^n:\; \max_{1\leq j\leq p}\;\max_{\delta:\;\|\delta\|_0\leq s}\; \sigma^{-1}\left|\pscal{L_\delta^{-1}X_j}{y-X\theta_\star}\right| \leq \sqrt{c(1+s)n\log(p)}\right\}.\]
Using the sub-Gaussianity of the regression error term in H\ref{H:post:contr}-(2), and by a union bound argument, we can choose $c>0$ depending solely on the absolute constant $c_1$ in H\ref{H:post:contr}-(2), such that for all $s\geq 0$, 
\begin{equation}\label{bound:e}
\PP_\star\left(Y\notin\e_s\right) \leq p{p\choose s}c_1 \exp\left(-\frac{c(1+s)n\log(p)}{2n}\right) \leq \frac{c_1p^{s+1}}{p^{\frac{c(1+s)}{2}}}\leq \frac{1}{p},\end{equation}
where we use the fact that ${p\choose s}\leq p^s$. Throughout the proof, whenever we use the event $\e_s$, the constant $c$ is always taken as above.

\paragraph{\texttt{(c) Sparse MCMC output}}\;\; It will be important in the proof to guarantee that the Markov chain $\{\delta^{(k)},\;k\geq 1\}$ remains in the set $\Delta_s\eqdef\{\delta\in\Delta:\;\|\delta\|_0\leq s\}$ for some small value of $s$. The following result could probably be improved, but will serve the purpose.  Let 
\[s_1 \eqdef s_\star+2\log(p).\]
We show in Lemma \ref{lem:moment} that under the sample size condition (\ref{ss:cond}), and $\mathsf{u}$ taken large enough as in (\ref{ss:cond}), it holds
\begin{equation}\label{eq:sparse:mcmc}
\textbf{1}_{\e_0}(Y)\max_{k\geq 0}\;\PP\left(\|\delta^{(k)}\|_0 > s_1\right) \leq   \frac{1}{p}.\end{equation}

\paragraph{\texttt{(d) Posterior contraction}}\;\; We show  below that the posterior distribution $\Pi(\cdot\vert\D)$ puts most probability mass on sparse super-sets of $\delta_\star$. More precisely, by Lemma \ref{post:contr} we can find constants $C_1, C_2$ that depends only on the constants appearing in H\ref{H:post:contr} ($\sigma^2,\|\theta_\star\|_\infty$, $c_0,c_1$ and $c_2$)  such that for $n,p$ such that $n\geq C_1 (1 + s_\star^3)\log(p)$, it holds
\[\PE_\star\left[\textbf{1}_{\e_{s_\star}}(Y) \Pi(\Cset\vert \D)\right]\geq 1-\frac{2}{p},\]
where 
\[\Cset \eqdef \left\{\delta\in\Delta:\;\delta\supseteq \delta_\star,\;\mbox{ and } \; \|\delta\|_0\leq C_2 (1+s_\star)\right\}.\]
We set 
\[s_2 \eqdef C_2 (1+s_\star).\]
Furthermore the  linear regression setting implies  that the conditional posterior distribution of $\theta \vert \delta$ is given by
\begin{equation}\label{cond:dist:theta:lm}
[\theta]_{\delta^c}\;\vert\;\delta\; \stackrel{i.i.d.}{\sim} \textbf{N}(0,\rho_0^{-1}),\;\;\mbox{ and }\;\; [\theta]_\delta\;\vert\;\delta\sim \textbf{N}\left(\hat\theta_\delta,\sigma^2\left(\sigma^2\rho_1 I_{\|\delta\|_0} + X_{\delta}'X_\delta\right)^{-1}\right),\end{equation}
where
\[\hat\theta_\delta \eqdef \argmax_{u\in\rset^{\|\delta\|_0}} \;\left[-\frac{1}{2\sigma^2}\|y - X_\delta u\|_2^2 -\frac{\rho_1}{2}\|u\|_2^2\right] = \left(X_\delta'X_\delta + \rho_1\sigma^2 I_{\|\delta\|_0}\right)^{-1} X_\delta'y.\]
Therefore, if for some $M>0$ we set 
\begin{multline*}
\cB_{\delta} \eqdef \left\{\theta\in\rset^p:\; \|\theta-\theta_\delta\|_\infty \leq  \sqrt{\frac{M\log(p)}{\rho_0}} \;\;\mbox{ and }\;\; \|\theta_\delta - \hat\theta_\delta\|_\infty \leq \sqrt{\frac{M\sigma\log(p)}{n}}\right\},\end{multline*}
then, provided that $n\geq Cs^2\log(p)$, for some constant $C$,  by the restricted eigenvalue bound in (\ref{rsc}), and by Gaussian tail bounds and a union bound argument, for all $\delta\in\Delta_s$, we have
\begin{equation}\label{proof:asynch:bound:B}
\Pi\left(\cB_\delta^c\vert \delta,\D\right) \leq \frac{4}{p^{-1+M/2}}\leq \frac{1}{p},\end{equation}
by taking $M>2$ appropriately.

\paragraph{\texttt{(e) Main arguments of the proof}}\;\;
With $s_1$ as in Paragraph (c) and $s_2$ as in Paragraph (d), we set
\[s \eqdef \max(s_1,s_2).\]
Fix $Y\in\e_{s_\star}$, and fix some arbitrary component $j$ such that $\delta_{\star j}=1$. We first note that $\delta_j^{(k+1)}\neq \check\delta^{(k+1)}_j$ if and only if $j\in\mathsf{J}^{(k)}$, and the corresponding Bernoulli's $(d_r^{(k)},\check d_r^{(k)})$ are different, or $\delta_j^{(k)}\neq \check\delta_j^{(k)}$, and $j\notin\mathsf{J}^{(k)}$. We write this as
\begin{multline*}
\PP\left[\delta_j^{(k+1)} \neq \check\delta_j^{(k+1)} \vert \delta^{(k)},\check\delta^{(k)}\right] \\
= \textbf{1}_{\left\{ \delta_j^{(k)} \neq \check\delta_j^{(k)}\right\}}  \left(1- \frac{J}{p}\right) +  \sum_{r=1}^J \PP\left[\mathsf{J}_r^{(k)} = j, d_r^{(k)} \neq \check d_r^{(k)}\vert \delta^{(k)},\check\delta^{(k)}\right]\\
= \textbf{1}_{\left\{ \delta_j^{(k)} \neq \check\delta_j^{(k)}\right\}}  \left(1- \frac{J}{p}\right) +  \frac{1}{p}\sum_{r=1}^J \PP\left[d_r^{(k)} \neq \check d_r^{(k)}\vert \mathsf{J}_r^{(k)} = j, \delta^{(k)},\check\delta^{(k)}\right],\end{multline*}
where we use the fact that $\PP(\mathsf{J}_r^{(k)}=j\vert \delta^{(k)},\check\delta^{(k)})=1/p$.  With $s_1,s_2$ as above, we  introduce the set $\mathbb{T}\eqdef \Delta_{s_1}\times \Cset_{s_2}$ where,
\[\Delta_{s_1}\eqdef\left\{\delta\in\Delta:\;\|\delta\|_0\leq s_1\right\},\;\;\mbox{ and }\;\;\Cset_{s_2} \eqdef \left\{\delta\in\Delta:\;\delta\supseteq\delta_\star,\;\|\delta\|_0\leq s_2\right\}.\]
It follows that
\begin{multline*}
\PP\left[\delta_j^{(k+1)} \neq \check\delta_j^{(k+1)} \vert \delta^{(k)},\check\delta^{(k)}\right] \\
\leq  \textbf{1}_{\left\{ \delta_j^{(k)} \neq \check\delta_j^{(k)}\right\}}\left(1- \frac{J}{p}\right)  + \textbf{1}_{\mathbb{T}}(\delta^{(k)}, \check\delta^{(k)}) \frac{1}{p} \sum_{r=1}^J \PP\left[d_r^{(k)} \neq \check d_r^{(k)}\vert \mathsf{J}_r^{(k)} = j, \delta^{(k)},\check\delta^{(k)}\right] \\
+ \frac{J}{p}\textbf{1}_{\mathbb{T}^c}(\delta^{(k)}, \check\delta^{(k)}).
\end{multline*}
Let us set
\[A^{(k)} \eqdef \PP\left((\delta^{(k)}, \check\delta^{(k)})\notin \mathbb{T}\right),\;\;\mathcal{I}_{r,j}^{(k)}(\theta,\check\theta) \eqdef \PP\left(d_r^{(k)} \neq \check d_r^{(k)}\vert \mathsf{J}_r^{(k)}=j,\delta^{(k)},\check\delta^{(k)},\theta,\check\theta\right).\]
Taking expectation on both sides of the last inequality, we get
\begin{eqnarray}\label{proof:eq1}
\PP\left(\delta_j^{(k+1)} \neq \check\delta_j^{(k+1)}\right) &  \leq &  \left(1- \frac{J}{p}\right)\PP\left(\delta_j^{(k)} \neq \check\delta_j^{(k)}\right) + \frac{J A^{(k)}}{p}\nonumber\\
&& + \frac{1}{p}\sum_{r=1}^J \PE\left[\textbf{1}_{\mathbb{T}}(\delta^{(k)}, \check\delta^{(k)})\PP\left(d_r \neq \check d_r\vert \mathsf{J}_r^{(k)}=j,\delta^{(k)},\check\delta^{(k)}\right)\right],\nonumber\\
& = & \left(1- \frac{J}{p}\right)\PP\left(\delta_j^{(k)} \neq \check\delta_j^{(k)}\right) +  \frac{J A^{(k)}}{p}\\
&&+ \frac{1}{p}\sum_{r=1}^J \PE\left[\textbf{1}_{\mathbb{T}}(\delta^{(k)}, \check\delta^{(k)})\int\mathcal{I}_{r,j}^{(k)}(\theta,\check\theta)\,\Pi(\rmd\theta\vert \delta^{(k)},\D)\,\Pi(\rmd\check\theta\vert\check\delta^{(k)},\D)\right]\nonumber.
\end{eqnarray}
We establish the following claim below
\begin{multline}\label{big:claim}
\textbf{1}_{\mathbb{T}}(\delta^{(k)}, \check\delta^{(k)})\int\mathcal{I}_{r,j}^{(k)}(\theta,\check\theta)\,\Pi(\rmd\theta\vert \delta^{(k)},\D)\,\Pi(\rmd\check\theta\vert\check\delta^{(k)},\D)\\
 \leq \left(e^{-C \sqrt{n}\underline{\theta}_\star} + \frac{1}{p}\right) + \frac{7}{10} \textbf{1}_{\{\delta_j^{(k)}\neq \check\delta_j^{(k)}\}} .
\end{multline}
Using (\ref{big:claim}) in (\ref{proof:eq1}), we obtain
\begin{multline}\label{proof:eq3}
\PP\left(\delta_j^{(k+1)} \neq \check\delta_j^{(k+1)}\right) \leq \left(1- \frac{J}{p}\right)\PP\left(\delta_j^{(k)} \neq \check\delta_j^{(k)}\right) + \frac{J A^{(k)}}{p} \\
+  \frac{J}{p}\left(e^{-C \sqrt{n}\underline{\theta}_\star} + \frac{1}{p}\right) + \frac{7}{10}\frac{J}{p} \PP\left(\delta_j^{(k)} \neq \check\delta_j^{(k)}\right)\\
\leq \left(1- \frac{3}{10}\frac{J}{p}\right)\PP\left(\delta_j^{(k)} \neq \check\delta_j^{(k)}\right) +  \frac{J}{p}\left(A^{(k)} +e^{-C \sqrt{n}\underline{\theta}_\star} + \frac{1}{p}\right).
\end{multline}
Iterating (\ref{proof:eq3}) yields
\begin{multline}\label{proof:eq4}
\max_{j:\;\delta_{\star j}=1}\;\PP\left(\delta_j^{(k)} \neq \check\delta_j^{(k)}\right) \leq \left(1- \frac{3}{10}\frac{J}{p}\right)^k  + \frac{10}{3}\left(e^{-C \sqrt{n}\underline{\theta}_\star} + \frac{1}{p}\right) \\
+ \frac{J}{p} \sum_{t=0}^{k-1}\left(1- \frac{3}{10}\frac{J}{p}\right)^tA^{(k-t)}.
\end{multline}
Recall that
\[A^{(k)} \eqdef \PP\left((\delta^{(k)}, \check\delta^{(k)})\notin \mathbb{T}\right)\leq \PP(\|\delta^{(k)}\|_0>s_1) + \Pi(\Cset_{s_2}^c\vert \D),\]
where $\Cset_{s_2}^c\eqdef \Delta\setminus\Cset_{s_2}$. By Lemma \ref{lem:moment} and Lemma \ref{post:contr} below, we have
\[\textbf{1}_{\e_{s_\star}}(Y)\max_{k\geq 0}  \PP(\|\delta^{(k)}\|_0>s_1) \leq \frac{1}{p},\;\mbox{ and }\;\; \PE_\star\left[\textbf{1}_{\e_{s_\star}}(Y)\Pi(\Cset_{s_2}^c\vert \D)\right]\leq \frac{1}{p} .\]

Taking the expectation over the data $Y$ in (\ref{proof:eq4}), and using the last display and (\ref{coupling:ineq}), we deduce that
\begin{multline}\label{proof:eq5}
\PE_\star\left[\textbf{1}_\e(Y)\max_{j:\;\delta_{\star j}=1}\;\;\left|\mathbb{P}(\delta_j^{(k)}=1) -\Pi(\delta_j=1\vert \D)\right|\right] \leq \left(1- \frac{3}{10}\frac{J}{p}\right)^k  \\
+ \frac{10}{3}\left(e^{-C \sqrt{n}\underline{\theta}_\star} + \frac{1}{p} + \frac{2}{p}\right)\\
\leq \left(1- \frac{3}{10}\frac{J}{p}\right)^k + 10\left(e^{-C \sqrt{n}\underline{\theta}_\star} + \frac{1}{p}\right).
\end{multline}
It remains only to  establish the claim (\ref{big:claim}). 

\paragraph{\texttt{Proof of Claim (\ref{big:claim})}}\;\; We consider two cases.

\paragraph{\underline{\texttt{Case 1}: $\delta_j^{(k)}\neq \check\delta_j^{(k)}$}}\;\;  Since $\check\delta^{(k)}\in\Cset_{s_1}$ (which implies that  $\check\delta^{(k)}_j=1$), we must then have $\delta_j^{(k)}=0$, and $\check\delta_j^{(k)}=1$. Set
\[\mathbb{S}\eqdef\left\{(\theta,\check\theta)\in\rset^p\times \rset^p:\; \theta\in\cB_{\delta^{(k)}}, \check\theta\in\cB_{\check\delta^{(k)}},\;\; \mbox{ and } \;\;\sqrt{\frac{1}{100\rho_0}} \leq \textsf{sign}(\theta_{\star j})\theta_{j} \leq \sqrt{\frac{4}{\rho_0}} \right\}.\]
It follows from (\ref{proof:asynch:bound:B}) and the fact that $\theta_j^{(k)}\vert \{\delta^{(k)}_j=0\}\sim\textbf{N}(0,\rho_0^{-1})$ that for $(\delta^{(k)},\check\delta^{(k)})\in \mathbb{T}$,
\begin{equation}\label{case1:eq1}
\PP\left((\theta^{(k)},\check\theta^{(k)})\notin\mathbb{S}\vert \delta^{(k)}, \check\delta^{(k)}\right)\leq \frac{4}{p^{M/2}} + \frac{3}{5}\leq \frac{7}{10}.\end{equation}
First we note that for  $(\theta,\check\theta)\in\mathbb{S}$, $|\theta_j| \leq 2/\sqrt{\rho_0} \leq Cn^{-1/2}$. Whereas for $i\neq j$ and $\delta_i^{(k)}=0$, we have $|\theta_i|\leq \sqrt{M\log(p)/n}$, and if $\delta_i^{(k)}=1$, using (\ref{proof:asynch:bound:B}), and Lemma \ref{lem:model:grad}-(1),
\[|\theta_i| = |\theta_i -\hat\theta_i| + |\hat\theta_i -\theta_{\star i}| + |\theta_{\star i}|\leq C\sqrt{\frac{\log(p)}{n}} + C\sqrt{\frac{\log(p)}{n}}+ \|\theta_\star\|_\infty \leq C,\] 
under the sample size condition (\ref{ss:cond}). Using the expression $\tilde q_j$ in (\ref{cond:dist:eq:2}), and since $\rho_0\geq \rho_1$, and ignoring the  nonpositive quadratic term, we have
\[1 - \tilde q_j(\vartheta,\theta) \leq \exp\left(\mathsf{a} -\frac{\theta_j}{\sigma^2}\pscal{X_j}{y- X\theta_{\vartheta}} \right).\]
 We write
\[\pscal{X_j}{y- X\theta_{\vartheta}} =  \left[\pscal{X_j}{X\theta_{\delta^{(k)}}} -\pscal{X_j}{X\theta_{\vartheta}} \right] +  \pscal{X_j}{y-X\theta_{\delta^{(k)}}}.\]
Since $|\theta_i|\leq C$, for $\delta_i^{(k)}=1$, we have
\begin{multline*}
\left|\pscal{X_j}{X\theta_{\vartheta}}  - \pscal{X_j}{X\theta_{\delta^{(k)}}}\right| = \left|\sum_{r\in\mathsf{J}^{(k)}:r\neq j,\;\delta^{(k)}_r=1} \;\;\pscal{X_j}{X_r}\theta_r\right|\\
 \leq C \min(J,s_1) \sqrt{n\log(p)}.\end{multline*}
We can rewrite the last display as
\[\left|\pscal{X_j}{y- X\theta_{\vartheta}} - \pscal{X_j}{y-X\theta_{\delta^{(k)}}}\right| \leq  C \min(J,s_1) \sqrt{n\log(p)}.\] 
We further expand the term $\pscal{X_j}{y-X\theta_{\delta^{(k)}}}$ as
\begin{multline*}
\pscal{X_j}{y-X\theta_{\delta^{(k)}}} = \pscal{X_j}{y-X\theta_\star} + \pscal{X_j}{X\theta_\star - X_{\delta^{(k)}}[\theta_\star]_{\delta^{(k)}}} \\
+  \pscal{X_j}{X_{\delta^{(k)}}([\theta_\star]_{\delta^{(k)}} - [\theta]_{\delta^{(k)}})}.\end{multline*}
Note that
\[\pscal{X_j}{X\theta_\star - X_\delta[\theta_\star]_{\delta^{(k)}}} = n\theta_{\star j} + \sum_{r:\;\delta_{\star r}=1,\delta_r^{(k)}=0}\theta_{\star r}\pscal{X_j}{X_r}.\]
  For $\theta\in\cB_{\delta^{(k)}}$, 
\begin{multline*}
\left|\pscal{X_j}{X_{\delta^{(k)}}([\theta_\star]_{\delta^{(k)}} - [\theta]_{\delta^{(k)}})}\right|\leq s_1\sqrt{c_0n\log(p)}\|[\theta_\star]_{\delta^{(k)}} - [\theta]_{\delta^{(k)}}\|_\infty\\
\leq Cs_1\sqrt{n\log(p)}\left(\|[\theta_\star]_{\delta^{(k)}} -\hat\theta_{\delta^{(k)}}\|_\infty + \sqrt{\frac{\log(p)}{n}}\right)\\
\leq C(1+m(\delta^{(k)})) s_1\log(p), \end{multline*}
where the last inequality uses Lemma \ref{lem:model:grad}. Using this, and since $\delta_{\star j}=1$, for $\theta\in\cB_{\delta^{(k)}}$ we have
\begin{multline}\label{grad:ineq}
\left|\pscal{X_j}{y-X\theta_{\vartheta}} - n\theta_{\star j} \right| \leq C\min(s_1,J)\sqrt{n\log(p)} + |\pscal{X_j}{y-X\theta_\star}| \\
+ \sum_{r:\;\delta_{\star r}=1,\delta_r^{(k)}=0}\left|\theta_{\star r}\pscal{X_j}{X_r}\right| + C(1+m(\delta^{(k)})) s_1\log(p) \\
\leq C\left(s_\star +\min(J,s_1)\right)\sqrt{n\log(p)},
\end{multline}
using the sample size condition $n\geq s_1^2\log(p)$. Since $|\theta_j|\leq C n^{-1/2}$, we conclude that
\[|\theta_j||\pscal{X_j}{y- X\theta_{\vartheta}}- n\theta_{\star j}| \leq  C\left(s_\star +\min(J,s_1)\right)\sqrt{\log(p)}.\]
It follows that  for $(\theta,\check\theta)\in\mathbb{S}$,
\begin{multline*}
\frac{\theta_j}{\sigma^2}\pscal{X_j}{y- X\theta_{\vartheta}} \geq \frac{n\theta_{\star j}\theta_j}{\sigma^2}  - C\left(s_\star +\min(J,s_1)\right)\sqrt{\log(p)} \\
\geq \frac{|\theta_{\star j}|\sqrt{n}}{10\sigma^2} - C\left(s_\star +\min(J,s_1)\right)\sqrt{\log(p)} \geq \frac{|\theta_{\star j}|\sqrt{n}}{20\sigma^2},\end{multline*}
under the sample size condition (\ref{ss:cond}). Hence, since $\mathsf{a} = \mathsf{u}\log(p) + \log(\rho_0)/2$, for $(\theta,\check\theta)\in\mathbb{S}$
\begin{equation}\label{case1:eq2}
1 - \tilde q_j(\vartheta,\theta) \leq \exp\left(\mathsf{u}\log(p) +\frac{1}{2}\log\left(\frac{n}{\sigma^2}\right) - \frac{\sqrt{n}\underline{\theta}_\star}{20\sigma}\right) \leq e^{-C \sqrt{n}\underline{\theta}_\star}.\end{equation}
We handle $1 - q_j(\check\vartheta^{(r-1)},\theta)$ similarly: since $\rho_0 \geq \rho_1$,
\[
1 - q_j(\check\vartheta^{(r-1)},\check\theta)\leq  \exp\left(\mathsf{a} + \frac{n\check\theta_j^2}{2\sigma^2}-\frac{\check\theta_j}{\sigma^2}\pscal{X_j}{y- X\check\theta_{(\check\vartheta^{(r-1)})^{(j,0)}}}\right).\]
The inequality (\ref{grad:ineq}) remains valid when applied to $\check\theta$ and $(\check\vartheta^{(r-1)})^{(j,0)}$ (but with $\min(J,s_1)$ replaced by $J$), and yields
\[\left|\pscal{X_j}{y- X\check\theta_{(\check\vartheta^{(r-1)})^{(j,0)}}} - n\theta_{\star j}\right| \leq C\left(s_\star +J\right)\sqrt{n\log(p)},
\]
leading to
\begin{multline*}
-\frac{\check\theta_j}{\sigma^2}\pscal{X_j}{y- X\check\theta_{(\check\vartheta^{(r-1)})^{(j,0)}}} \\
\leq  -\frac{n}{\sigma^2}\theta_{\star j}^2 + \frac{n\theta_{\star j}}{\sigma^2}|\check\theta_j - \theta_{\star j}|+  C\left(s_\star +J\right)\sqrt{n\log(p)}\\
\leq  \frac{n}{\sigma^2}\theta_{\star j}^2 +  C\left(s_\star +J\right)\sqrt{n\log(p)},
\end{multline*}
where we use Lemma \ref{lem:model:grad} to derive the bound $|\check\theta_j - \theta_{\star j}|\leq \left(|\check\theta_j - \hat\theta_{j}| + |\hat\theta_j - \theta_{\star j}|\right) \leq C\sqrt{n/\log(p)}$. The same bound implies that
\[
\frac{n\check\theta_j^2}{2\sigma^2} =  \frac{n\theta_{\star j}^2}{2\sigma^2}  +\frac{n(\check\theta_j^2 - \theta_{\star j}^2)}{2\sigma^2}\leq  \frac{n\theta_{\star j}^2}{2\sigma^2} + C\sqrt{n\log(p)}.\]
We conclude that
\begin{multline*}
\frac{n\check\theta_j^2}{2\sigma^2} -\frac{\check\theta_j}{\sigma^2}\pscal{X_j}{y- X\check\theta_{(\check\vartheta^{(r-1)})^{(j,0)}}}\\
 \leq - \frac{n\theta_{\star j}^2}{2\sigma^2} + C\left(s_\star +J\right)\sqrt{n\log(p)}  \leq -\frac{n\underline{\theta_\star}^2}{4\sigma^2},
 \end{multline*}
under the sample size condition (\ref{ss:cond}).
Hence
\begin{equation}\label{case1:eq3}
1 - q_j(\check\vartheta^{(r-1)},\check\theta) \leq  \exp\left(\mathsf{a}-\frac{n\underline{\theta_\star}^2}{4\sigma^2}\right)  \leq e^{-C n\underline{\theta}_\star^2}.\end{equation}
 Since the Bernoulli random variables $d_r^{(k)}$ and $\check d_r^{(k)}$ are maximally coupled, (\ref{case1:eq2}) and (\ref{case1:eq3}) imply that for $(\delta^{(k)},\check\delta^{(k)})\in\mathbb{T}$, and $\delta^{(k)}_j\neq \check\delta_k^{(k)}$, 
\begin{equation}\label{case1:eq4}
\int\mathcal{I}_{r,j}^{(k)}(\theta,\check\theta)\,\Pi(\rmd\theta\vert \delta^{(k)},\D)\,\Pi(\rmd\check\theta\vert\check\delta^{(k)},\D) \leq e^{-C \sqrt{n}\underline{\theta}_\star} + \frac{7}{10}.
\end{equation}

\paragraph{\underline{\texttt{Case 2}: $\delta_j^{(k)} = \check\delta_j^{(k)}$}}\;\;  Since $\check\delta^{(k)}\in\Cset_{s_1}$, we must then have $\delta_j^{(k)}= \check\delta_j^{(k)}=1$. Here we define the set $\mathbb{S}$ as
\[\mathbb{S}\eqdef\left\{(\theta,\check\theta)\in\rset^p\times \rset^p:\; \theta\in\cB_{\delta^{(k)}}, \check\theta\in\cB_{\check\delta^{(k)}}\right\}.\]
It follows from (\ref{proof:asynch:bound:B})  that for $(\delta^{(k)},\check\delta^{(k)})\in \mathbb{T}$,
\begin{equation}\label{case2:eq1}
\PP\left((\theta^{(k)},\check\theta^{(k)})\notin\mathbb{S}\vert \delta^{(k)}, \check\delta^{(k)}\right)\leq \frac{4}{p^{1-M/2}}.\end{equation}
For $(\delta^{(k)},\check\delta^{(k)})\in\mathbb{T}$, and $(\theta,\check\theta)\in\mathbb{S}$, the calculations on $1 - q_j(\check\vartheta^{(r-1)},\check\theta)$ remain valid, and we have
\[1 - q_j(\check\vartheta^{(r-1)},\check\theta) \leq e^{-C n\underline{\theta}_\star^2}.\]
For $\delta^{(k)}_j=1$, and $\theta\in\cB_{\delta^{(k)}}$, it follows from (\ref{grad:ineq}) that
\[\theta_j\pscal{X_j}{y- X\theta_{\vartheta}} \geq   \frac{n\theta_{\star j}^2}{\sigma^2} - C\left(s_\star +\min(J,s_1)\right)\sqrt{n\log(p)} \geq  \frac{n\underline{\theta}_\star^2}{2\sigma^2},\]
under the sample size condition (\ref{ss:cond}). We deduce  that
\begin{multline*}
1 - \tilde q_j(\vartheta,\theta) \leq \exp\left(\mathsf{a}  -\frac{\theta_j}{\sigma^2}\pscal{X_j}{y- X\theta_{\vartheta}} \right) \leq \exp\left(\mathsf{a} - \frac{n \underline{\theta}_\star^2}{2\sigma^2}\right) \leq e^{-Cn\underline{\theta}_\star^2}.\end{multline*}
The last two majorations on $1 - q_j(\check\vartheta^{(r-1)},\check\theta)$ and $1 - \tilde q_j(\vartheta,\theta)$, and (\ref{case2:eq1}) implies that for $\delta_j^{(k)} = \check\delta_j^{(k)}$, 
\begin{equation}\label{case2:eq2}
\int\mathcal{I}_{r,j}^{(k)}(\theta,\check\theta)\,\Pi(\rmd\theta\vert \delta^{(k)},\D)\,\Pi(\rmd\check\theta\vert\check\delta^{(k)},\D) \leq  e^{-Cn\underline{\theta}_\star^2} + \frac{4}{p^{1-M/2}}.
\end{equation}
The claim (\ref{big:claim}) follows from (\ref{case1:eq2}) and (\ref{case2:eq2}) together.
\begin{flushright}
\vspace{-0.2cm}
$\square$
\end{flushright}

\subsection{Technical lemmas}

\begin{lemma}\label{lem:moment}
Assume H\ref{H:post:contr}, and let $\e_0$ as in (\ref{bound:e}). Let $\{\delta^{(k)},\;k\geq 0\}$ be the $\delta$-marginal chain generated by Algorithm \ref{algo:2} for the linear regression posterior. There exists a constant $C>0$ that depends only on $c$ (in the definition of $\e_0$) and the constants in H\ref{H:post:contr} such that for 
\begin{equation}\label{cond:lem:mom}
\mathsf{u} \geq C(1 + s_\star^2),\;\;\mbox{ and }\;\; n\geq \left(\|\delta^{(0)}\|_0+s_\star+2\log(p)\right)^2\log(p),
\end{equation} 
it holds
\[\textbf{1}_{\e_0}(Y)\PP\left(\|\delta^{(k)}\|_0 > \|\delta^{(0)}\|_0+s_\star+2\log(p)\right) \leq  \left(\frac{p-s_\star}{2p}\right)^{s-\|\delta^{(0)}\|_0-s_\star} \leq  \frac{1}{p}.\]
\end{lemma}
\begin{proof} Fix $Y\in\e_0$. Set $s_1 \eqdef  \|\delta^{(0)}\|_0+s_\star+2\log(p)$. Referring to the coupling construction at the beginning of the proof,  the event $\{\|\delta^{(k)}\|_0>s_1\}$ means that we can find at least $s_1 -\|\delta^{(0)}\|_0 -s_\star$ terms among $\{(\delta^{(t)},  \mathsf{J}_r^{(t)}, d_r^{(t)}),\;1\leq t\leq k-1, 1\leq r\leq J\}$ where $\|\delta^{(t)}\|_0\leq s_1$, $\mathsf{J}_r^{(t)}\in \{j:\;\delta_{\star j}=0,\;\mbox{ and }\; \delta_j^{(t)}=0\}$,  and $d_r^{(t)} = 1$.
\[\PP\left(\mathsf{J}_r^{(t)}\in \{j:\;\delta_{\star j}=0,\;\mbox{ and }\; \delta_j^{(t)}=0\}\vert \delta^{(t)}\right) \leq  1 - \frac{s_\star}{p}.\]
We show next that on the event $\|\delta^{(t)}\|_0\leq s_1$, and $\mathsf{J}_r^{(t)}\in \{j:\;\delta_{\star j}=0,\;\mbox{ and }\; \delta_j^{(t)}=0\}$, 
\begin{equation}\label{claim:lem:moment}
\PP\left(d_r^{(t)} = 1 \vert \delta^{(t)}, \mathsf{J}_r^{(t)}=j\right) \leq \frac{1}{2},\end{equation}
to conclude that
\[\PP\left(\|\delta^{(k)}\|_0 > s_1\right) \leq  \left(\frac{p-s_\star}{2p}\right)^{s_1-\|\delta^{(0)}\|_0-s_\star}\leq \exp\left(-(s_1-\|\delta^{(0)}\|_0-s_\star)\log(2)\right)\leq \frac{1}{p},\]
which would end the proof. In order to prove (\ref{claim:lem:moment}), for some absolute constant $m>0$, let 
\[\mathbb{S}\eqdef\left\{\theta\in\rset^p:\;\theta\in\cB_{\delta^{(t)}},\;\;\mbox{ and }\;\; |\theta_j|\leq \sqrt{\frac{m}{\rho_0}} \right\}.\]
As seen in (\ref{proof:asynch:bound:B}), we can choose $m$ such that  $\Pi(\theta\notin \mathbb{S}\vert \delta^{(t)},\D)\leq \frac{1}{4}$. Fix $j$ such that $\delta_{\star j}=0$ and $\delta^{(t)}_j=0$. Recalling  the expression of $\tilde q_j$ in (\ref{cond:dist:eq:2}), it follows then that 
\begin{multline*}
\PP\left(d_r^{(t)} = 1 \vert \mathsf{J}_r^{(t)} =j,\;\delta^{(t)}\right) \leq  \frac{1}{4} + \\
\int_{\mathbb{S}} \exp\left(-\mathsf{a} +\frac{\theta_j^2}{2}(\rho_0-\rho_1) +\frac{\theta_j}{\sigma^2}\pscal{X_j}{y-X\theta_{\vartheta}} + \frac{\theta_j^2}{2\sigma^4}\pscal{X_j}{y-X\theta_{\vartheta}}^2\right) \Pi(\theta\vert \delta^{(t)},\D)\rmd\theta.\end{multline*}
Then we write
\begin{multline*}
\pscal{X_j}{y-X\theta_{\vartheta}} = \pscal{X_j}{y-X\theta_\star} + \pscal{X_j}{X\theta_\star - X\theta_\vartheta} \\
= \pscal{X_j}{y-X\theta_\star} + \sum_{r\neq j} \pscal{X_j}{X_r}(\theta_{\star r} - \theta_r\vartheta_r).
\end{multline*}
The last summation does not include $j$ because $\theta_{\star j}=0$, and $\mathsf{J}_r^{(t)}=j$, which implies that $\vartheta_j=0$.  Since $\theta\in \mathbb{S}$, we see that $|\theta_{\star r} - \theta_r\vartheta_r|\leq C$ for all $r$, for some constant $C$. If $\delta_{\star r}=0$, then $|\theta_r| \leq C\sqrt{\log(p)/n}$. It follows that for $Y\in\e_0$,
 \begin{multline*}
\left| \theta_j\pscal{X_j}{y- X\theta_{\vartheta}}\right| \leq |\theta_j| \left(\sqrt{cn\log(p)} + Cs_\star\sqrt{c_0n\log(p)} + Cs_1\sqrt{\frac{\log(p)}{n}}\sqrt{c_0n\log(p)}\right)\\
\leq C|\theta_j| \left(s_\star\sqrt{n\log(p)}  + s_1\log(p)\right)\leq Cs_\star\sqrt{\log(p)},\end{multline*}
under the sample size condition $n\geq s_1^2\log(p)$.
Hence taking $\mathsf{u}>C(1+s_\star)^2$ large enough, it follows that
\begin{multline*}
\PP\left(d_r^{(t)} = 1 \vert \mathsf{J}_r^{(t)} =j,\;\delta^{(t)}\right) \leq  \frac{1}{4} + \\
\int_{\mathbb{S}} \exp\left(-\mathsf{u}\log(p)-\frac{1}{2}\log\left(\frac{n}{\sigma^2}\right) + C(1+ s_\star)^2\log(p)\right)\Pi(\theta\vert \delta^{(t)},\D)\rmd\theta + \frac{1}{4} \leq \frac{1}{4} + \frac{1}{4} \leq \frac{1}{2}.\end{multline*}
\end{proof}

\begin{lemma}\label{lem:model:grad}
Assume H\ref{H:post:contr}, and let $\e_0$ as in (\ref{bound:e}). Fix $0< s_1\leq p$. Then we can find constants $C,C'$ that depends only on $\sigma^2,\|\theta_\star\|_\infty,c_0$, and $c$ (in the definition of $\e_0$) such that for $n\geq C s_1^2\log(p)$, the following holds. For all $\delta\in\Delta$ such that $\|\delta\|_0\leq s_1$, 
\begin{equation}\label{bound:freq:est}
\|\hat\theta_{\delta} - [\theta_\star]_{\delta}\|_\infty \leq C'\left( 1 + m(\delta)\right)\sqrt{\frac{\log(p)}{n}},
\end{equation}
where $ m(\delta) \eqdef  |\{k:\;\delta_{\star k}=1,\;\delta_k=0\}|$, and $\hat\theta_\delta$ as in (\ref{cond:dist:theta:lm}).
\end{lemma}
\begin{proof}
The proof follows (\cite{lounici:08}~Theorem 1). Fix $Y\in\e_0$. The first order optimality condition of $\hat\theta_\delta$ is given by $-\rho_1\hat\theta_\delta + X_\delta'(Y - X_\delta\hat\theta_\delta)/\sigma^2=0$, which can be rewritten as 
\[\left(\rho_1 I_{\|\delta\|_0} + \frac{1}{\sigma^2}X_\delta'X_\delta\right)([\theta_\star]_{\delta} - \hat\theta_\delta) -\rho_1[\theta_\star]_\delta +\frac{1}{\sigma^2}X_\delta'(X\theta_\star - X_\delta[\theta_\star]_{\delta}) + \frac{1}{\sigma^2}X_\delta'(Y-X\theta_\star) = 0.\]
We deduce that
\begin{multline*}
\left\|\left(\rho_1 I_{\|\delta\|_0} + \frac{1}{\sigma^2}X_\delta'X_\delta\right)([\theta_\star]_{\delta} - \hat\theta_\delta)\right\|_\infty \leq \rho_1\|\theta_\star\|_\infty \\
+ \frac{1}{\sigma^2} \max_{k:\;\delta_k=1}\sum_{j:\;\delta_{\star j}=1,\delta_j=0}|\theta_{\star j}| |\pscal{X_j}{X_k}| + \frac{1}{\sigma} \sqrt{cn\log(p)}\\
\leq C(1 + m(\delta))\sqrt{n\log(p)}.
 \end{multline*}
 Using this conclusion and the restricted strong convexity in (\ref{rsc}), for $n\geq C s_1^2\log(p)$, we have
\begin{multline*}
\frac{n}{2} \|\hat\theta_\delta  - [\theta_\star]_{\delta}\|_2^2 \leq (\hat\theta_\delta  - [\theta_\star]_{\delta})' \left(\rho_1 I_{\|\delta\|_0} + \frac{1}{\sigma^2}X_\delta'X_\delta\right)(\hat\theta_\delta  - [\theta_\star]_{\delta}) \\
\leq C(1 + m(\delta))\sqrt{n\log(p)}\|\hat\theta_\delta  - [\theta_\star]_{\delta}\|_1\\
\leq Cs_1^{1/2}(1 + m(\delta))\sqrt{n\log(p)}\|\hat\theta_\delta  - [\theta_\star]_{\delta}\|_2,
\end{multline*}
which implies that
\[\|\hat\theta_\delta  - [\theta_\star]_{\delta}\|_2 \leq C(1 + m(\delta)) \sqrt{\frac{s_1\log(p)}{n}}.\]
On the other hand for $j$ such that $\delta_j=1$, 
\begin{multline*}
\left(\left(\rho_1 I_{\|\delta\|_0} + \frac{1}{\sigma^2}X_\delta'X_\delta\right)([\theta_\star]_{\delta} - \hat\theta_\delta)\right)_j = (\rho_1 +\frac{n}{\sigma^2})\left(\hat\theta_\delta-[\theta_\star]_\delta\right)_j \\
+ \frac{1}{\sigma^2}\sum_{k\neq j:\;\delta_k=1}\pscal{X_k}{X_j}\left(\hat\theta_\delta-[\theta_\star]_\delta\right)_k,
\end{multline*}
which we use to deduce that
\begin{multline*}
\|\hat\theta_\delta-[\theta_\star]_\delta\|_\infty \leq \frac{\sigma^2}{n}  \left\|\left(\rho_1 I_{\|\delta\|_0} + \frac{1}{\sigma^2}X_\delta'X_\delta\right)([\theta_\star]_{\delta} - \hat\theta_\delta)\right\|_\infty + \frac{1}{n}\sqrt{c_0n\log(p)}\|\hat\theta_\delta-[\theta_\star]_\delta\|_1\\
\leq C(1+ m(\delta))\sqrt{\frac{\log(p)}{n}} + s_1^{1/2}\sqrt{\frac{c_0\log(p)}{n}} \|\hat\theta_\delta-[\theta_\star]_\delta\|_2\\
\leq C(1+ m(\delta))\sqrt{\frac{\log(p)}{n}}  + C(1+ m(\delta))\sqrt{\frac{\log(p)}{n}}  \sqrt{\frac{s_1^2\log(p)}{n}}\\
\leq C(1+ m(\delta))\sqrt{\frac{\log(p)}{n}},
\end{multline*}
under the stated sample size condition.
\end{proof}

We show in the next result  that the posterior distribution puts most of its probability mass on models that contain the true model $\delta_\star$.

\begin{lemma}\label{post:contr}
Assume H\ref{H:post:contr}, and let $\e_{s_\star}$ be as in (\ref{bound:e}). Then we can find constants $C_1, C_2$ that depends only on $\sigma^2,\|\theta_\star\|_\infty,c_0$, $c_1$ $c_2$ and $c$ (in the definition of $\e_{s_\star}$) such that for $n\geq C_1 \underline{\theta}_\star^{-2} (1 + s_\star^3)\log(p)$, it holds
\[\PE_\star\left[\textbf{1}_{\e_{s_\star}}(Y) \Pi(\Cset\vert \D)\right]\geq 1-\frac{3}{p},\]
where 
\[\Cset \eqdef \left\{\delta\in\Delta:\;\delta\supseteq \delta_\star,\;\mbox{ and } \; \|\delta\|_0\leq C_2 (1+s_\star)\right\}.\]
\end{lemma}
\begin{proof}
By Lemma \ref{lem:post:spar}, there exist positive constant $C_1,C_2$ that depends only on $c_0,c_2$, and $c$  such that for $n\geq C_1 s_\star^2\log(p)$,
\[\PE_\star\left[\textbf{1}_{\e_0}(Y)\Pi\left(\|\delta\|_0> C_2(1 + s_\star)\; \vert\; \D\right)\right]\leq \frac{2}{p}.\]
We set 
\[s \eqdef C_2(1 + s_\star),\]
and $\A\eqdef \{\delta:\;\delta\nsupseteq\delta_\star,\;\|\delta\|_0\leq s\}$, so that 
\[\Delta = \Cset_{s} \cup\A \cup\left\{\delta\in\Delta:\;\|\delta\|_0>s\right\}.\]  
Therefore, 
\[\PE_\star\left[\textbf{1}_{\e_{s_\star}}(Y)\Pi(\Cset_{s}\vert \D)\right] \geq 1 - \frac{2}{p} - \PE_\star\left[\textbf{1}_{\e_{s_\star}}(Y)\Pi\left(\A\vert \D\right)\right].\]
To finish the proof we will establish that for $Y \in\e_{s_\star}$, $\Pi\left(\A\vert \D\right)\leq \frac{1}{p}$. To that end,  let $\mathcal{P}\eqdef\{\delta^{(0)}\in\Delta:\;\delta^{(0)}\subseteq\delta_\star,\;\delta^{(0)}\neq \delta_\star\}$, and for each $\delta^{(0)}\in\mathcal{P}$, we set 
\[\A(\delta^{(0)})\eqdef \left\{\delta\in\Delta:\; \delta\supseteq \delta^{(0)},\;\|\delta\|_0\leq s_1,\mbox{ and } \delta_{\star j}=0 \mbox{ whenever } (\delta_j=1,\mbox{ and } \delta_j^{(0)}=0)\right\}.\] 
We then write
\begin{equation}\label{eq:PiA}
\Pi(\A\vert \D) = \Pi(\delta_\star\vert \D)\sum_{\delta^{(0)}\in\mathcal{P}} \frac{\Pi(\delta^{(0)}\vert \D)}{\Pi(\delta_\star\vert \D)} \sum_{\delta\in  \A(\delta^{(0)})} \frac{\Pi(\delta\vert \D)}{\Pi(\delta^{(0)}\vert \D)}.\end{equation}
For any subset $\delta,\vartheta\in\Delta$, we have
\[
\frac{\Pi(\delta\vert\D)}{\Pi(\vartheta\vert \D)}  =  \left(\frac{1}{p^{\textsf{u}}}\sqrt{\frac{\rho_1}{2\pi}}\right)^{\|\delta\|_0-\|\vartheta\|_0}\frac{\int_{\rset^{\|\delta\|_0}}\exp\left(-\frac{\rho_1}{2}\|u\|_2^2 -\frac{1}{2\sigma^2}\|Y-X_\delta u\|_2^2\right)\rmd u}{\int_{\rset^{\|\vartheta\|_0}}\exp\left(-\frac{\rho_1}{2}\|u\|_2^2 -\frac{1}{2\sigma^2}\|Y-X_\vartheta u\|_2^2\right)\rmd u}.\]
We calculate that for any $\delta\in\Delta$,
\[\int_{\rset^{\|\delta\|_0}}\exp\left(-\frac{\rho_1}{2}\|u\|_2^2 -\frac{1}{2\sigma^2}\|Y-X_\delta u\|_2^2\right)\rmd u = \left(\frac{2\pi}{\rho_1}\right)^{\|\delta\|_0/2}\frac{e^{-\frac{1}{2\sigma^2}Y'\left(I_{n} + \frac{1}{\rho_1\sigma^2}X_\delta X_\delta'\right)^{-1}Y}}{\sqrt{\det\left(I_{n} + \frac{1}{\rho_1\sigma^2}X_\delta X_\delta'\right)}}.\]
And we deduce, using $\rho_1=1$, that
\begin{equation}\label{eq:ratio:Pi:1}
\frac{\Pi(\delta\vert\D)}{\Pi(\vartheta\vert \D)}  =  \left(\frac{1}{p^{\textsf{u}}}\right)^{\|\delta\|_0-\|\vartheta\|_0}\frac{}{}e^{\frac{1}{2\sigma^2}\left(Y'L_{\vartheta}^{-1}Y - Y'L_{\delta}^{-1}Y \right)}\sqrt{\frac{\det\left(L_\vartheta\right)}{\det\left(L_\delta\right)}},\end{equation}
where
\[L_\delta\eqdef I_n + \frac{1}{\sigma^2} X_\delta X_\delta'.\]
Suppose that $\vartheta\supseteq\delta$, and $\|\vartheta\|_0\leq s_1$. In that case
\[L_\vartheta = L_\delta + \frac{1}{\sigma^2}X_{\vartheta-\delta}X_{\vartheta-\delta}'.\]
and by the determinant lemma ($\det(A + UV') = \det(A)\det(I_m + V'A^{-1}U)$ valid for any invertible matrix $A\in\rset^{n\times n}$, and $U,V\in\rset^{n\times m}$), and using the  lower bound on  the smallest eigenvalue of $L_\delta$ resulting from (\ref{rsc}), we have
\[1 \leq \frac{\det\left(L_\vartheta\right)}{\det\left(L_\delta\right)} = \det\left(I_{\vartheta-\delta} +\frac{1}{\sigma^2} X_{\vartheta-\delta}' L_\delta^{-1} X_{\vartheta-\delta}\right) \leq\left(1 + 2\|\vartheta-\delta\|_0\right)^{\|\vartheta-\delta\|_0}.\]
We use this to deduce from (\ref{eq:ratio:Pi:1}) that when $\vartheta\supseteq\delta$, and $\|\vartheta\|_0\leq s$, it holds
\begin{multline}\label{eq:ratio:Pi:2}
p^{\textsf{u}\|\vartheta-\delta\|_0} e^{\frac{1}{2\sigma^2}\left(Y'L_{\vartheta}^{-1}Y - Y'L_{\delta}^{-1}Y \right)} \leq \frac{\Pi(\delta\vert\D)}{\Pi(\vartheta\vert \D)} \\
\leq p^{\textsf{u}\|\vartheta-\delta\|_0} \left(1+ 2s\right)^{\|\vartheta-\delta\|_0} e^{\frac{1}{2\sigma^2}\left(Y'L_{\vartheta}^{-1}Y - Y'L_{\delta}^{-1}Y \right)}.\end{multline}
By the Woodbury formula which states that any set of matrices $U,V, A,C$ with matching dimensions, $(A + UCV)^{-1} = A^{-1} - A^{-1}U(C^{-1} + VA^{-1}U)^{-1}VA^{-1}$, we have
\[Y'L_\vartheta^{-1}Y - Y'L_\delta^{-1}Y =-\frac{1}{\sigma^2}y'L_\delta^{-1}X_{\vartheta-\delta}\left(I_{\|\vartheta-\delta\|_0} + \frac{1}{\sigma^2}X_{\vartheta-\delta}'L_\delta^{-1}X_{\vartheta-\delta}\right)^{-1}X_{\vartheta-\delta}L_\delta^{-1}y.\]
It follows from Equation (\ref{control:L}) of Lemma \ref{lem:Ls} that for any non-zero vector $u\in\rset^{\|\vartheta-\delta\|_0}$,
\[u'X_{\vartheta-\delta}'L_\delta^{-1}X_{\vartheta-\delta} u \geq \frac{n}{2}\| u\|_2^2 - C_0\sqrt{n\log(p)}\|u\|_1^2\geq \frac{n}{4}\|u\|_2^2,\]
for some absolute constant $C_0$, provided that $n\geq 4C_0^2s^2\log(p)$. We deduce that for $\delta\subseteq\vartheta$, $\|\vartheta\|_0\leq s$, it holds
\begin{equation}\label{eq:diff:L}
\frac{\|Y'L_\delta^{-1}X_{\vartheta-\delta}\|_2^2}{\sigma^2(1 + \|\vartheta-\delta\|_0n)} \leq Y'L_\delta^{-1}Y - Y'L_\vartheta^{-1}Y \leq \frac{4\|Y'L_\delta^{-1}X_{\vartheta-\delta}\|_2^2}{\sigma^2 n}.
\end{equation}
We put  (\ref{eq:diff:L}) and (\ref{eq:ratio:Pi:2}) to write the second summation of  (\ref{eq:PiA}) as
\begin{multline*}
\sum_{\delta\in  \A(\delta^{(0)})} \frac{\Pi(\delta\vert \D)}{\Pi(\delta^{(0)}\vert \D)} = \sum_{k=0}^{s-\|\delta^{(0)}\|_0} \;\;\sum_{\delta\in  \A(\delta^{(0)}):\;\|\delta\|_0=\|\delta^{(0)}\|_0+k}\;\;\frac{\Pi(\delta\vert \D)}{\Pi(\delta^{(0)}\vert \D)} \\
\leq \sum_{k=0}^{s-\|\delta^{(0)}\|_0} \;\;\sum_{\delta\in  \A(\delta^{(0)}):\;\|\delta\|_0=\|\delta^{(0)}\|_0+k}\;\;\left(\frac{1}{p^{\mathsf{u}}}\right)^k \exp\left(\frac{2\left\|Y'L_{\delta^{(0)}}^{-1}X_{\delta-\delta^{(0)}}\right\|_2^2}{\sigma^4 n}\right).
\end{multline*}
We can write $Y = \sigma V + \sum_{k:\;\delta_{\star k}=1}\theta_{\star k}X_k$, where $V = (Y-X\theta_\star)/\sigma$. Fix a component $i$ such that $(\delta-\delta^{(0)})_i=1$. Note that  we have $\delta^{(0)}_i=0$, and $\delta_{\star i}=0$. We can then write
\[Y'L_{\delta^{(0)}}^{-1}X_i = \sigma V'L_{\delta^{(0)}}^{-1}X_i + \sum_{k:\;\delta_{\star k}=1}\theta_{\star k}X_k'L_{\delta^{(0)}}^{-1}X_i.\]
Therefore, by (\ref{control:L}) from Lemma \ref{lem:Ls}, we have
\[|Y'L_{\delta^{(0)}}^{-1}X_i| \leq C_1\sqrt{(1+s_\star)n\log(p)} + C_1(s_\star - \|\delta^{(0)}\|_0)\sqrt{n\log(p)}.\]
It follows that
\[\frac{\left\|Y'L_{\delta^{(0)}}^{-1}X_{\delta-\delta^{(0)}}\right\|_2^2}{\sigma^4 n}   \leq C_1k(1+s_\star^2)\log(p),\]
for some constant $C_1$. Therefore,
\[\sum_{\delta\in  \A(\delta^{(0)})} \frac{\Pi(\delta\vert \D)}{\Pi(\delta^{(0)}\vert \D)}\leq p^{C_1s(1+s_\star^2)}\;\;\sum_{k=0}^{s-\|\delta^{(0)}\|_0} \;{p-\|\delta^{(0)}\|_0\choose k} \left(\frac{1}{p^{\mathsf{u}}}\right)^k \leq 2 p^{C_1s(1+s_\star^2)},\]
by choosing $\mathsf{u}>2$, assuming $p\geq 2$. The last display, with  (\ref{eq:PiA}) and (\ref{eq:ratio:Pi:2}) yield
\begin{multline*}
\Pi(\A\vert \D) \leq 2 p^{C_1s(1+s_\star^2)} \sum_{\delta^{(0)}\in\mathcal{P}} \frac{\Pi(\delta^{(0)}\vert \D)}{\Pi(\delta_\star\vert \D)} \\
\leq 2 p^{C_1s(1+s_\star^2)}   \sum_{k=0}^{s_\star-1}\;\sum_{\delta^{(0)}\in\mathcal{P}:\;\|\delta^{(0)}\|_0 = s_\star -k}\;\; p^{k\mathsf{u}} (1+2s)^k\exp\left(-\frac{\left\|Y'L_{\delta^{(0)}}^{-1}X_{\delta_\star-\delta^{(0)}}\right\|_2^2}{2\sigma^4 (1 + kn)}\right).
\end{multline*}
As above, given $i$ such that $\delta_{\star i}=1$,  we write
\[Y'L_{\delta^{(0)}}^{-1}X_i = \sigma V'L_{\delta^{(0)}}^{-1}X_i +\theta_{\star i}X_i'L_{\delta^{(0)}}^{-1}X_i + \sum_{k\neq i:\;\delta_{\star k}=1}\theta_{\star k}X_k'L_{\delta^{(0)}}^{-1}X_i,\]
and using Lemma \ref{lem:Ls}, we deduce that 
\[|Y'L_{\delta^{(0)}}^{-1}X_i| \geq \frac{n|\theta_{\star i}|}{2}   - C_1\sqrt{(1+s_\star)n\log(p)} \geq \frac{n|\underline{\theta}_\star|}{4},\]
under the sample size condition $n\geq C_2 (1+s_\star)\log(p)/\underline{\theta}_\star^2$. Hence
\[\frac{\left\|Y'L_{\delta^{(0)}}^{-1}X_{\delta_\star-\delta^{(0)}}\right\|_2^2}{2\sigma^4 (1 + kn)}\geq \frac{n\underline{\theta}_\star^2}{64\sigma^4},\]
so that
\begin{multline*}
\Pi(\A\vert \D) \leq 2 p^{C_1s(1+s_\star^2)}  e^{-\frac{n\underline{\theta}_\star^2}{64\sigma^4}} \sum_{k=0}^{s_\star-1} {s_\star\choose k} (1+2s)^k\\
\leq 2 p^{C_1s(1+s_\star^2)}  e^{-\frac{n\underline{\theta}_\star^2}{64\sigma^4}}  \sum_{k=0}^{s_\star-1}  \left(C_1(1+s_\star)\right)^k\\
\leq 2 p^{C_1s(1+s_\star^2)}  e^{-\frac{n\underline{\theta}_\star^2}{64\sigma^4}}  \left(C_1(1+s_\star)\right)^{s_\star}\\
\leq2 \exp\left(-\frac{n\underline{\theta}_\star^2}{64\sigma^4} + C_1(1+s_\star^3)\log(p)\right)\\
 \leq \exp\left(-\frac{n\underline{\theta}_\star^2}{128\sigma^4}\right)\leq \frac{1}{p},
\end{multline*}
for $n\geq C_2\underline{\theta}_\star^{-2}(1 + s_\star^3)\log(p)$, for some constant $C_2$. This ends the proof.

\end{proof}

\begin{lemma}\label{lem:Ls}
Assume H\ref{H:post:contr}, and fix $s>0$.  There exist constants $C_1,C_2$ that depends only on $\sigma,\underline{\kappa},c_0$ and $\|\theta_\star\|_\infty$ such that for $n \geq C_1 s^2\log(p)$, the following holds. For all $\delta\in\Delta_{s}$, and for all pair $j \neq k$, such that $\delta_j=0$,  it holds
\begin{equation}\label{control:L}
|X_j'L_\delta^{-1}X_k| \leq  C_2\left(1 + \frac{\|\delta\|_0^{1/2}}{n}\textbf{1}_{\{\delta_k=1\}}\right)\sqrt{n\log(p)} ,\;\;\;\mbox{ and }\;\;  X_j'L_\delta^{-1}X_j \geq \frac{n}{2}.
\end{equation}
\end{lemma}
\begin{proof}
Applying the Woodbury identity to $L_\delta$, we have
\begin{equation}\label{eq1:lem:Ls}
L_\delta^{-1} = I_n - \frac{1}{\sigma^2}X_\delta\left(I_{\|\delta\|_0} + \frac{1}{\sigma^2}X_\delta'X_\delta\right)^{-1} X_\delta'.\end{equation}
It follows that 
\[X_j'L_\delta^{-1}X_k = \pscal{X_j}{X_k}- \frac{1}{\sigma^2}X_j'X_\delta\left(I_{\|\delta\|_0} + \frac{1}{\sigma^2}X_\delta'X_\delta\right)^{-1}X_\delta' X_k.\]
Under the sample size condition, by (\ref{rsc}), we have
\[\left|\frac{1}{\sigma^2}X_j'X_\delta\left(I_{\|\delta\|_0} + \frac{1}{\sigma^2}X_\delta'X_\delta\right)^{-1}X_\delta' X_k\right| \leq \frac{2\|X_\delta'X_j\|_2\|X_\delta'X_k\|_2}{\sigma^2n}.\]
By assumption $\delta_j=0$. If $\delta_k=0$, then 
\[\|X_\delta'X_j\|_2\|X_\delta'X_k\|_2\leq \sqrt{c_0\|\delta\|_0n\log(p)} \sqrt{c_0\|\delta\|_0 n\log(p)},\]
and we deduce that
\[|X_j'L_\delta^{-1}X_k| \leq \sqrt{c_0n\log(p)} + \frac{2c_0\|\delta\|_0\log(p)}{\sigma^2}\leq C\sqrt{n\log(p)},\]
provided that $n\geq s^2\log(p)$. Suppose now that $\delta_k=1$. Note that starting from (\ref{eq1:lem:Ls}) we can also write 
\[X_\delta'L_\delta^{-1}X_j = \left(I_{\|\delta\|_0} + \frac{1}{\sigma^2}X_\delta'X_\delta\right)^{-1}X_\delta'X_j.\]
This implies that if $\delta_k=1$, then
\[|X_j'L_\delta^{-1}X_k| \leq \frac{2\|X_\delta'X_j\|_2}{n} \leq \frac{2\sqrt{c_0\|\delta\|_0n\log(p)}}{n}\leq C_2\sqrt{\frac{\|\delta\|_0\log(p)}{n}},\]
which establishes the first part of (\ref{control:L}). When $j=k$, we get
\[X_j'L_\delta^{-1}X_j \geq n - \frac{2c_0\|\delta\|_0\log(p)}{\sigma^2} \geq \frac{n}{2},\]
under the sample size condition $n\geq 2c_0s\log(p)/\sigma^2$.

\end{proof}

\begin{lemma}\label{lem:post:spar}
Assume H\ref{H:post:contr},  and let $\e_0$ as in (\ref{bound:e}).
There exist positive constant $C_1,C_2$ that depends only on $c_0,c_2$, and $c$ (in the definition of $\e_0$)  such that for $n\geq C_1 s_\star^2\log(p)$,
\[\PE_\star\left[\textbf{1}_{\e_0}(Y)\Pi\left(\|\delta\|_0>(1 + C_2)s_\star\; \vert\; \D\right)\right]\leq \frac{2}{p}.\]
\end{lemma}
\begin{proof}
The lemma follows from Theorem 2.2 of \cite{AB:19}, applied with $\bar\rho =2 \sqrt{cn\log(p)}/\sigma$, and $\bar\kappa = s_\star n$.  The sub-Gaussian assumption in H\ref{H:post:contr}-(1) implies that Equation (2.1) of \cite{AB:19}  holds with $\mathsf{r}_0 = n/(2\sigma^2\bar\rho)$ under the sample size condition. Then  using the assumption in H\ref{H:post:contr} that $n/p$ and $\|\theta_\star\|_\infty/\log(p)$ remain bounded from above by $c_2$, we checked that Equation 2.2 of \cite{AB:19} is satisfies for some absolute constant $c_0$.
\end{proof}

\section{Description of the coupled chains for mixing time estimation}\label{append:coupled:chains}
We describe here the specific coupled Markov chain employed to estimate the mixing time plots presented in Section \ref{sec:lm}. We describe the method for Algorithm \ref{algo:1}. Algorithm \ref{algo:2} proceeds similarly. 

We start with a brief description of the method. Let $\{X^{(t)},\;t\geq 0\}$ be the Markov chain generated by Algorithm \ref{algo:1}, where $X^{(t)} = (\delta^{(t)},\theta^{(t)}) \in \mathsf{X}$. Let $K$ denote the transition kernel of the Markov chain $\{X^{(t)},\;t\geq 0\}$. The basic idea of the method is to construct a coupling $\check K$ of $K$ with itself: that is, a transition kernel on $\mathsf{X}\times \mathsf{X}$ such that $\check K((x,y), A\times \mathsf{X}) = K(x,A)$, $\check K((x,y),\mathsf{X}\times B) = K(y,B)$, for all $x,y\in\mathsf{X}$, and all measurable sets $A,B$. The coupling $\check K$ is constructed in such  a way that $\check K((x,x), \mathcal{D}) =1$, where $\mathcal{D} \eqdef\{(x,x):\;x\in\mathsf{X}\}$.  The method then proceeds as follows. Fix a lag $L\geq 1$. Draw $X^{(0)}\sim\Pi^{(0)}$, $Y^{(0)}\sim\Pi^{(0)}$ (where $\Pi^{(0)}$ is the initial distribution as given in the initialization step in Algorithm~\ref{algo:1}). Draw $X^{(L)}\vert (X^{(0)},Y^{(0)}) \sim K^L(X^{(0)},\cdot)$. Then for any $k\geq 1$, draw,
\[(X^{(L+k)},Y^{(k)})\vert \left\{(X^{(L+k-1)},Y^{(k-1)}),\ldots,(X^{(L)},Y^{(0)})\right\} \sim \check K\left((X^{(L+k-1)},Y^{(k-1)}),\cdot\right).\]
Setting 
\[\tau^{(L)} \eqdef\inf\left\{k>L:\; X^{(k)} = Y^{(k-L)}\right\},\]
it then holds under some ergodicity assumptions on $P$ (see \cite{biswas2019estimating}) that
\begin{equation}\label{eq:coupl}
\| \Pi^{(t)}-\Pi\|_{\textsf{tv}} \leq \PE\left[\max\left(0,\left\lceil\frac{\tau^{(L)} -L - t}{L}\right\rceil\right)\right],\end{equation}
where $\lceil x \rceil$ denote the smallest integer above $x$. The implication of (\ref{eq:coupl}) is that we can empirically upper bound the left hand side of (\ref{eq:coupl}) by simulating multiple copies of the joint chain as described above and then approximating the expectation on the right hand side of (\ref{eq:coupl}) by Monte Carlo. We refer the reader to \cite{biswas2019estimating} for more details on the construction of such coupled kernels. 

We modify  Algorithm~\ref{algo:1}  to  construct the coupled kernel $\check P$.
Let $(\delta^{(1,t)},\theta^{(1,t)})$ and let $(\delta^{(2,t)},\theta^{(2,t)})$ denote the states of the two chains at time $t$. At some iteration $ t \geq 1$, given $(\delta^{(1, L + t)},\theta^{(1, L+t)}) = (\delta^{(1)},\theta^{(1)})$ and $(\delta^{(2,t)},\theta^{(2,t)}) = (\delta^{(2)},\theta^{(2)})$, we now describe how to generate the next state of the coupled chain. 

In step 1, to update $\delta^{(1)}$ and $\delta^{(2)}$, we first make use of the same randomly drawn subset $\mathsf{J}$. For $i = 1,2$, drawing $\bar \delta^{(i)} \sim Q_{\theta}^{(\mathsf{J})} (\delta^{(i)}, \cdot)$ is equivalent to let $ \bar \delta^{(i)}_{-\mathsf{J}} = \delta^{(i)}_{-\mathsf{J}}$, and for any $j\in\mathsf{J}$, draw  $\bar \delta^{(i)}_j\sim \textbf{Ber}(q^{(i)}_{j})$ which we implement in the following way. We first draw a common uniform number $u_j\sim \textbf{Uniform}(0,1)$, then we obtain  $\bar \delta^{(i)}_j = \mathbf{1}\{q^{(i)}_{j} \leq u_j\}$ for $i = 1,2$.

In step 2, to update $\theta^{(1)}$ and $\theta^{(2)}$, for simplicity, we partition the indices $\{1, \ldots,p\}$ into four groups: $G_{ab}=\{j:\; \bar \delta^{(1)}_{j}=a, \bar \delta^{(2)}_{j}=b\}$ for $a, b = 0, 1$. 

To update the components of $\theta^{(1)}_{G_{00}}$ and $\theta^{(2)}_{G_{00}}$, for any $j \in G_{00}$ we first draw a common standard normal random variables $Z_j$, and then obtain $\bar\theta^{(i)}_j = \rho_0^{-\frac{1}{2}} Z_j $ for $i = 1,2$. 

To update the components of $\theta^{(1)}_{G_{01}}$ and $\theta^{(2)}_{G_{01}}$, 
Since in linear regression, $[\theta]_\delta \;\vert\; \delta \sim \textbf{N}(\hat{\theta}_\delta, \Sigma)$, where $\hat{\theta}_\delta$ is describled in (\ref{cond:dist:theta:lm})  and $\Sigma = \sigma^2(\sigma^2\rho_1 I_{\|\delta\|_0})^{-1}$,  we then have $\theta^{(1)} \;\vert\; \delta^{(1)} \sim \textbf{N}(\hat{\theta}^{(1)}, \Sigma^{(1)})$ and  $\theta^{(2)} \;\vert\; \delta^{(2)} \sim \textbf{N}(\hat{\theta}^{(2)}, \Sigma^{(2)})$, respectively.
Then with the property of gaussian random variables, we have $\theta^{(2)}_{G_{01}} \sim \textbf{N}(\hat{\theta}^{(2)}_{G_{01}}, \Sigma^{(2)}_{G_{01}})$, where $\hat{\theta}^{(2)}_{G_{01}}$ are the $G_{01}$ components of $\hat{\theta}^{(2)}$ and $\Sigma^{(2)}_{G_{01}}$ is the submatrix of $\Sigma^{(2)}$ with $G_{01}$ rows and columns. With $\theta^{(1)}_{G_{01}} \sim \textbf{N}(\textbf{0}, \rho_0^{-1}I_{\|\delta_{G_{01}}\|_0})$, we draw the maximal coupling of these two gaussian distributions to update $\bar{\theta}^{(1)}_{G_{01}}$ and $\bar{\theta}^{(2)}_{G_{01}}$. A similar updating procedure is used for the components of $\bar\theta^{(1)}_{G_{10}}$ and $\bar\theta^{(2)}_{G_{10}}$.

 For components of $\theta^{(1)}_{G_{11}}$ and $\theta^{(2)}_{G_{11}}$, since we have $\theta^{(1)}_{G_{11}} \sim \textbf{N}(\hat{\theta}^{(1)}_{G_{11}}, \Sigma^{(1)}_{G_{11}})$, where $\hat{\theta}^{(1)}_{G_{11}}$ are the $G_{11}$ components of $\hat{\theta}^{(1)}$ and $\Sigma^{(1)}_{G_{11}}$ is the submatrix of $\Sigma^{(1)}$ with $G_{11}$ rows and columns, and similarly $\theta^{(2)}_{G_{11}} \sim \textbf{N}(\hat{\theta}^{(2)}_{G_{11}}, \Sigma^{(2)}_{G_{11}})$, we could construct another maximal coupling to update $\bar\theta^{(1)}_{G_{11}}$ and $\bar\theta^{(2)}_{G_{11}}.$

\bibliographystyle{ims}
\bibliography{biblio_graph,biblio_mcmc,biblio_optim,biblio_new}

\end{document}